\documentclass[11pt,a4paper]{article}
\pdfoutput=1
\usepackage[margin=1.5cm,includefoot,footskip=30pt]{geometry}
\usepackage[ascii]{inputenc}
\usepackage[bookmarksnumbered,linktocpage,hypertexnames=false,colorlinks=true,linkcolor=blue,urlcolor=blue,citecolor=magenta,anchorcolor=green,breaklinks=true,pdfusetitle=true]{hyperref}
\usepackage{bbm,braket,microtype,mathrsfs,amsmath,amssymb,xcolor,amsthm,amsfonts,latexsym,mathtools,graphicx,enumitem,booktabs,bm,xspace,float,mathdots,caption,subcaption,ellipsis,mleftright,crossreftools,dsfont}
\usepackage[T1]{fontenc}
\usepackage{textcomp}
\usepackage[english]{babel}
\usepackage[capitalize,nameinlink]{cleveref}
\newtheorem{theorem}{Theorem}[section]
\newtheorem{lem}[theorem]{Lemma}\crefname{lem}{Lemma}{Lemmas}
\newtheorem*{lem*}{Lemma}
\newtheorem{fact}[theorem]{Fact}
\newtheorem{prop}[theorem]{Proposition}\crefname{prop}{Proposition}{Propositions}
\newtheorem{cor}[theorem]{Corollary}\crefname{cor}{Corollary}{Corollaries}
\theoremstyle{definition}
\newtheorem{defn}[theorem]{Definition}
\numberwithin{equation}{section}
\DeclareMathOperator{\CNOT}{CNOT}
\DeclareMathOperator{\SWAP}{SWAP}
\DeclareMathOperator{\tr}{tr}
\renewcommand{\psi}{\phi}
\newcommand{\Td}{\ensuremath{\mathsf{Td}}}
\newcommand{\Init}{\ensuremath{\mathsf{Init}}\xspace}
\newcommand{\Rec}{\ensuremath{\mathsf{Rec}}\xspace}
\newcommand{\SPO}{\ensuremath{\mathsf{SPO}}\xspace}

\newcommand{\OSPO}{O^{\SPO}}
\newcommand{\OinvSPO}{O^{\SPO,\mathrm{inv}}}

\newcommand{\OSPOz}{{O}^{\SPO,z}}
\newcommand{\OSPOx}{{O}^{\SPO,x}}
\newcommand{\OinvSPOz}{{O}^{\SPO,\mathrm{inv},z}}

\newcommand{\FTSPOsigmatau}{\ensuremath{\mathsf{TSPO}^{\sigma,\tau}}\xspace}
\newcommand{\FTSPOsigmaztauz}{\ensuremath{\mathsf{TSPO}^{\sigma_0,\tau_0}}\xspace}
\newcommand{\OFTSPOsigmatau}{O^{\FTSPOsigmatau}}
\newcommand{\OinvFTSPOsigmatau}{O^{\FTSPOsigmatau,\mathrm{inv}}}

\newcommand{\Z}{\mathbbm{Z}}
\newcommand{\uglyterm}{\zeta}

\newcommand{\C}{\mathbbm{C}}
\newcommand{\E}{\mathop{\bf E\/}}

\newcommand{\ot}{\otimes}
\newcommand{\op}{\oplus}
\newcommand{\proj}[1]{\ensuremath{\lvert#1\rangle \langle #1\rvert}}
\newcommand{\adver}{\ensuremath{\mathcal{A}}\xspace}
\newcommand{\ketbra}[2]{\left|#1\right\rangle\!\!\left\langle #2\right|}

\newcommand{\tp}[2]{\left( #1 \; #2 \right)}

\DeclarePairedDelimiter\abs{\lvert}{\rvert}
\DeclarePairedDelimiter\norm{\lVert}{\rVert}

\DeclarePairedDelimiter\parens{\lparen}{\rparen}
\DeclarePairedDelimiter\braces{\lbrace}{\rbrace}

\allowdisplaybreaks[4]
\setenumerate{label=(\roman*),noitemsep}
\newcommand{\email}[1]{\emph{Email:} \href{mailto:#1}{#1}}
\definecolor{dgreen}{rgb}{.2,.6,.2}

\begin{document}

%=============================================================================
\title{Permutation Superposition Oracles for Quantum Query Lower Bounds}
\author{Christian Majenz%
\texorpdfstring{\footnote{Department of Applied Mathematics and Computer Science, Technical University of Denmark. \email{chmaj@dtu.dk}}}{}
\and Giulio Malavolta%
\texorpdfstring{\footnote{Department of Computing Sciences, Bocconi University, Italy. \email{giulio.malavolta@unibocconi.it}}}{}
\texorpdfstring{\footnote{Max Planck Institute for Security and Privacy, Germany.}}{}
\and Michael Walter%
\texorpdfstring{\footnote{Faculty of Computer Science, Ruhr-Universit\"at Bochum, Germany. \email{michael.walter@rub.de}}}{}}
\date{}
\maketitle
%=============================================================================
% \tableofcontents
%=============================================================================
\begin{abstract}
We propose a generalization of Zhandry's compressed oracle method to random permutations, where an algorithm can query both the permutation and its inverse.
We show how to use the resulting oracle simulation to bound the success probability of an algorithm for any predicate on input-output pairs, a key feature of Zhandry's technique that had hitherto resisted attempts at generalization to random permutations.
One key technical ingredient is to use \emph{strictly monotone factorizations} to represent the permutation in the oracle's database.
As an application of our framework, we show that the one-round sponge construction is unconditionally preimage resistant in the random permutation model. This proves a conjecture by Unruh.
\end{abstract}

% \bigskip

% \clearpage
\tableofcontents
\clearpage
% !TEX root = compressed-pi.tex
%=============================================================================
\section{Introduction}
%=============================================================================
The random oracle model \cite{BR93} is a popular heuristic in the analysis of cryptographic protocols, that abstracts cryptographic objects as random functions and provides oracle access to other algorithms. From a theoretical standpoint, the random oracle model allows one to prove unconditional statements about cryptographic protocols, in a clean and well-defined model. On the practical side, the random oracle model enables efficient cryptographic schemes and essentially every construction (be it a digital signature or a public-key encryption) used in practice relies, in one way or another, on this heuristic in order to analyze security.\footnote{Note however that the ROM (and QROM) are both fundamentally uninstantiable \cite{Canetti:2004:ROM:1008731.1008734}.} When considering security against quantum algorithms, it is natural to extend this model to allow the algorithms to query the random function on a superposition of inputs.  This is commonly referred to as the quantum random oracle model (QROM) \cite{BDF+2011}. However, many of the techniques (and proofs) developed in the (classical) random oracle model do not immediately carry over to the QROM. To illustrate the difference, it suffices to note that in the classical settings the reduction can read the queries of the algorithm, whereas the same action in the QROM may arbitrarily disturb the state of the algorithm.

To cope with this, the community has developed a series of new techniques to analyze quantum algorithm in these settings.
An important method in these settings is the \emph{compressed oracle} technique \cite{Zhandry2019}. Conceptually, this technique is the quantum analogue of the classical lazy-sampling method, which allows the reduction to define the random function only on the inputs queried by the attacker. At a more technical level, the technique considers a purified version of the random function, that allows the reduction to directly inspect the internal state of the compressed oracle simulation (the so-called \emph{database}), in order to gain partial/approximate knowledge about the queries made by the algorithm. An extremely useful property in this context is that the compressed oracle simulation stores (a superposition of) a list of input-output pairs, so to learn something about the value $H(x)$, and whether $H(x)$ is known to the adversary, it is only necessary to inspect one register. This technique has proven extremely successful in analyzing indifferentiability of cryptographic schemes \cite{Zhandry2019}, security reductions for the Fiat-Shamir transformation \cite{LZ19,DFMS21,DFMS21} and the Fujisaki-Okamoto transformation \cite{BHHHP19,DFMS21,HHM22}, and even new lower bounds on the query complexity of quantum algorithms \cite{Zhandry2019,Liu2019,CFHL21} and space-time trade-offs \cite{HM23}.

\paragraph{The Random Permutation Model.} In the random permutation model, algorithms are given oracle access to a uniformly sampled permutation $\pi \in S_N$, as well as its inverse~$\pi^{-1}$. This variant of the random oracle model is motivated by cryptographic schemes, such as the Feistel construction for pseudorandom permutations \cite{LR88} or the industry-standard SHA-3 hash function \cite{dworkin2015sha}, where an attacker has access to both the permutation and its inverse. When considering quantum attackers, it is therefore equally natural to assume that such a permutation can be implemented on a quantum computer (as it is publicly known), and hence queried in superposition. Accordingly, it is natural to model this situation by a \emph{quantum-accessible random permutation oracle}, where one
considers the unitary%
\footnote{Alternatively, one can consider the \emph{in-place} permutation~$V^\pi$ defined as $V^\pi \ket x = \ket{\pi(x)}$, which is well defined since $\pi$ is a bijection, and give the adversary query access to $V^\pi$ and $V^{\pi^{-1}} = (V^\pi)^\dagger$. For the purposes of query bounds both models are equivalent, since either can be simulated using two queries to the other. We discuss this in more detail in \cref{sec:oracle}.}
\begin{align*}
    U^\pi \ket x \ket y = \ket{x} \ket{y\oplus\pi(x)}
\end{align*}
and one gives the adversary query access to~$U^\pi$ and~$U^{\pi^{-1}}$.
Classically, the random oracle and the random permutation model are essentially equivalent. This is in stark contrast to the quantum setting where no such connection is known. So far, it has proven difficult to repeat the success of Zhandry's compressed oracle technique for analyzing quantum query access to a uniformly random permutation: Despite several attempts to come up with a full-fledged compressed permutation oracle \cite{Czajkowski2019a,Unruh21,Czajkowski21} the problem is still open. On the other hand, few existing results rely on a bare-bones, ``un-compressed'' superposition oracle for permutations~\cite{ABKM21,ABKMS22,alagic2023two-sided}, whereas recent works \cite{rosmanis2021tight,unruh2023towards} have made partial progress on this problem but without being able to apply the formalism towards new query bounds.

%\gm{@Chris: Can you have a look at the comparison with related works above?}

Arguably, the main reason for the lack of progress is, that the compressed oracle technique relies on the statistical independence of the output values of a random function, but the output values of a random permutation are, of course, not independent (more discussion on this later). The purpose of this work is to make progress on this front, and expand our technical toolkit in the analysis of random permutation oracles.

%-----------------------------------------------------------------------------
\subsection{Summary of Contributions}
%-----------------------------------------------------------------------------

In this work we propose a new approach to analyze quantumly-accessible permutation and prove query lower bounds in this setting. Our main ingredient is a new analysis of a representation of a permutation, known as its \emph{strictly monotone factorization}, which may be of independent interest. At a technical level, we prove a series of lemmas that facilitate proofs of query lower bound in the random permutation model. Specifically:
\begin{itemize}
    \item We present a \emph{permutation} oracle (\cref{sec:oracle}), that is the permutation analogue of Zhandry's compressed oracle technique. Our oracle makes crucial use of a particular representation of a permutation, known as its strictly monotone factorization (\cref{sec:permutations}).
    \item We prove a \emph{fundamental lemma} for permutation oracles (\cref{sec:fundamental}), where we describe a procedure to (approximately) determine whether an input was queried by the adversary, and we give a bound on the precision of its output.
    \item We propose a \emph{progress measure} for permutation oracles (\cref{sec:progress}), that bounds the success probability of an algorithm after $q$ queries to find an input/output pair that satisfies a given relation.
    \item We prove, as our \emph{main theorem}, a general bound (\cref{thm:main}) for any adversary to produce an input-output pair satisfying some relation $R$, for any $R$.
    This bound states the following:
    % This bound captures, as a special case, the double-sided zero-search conjecture~\cite{unruh2023towards} and the preimage resistance of one-round sponge (\cref{sec:sponge}).
\end{itemize}

\begin{theorem}[Informal]\label{thm:intro main}
    Let $\mathcal{A}$ be an algorithm with quantum query access to a random permutation $\pi\in S_N$ and its inverse $\pi^{-1}$, and let $R$ be a relation. If $\mathcal{A}$ makes at most $q$ queries and outputs $x$, then
    \[
    \Pr[(x, \pi(x)) \in R]\leq O\left(\frac{q^3 r_{\max} \ln(N)}{N} \right),
    \]
    where $r_{\max} = \max \braces*{ \max_x \, \abs{R_x}, \max_y \, \abs{R^\mathrm{inv}_y} }$, with $R_x = \{ y : (x,y) \in R \}$ and~$R^\mathrm{inv}_y = \{ x : (x,y) \in R \}$.
\end{theorem}

Illustrating the power of the new approach, we obtain as special cases the pre-image resistance of the sponge construction for the special case of one absorption round (see \cref{fig:sponge}), and the double-sided zero-search conjecture~\cite{unruh2023towards}.
The former is the problem of finding an $x$ such that $\pi(x\|0^c) = y\|0^c$ for some~$y$; it was the original motivation of our work.
Thus (\cref{sec:sponge}):

\begin{figure}
	\centering
	\includegraphics[width=.7\textwidth]{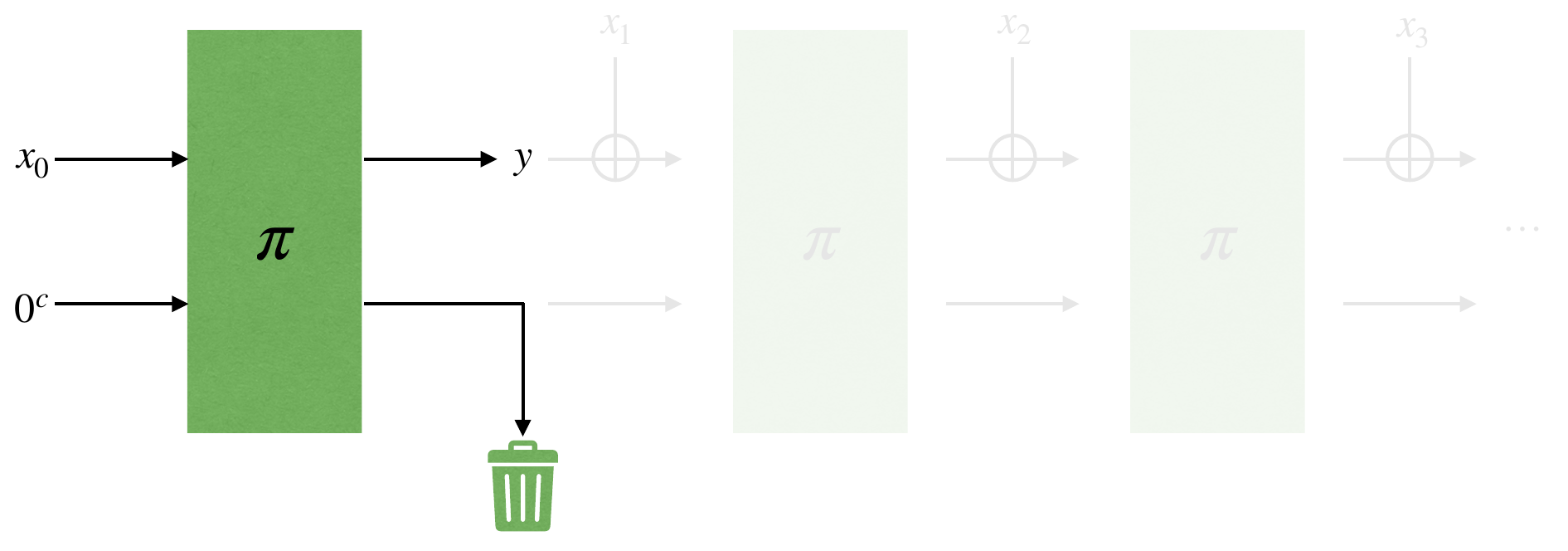}
	\caption{The 1-round sponge.}\label{fig:sponge}
\end{figure}

\begin{cor}[One-Round Sponge, informal]\label{cor:intro sponge}
    Let $\mathcal{A}$ be an algorithm with quantum query access to a random permutation $\pi\in S_{\{0,1\}^n}$ and its inverse $\pi^{-1}$, let $c \in [n]$, and let $y\in\{0,1\}^{n-c}$.  If $\mathcal{A}$ makes at most $q$ queries and outputs~$x$,~then
 \[
        \Pr[\exists y'\in\{0,1\}^c:\pi(x\|0^c)=y\|y'] \leq O\left(\frac{q^3 n}{2^{\min(c,n-c)}}\right).
        %\frac {q^3 (n+2)} {2^{\min(c,n-c)}}.$
\]
\end{cor}

\begin{cor}[Double-Sided Zero-Search, informal]\label{cor:intro zero search}
    Let $\mathcal{A}$ be an algorithm with quantum query access to a random permutation $\pi\in S_{\{0,1\}^{n}}$ and its inverse $\pi^{-1}$ and let $c\in[n]$.  If $\mathcal{A}$ makes at most $q$ queries and outputs $x$, then
 \[
        \Pr[\exists y\in\{0,1\}^{n-c}:\pi(x\|0^c)=y\|0^c] \leq O\left(\frac{q^3 n}{2^c}\right).
\]
\end{cor}

%-----------------------------------------------------------------------------
\subsection{Key Challenges and Techniques}
%-----------------------------------------------------------------------------
To understand the main challenges of extending Zhandry's method to permutations, and how we overcome these to prove our results, it is useful to recall the key properties of the compressed oracle method, and see why they fail for the case of permutation.
The output values of a random function are independent as random variables.
The crucial implication of this fact is that there exists a list of random variables that are
\begin{enumerate}
	\item independent, and
	\item learning one output requires looking at only one of the random variables, and
	\item each variable stores only information about one output.
\end{enumerate}
It is easy to find a representation of a random permutation as a list of random variables that has some of the  properties, but it is manifestly impossible to find one that has all, as the outputs are not independent.
%The first (seemingly obvious) property is that the compressed oracle stores (superpositions of) input-output pairs. This means that, in order to learn something about the value~$H(x)$ of the random function~$H$ at some point~$x\in[N]$, and whether this value is known to the adversary, it is only necessary to measure one register. On the other hand, for the case of permutations there is no such natural encoding strategy: The information about permutation $P$ of an input $x$ is not ``localized'' to $x$, since it affects all other inputs as well.
For instance, by knowing that $\pi(x) = y$ we can also infer that $\pi(x') \neq y$, for all $x' \neq x$.

One can therefore parameterize the set of approaches by which of the properties is given up on. In \cite{Unruh21}, for example, a formalization using lazy sampling of permutations as partial functions is used, giving up on independence. In \cite{rosmanis2021tight}, an analysis via the representation theory of the permutation group is conducted, which yields independent data via the decomposition into irreducible representation and symmetries, but gives up at least partially on locality.
%\MW{Would be nice to give pointers to the precise results in the main part for each of the claims below.}

\paragraph{Random Permutations from Independent Transpositions.} Our main idea is to use the well-known fact that any permutation $\pi \in S_N$ admits a unique decomposition
\begin{equation}\label{eq:strictly-monotone-intro}
	\pi = \tp N {t_N} \tp {N\!-\!1} {t_{N-1}} \cdots \tp 2 {t_2} \tp 1 {t_1}
\end{equation}
where we denote by $\tp{k}{t_k}$ the transposition that sends $k$ to $t_k$ and viceversa. If one leaves out trivial transpositions (i.e., the factors with~$k = t_k$), one obtains a so-called \emph{strictly monotone} factorization of a permutation.
It is well-known that any permutation~$\pi \in S_N$ has a unique strictly monotone factorization, with number of terms equal to the Cayley distance between~$\pi$ and the identity permutation, that is, the minimum number of transpositions (of arbitrary type) in any factorization of~$\pi$.

A useful property of this representation is that the transpositions making up a uniformly random permutation via the strictly monotone factorization are  independent; $t_k$ is uniformly chosen from the set $\{1, \dots, k\}$.  The main, less easy-to-see, property of this decomposition that we are going to use, is the fact that we can ``track'' the set of transpositions that act non-trivially on a given input $x$ (which we refer to as being \emph{active} for $x$). Crucial to our analysis is the fact that, for any given input $x$, the expected number of active transpositions is \emph{small}, i.e., at most about $\ln(N)$. To compute the random permutation, only a small amount of information about non-active transpositions is retrieved, namely that they are not active. On the other hand, the output sensitively depends on the value $t_k$ for an active location $k$, in the sense that changing $t_k$ to any other value changes the output. This gives us a way to quantify the \emph{sparsity} of the (quantum analogue of the) list of transposition that has been read. Although this means we only have an approximate \emph{and} relaxed variant of the second property above, we will be able to show that this quantity is small enough to emulate the functionality of the compressed oracle method for the case of permutations.

\paragraph{The Permutation Oracle.}
With the above discussion in mind, let us now describe (a simplified version of) our permutation oracle. The simulation initializes a \emph{database} of $N$ registers $D = D_1 \dots D_N$ and the $k$-th register $D_k$ is initialized with the state
\[
\ket{+_k} = \frac1{\sqrt k} \sum_{t=1}^k \ket{t}
\]
which will be interpreted as the uniform superposition over all possible values $t_k$ in \cref{eq:strictly-monotone-intro}. Since any given permutation $\pi \in S_N$ uniquely determines a basis state $\ket{\pi} = \ket{t_1, \dots, t_N}$, it is easy to see that the initial state of the database corresponds to a uniform superposition over all possible permutations:
\[
\bigotimes_{k=1}^N \ket{+_k}_{D_k} = \frac1{\sqrt{N!}} \sum_{\pi \in S_N} \ket{\pi}_D.
\]
We can then define the forward query oracle ($\OSPO_{XYD}$) and the inverse query oracle ($\OinvSPO_{XYD}$) provided to the adversary by their actions of the basis states $\ket{x}_X$, $\ket{y}_Y$, and $\ket{\pi}_D$ as:
\begin{align*}
    \OSPO_{XYD} \ket{x,y,\pi}_{XYD} &= \ket{x,y\oplus\pi(x),\pi}_{XYD}, \\
    \OinvSPO_{XYD} \ket{x,y,\pi}_{XYD} &= \ket{x,y\oplus\pi^{-1}(x),\pi}_{XYD}.
\end{align*}
In other words, the oracles apply a series of transpositions controlled on the current state of the database.

\paragraph{A Fundamental Lemma for Permutation Oracles.}
What makes the Zhandry's framework so useful is the ability to read off the database information about the adversary's state.\footnote{Although the analogy with Zahndry's technique is helpful for understanding our framework, the direct comparison is somewhat inaccurate. In particular, contrary to Zhandry's method, we will not attempt to \emph{compress} the database in any way. As a consequence of this, our simulation will not be computationally efficient, which is sufficient to prove query lower bounds. We leave the development for an efficient version of our framework as ground for future work.} Thus our next task is to develop the necessary machinery for doing so. The first step is to prove a \emph{fundamental lemma}, a statement akin to the existing bound for random functions \cite{Zhandry2019,CFHL21}. Fix an input $x$ and an output $y$, simplifying a bit, the content of the fundamental lemma is the approximate equivalence of the following two experiments: %\cm{to make it easier to understand, can we here in the intro write "...distinguishability of the following to experiments", or maybe even better write "...,simplifying a bit, the content of the fundamental lemma is the approximate equivalence of the following two experiments."}\gm{changed}

\begin{itemize}[leftmargin=2cm]
\item[\emph{Exp.~I:}] Read the value $\pi(x)$ off the database state and accept if $\pi(x) = y$, reject otherwise.
\item[\emph{Exp.~II:}] Check if the adversary has queried $x$, and reject if this is not the case. Else proceed as above.
\end{itemize}
The (approximate) equivalence of these two experiments is useful to implement a somewhat ``gentle'' measurement on the database state, for a given input $x$, since if we detect that the adversary never queried $x$, there is no need to disturb the state any further.

Of course at this point it is not clear what we exactly mean by ``the adversary has queried $x$'', and so the next step is to make this notion more precise. We observe that querying the superposition oracle on a basis state $\ket{x}$ must have a non-trivial effect on database location $D_x$ in the computational basis, since it is determining the value of the first transposition, which is always active for $x$. Thus, we can formally define this quantity via a binary-outcome projective measurement
\[
\mathcal{M}_{D_x} := \left\{\proj{+_x}_{D_x}, I - \proj{+_x}_{D_x}\right\}
\]
and conditioning on the second outcome occurring. As a sanity check, note that if the first outcome is observed instead, then the database state is in its initial condition. With some routine calculation, we can then derive a bound
\[
\left|\Pr[\text{Exp.~I accepts}] - \Pr[\text{Exp.~II accepts}] \right| \leq \sqrt{\frac{1}{x}}
\]
where the term $1/x$ comes from the non-commutativity of the measurement $\mathcal{M}_{D_x}$ and the standard basis measurement used to determine $\pi$.

Unfortunately this bound on its own is not very meaningful: To see why, simply take $x=1$ where we obtain a trivial bound. The source of this problem is the \emph{asymmetric} treatment of different registers in the representation of the permutation, where lower registers have a much smaller set of possible transpositions. To deal with this, we introduce our next idea.

\paragraph{Twirling the Oracle.} We overcome the challenge by randomizing the order in which we apply the strictly monotone decomposition. We address this by pre- and post-composing, or ``twirling'', the permutation oracle with two random permutations $\tau$ and $\sigma$. In other words, we define the \emph{twirled} version of the permutation oracle as
\begin{align*}
  \OFTSPOsigmatau_{XYD} \ket{x,y,\pi}_{XYD}
&= \ket{x,y \op \tau^{-1}(\pi(\sigma(x))),\pi}_{XYD} \\
  \OinvFTSPOsigmatau_{XYD} \ket{x,y,\pi}_{XYD}
&= \ket{x,y \op \sigma^{-1}(\pi^{-1}(\tau(x))),\pi}_{XYD}
\end{align*}
for forward and inverse queries, respectively. This yields our final construction, which we call the \emph{twirled superposition permutation oracle}. Note that the permutations $\tau$ and $\sigma$ are treated differently than $\pi$ in the simulation, since we do not require any special property from their representation and their sole purpose is to randomize the view of the adversary in the sense that the adversary does not know which inputs correspond to small values $x$ in the untwirled oracle.  Another way to describe the twirled superposition permutation oracle is that it is constructed based on the strictly monotone factorization \emph{in a random order}. This technique renders the permutation actually stored in a quantum register independent from the view of the algorithm interacting with the twirled superposition permutation oracle.

Equipped with this oracle, our bound obtained in the analysis of the fundamental lemma translates into an \emph{expectation} over the random choice of the register $x$. Considering the \emph{square} of the difference between the success probability of the two experiments, and taking the expectation over $x$, we obtain the bound
\[
\frac1 N \sum_{x=1}^N \frac 1 x
\leq \frac {\ln(N) + 1} N
\]
using a standard bound on the harmonic sum. Note that the final bound is independent of $x$.

\paragraph{The Progress Measure.} Once we have established a procedure to read information off the database, what is left to be decided is \emph{what} we want to read from the database. Due to the challenge described above, we cannot straightforwardly implement a predicate that checks whether there exist an input-output pair $(x, \pi(x))$ that was read by the adversary, such that $(x, \pi(x)) \in R$, for some relation $R$. Instead, for a given input~$x$, our progress measure is defined in terms of the following two-step procedure:
\begin{enumerate}
    \item Apply the projective measurement $\mathcal{M}_x$ defined above, and reject unless the second outcome is obtained (intuitively, this rejects unless the permutation was queried on input~$x$ by the adversary).
    \item Check if $(x,\pi(x))$ indeed satisfies the relation $R$, which can be implemented by the following predicate:
    \[
    \Pi_D^{R,x} := \sum_{\pi \in S_N : (x,\pi(x)) \in R} \proj\pi_D.
    \]
\end{enumerate}
This procedure can be summarized by the following measurement operator:
\[
E^{R,x}_D := \Pi^{R,x}_D (I - \proj{+_x}_{D_x}).
\]
%\MW{I think we still need to define the actual progress measure.}
It is easy to bound the progress measure if the adversary makes no query, since the above projection is acting on a uniform superposition. The challenge, which is the technically most involved part of our work, is to track the bound on the progress measure as the adversary queries the oracles in the forward and inverse direction. At a very high-level, we achieve this by splitting the effect of the action of the oracle $Q^{\SPO}_{XYD} = \{\OSPO_{XYD}, \OinvSPO_{XYD}\}$ in two terms
\begin{align*}
&\norm*{ E^{R,x}_D Q^{\SPO}_{XYD} \ket\phi } - \norm*{ E^{R,x}_D \ket\phi } \\
&\leq \norm*{E^{R,x}_D Q^{\SPO}_{XYD} (I - E^{R,x}_D) \ket\phi}\\
&\leq \norm*{ E^{R,x}_D  Q^{\SPO}_{XYD} (I - \Pi_D^{R,x})(I - \proj{+_x}_{D_x}) \ket\phi }
+ \norm*{ E^{R,x}_D Q^{\SPO}_{XYD} \proj{+_x}_{D_x} \ket\phi }
\end{align*}
for any state $\ket{\phi}$. We bound the two summands separately.

The RHS summand is bounded with a delicate analysis on the effect of the query unitary on the joint database-adversary state. We refer the reader to the technical sections for the calculations and we only mention here that the bound that one obtains with such analysis will not be sufficient for our main theorem. Instead, we will once again use the \emph{randomization} of the register $x$ and the twirling of the permutation $\pi$ to transform the worst-case bound in a much sharper \emph{average-case} bound
\[
\frac{1}{N} \sum_{x=1}^N\norm*{ E^{R,x}_D Q^{\SPO}_{XYD} \proj{+_x}_{D_x} \ket\phi }^2 \leq O\left(\frac{q^2 r_{\max} \ln(N)}{N^2}\right)
\]
where $r_{\max}$ is the maximum number of $y$ such that $(x,y)\in R$, for all $x$. We can then turn this inequality back to a \emph{worst-case} bound over $x$ with a pidgeonhole argument, at the cost of losing a factor $N$ in the bound. Fortunately, the resulting term is still small enough to obtain a good bound. Next, we deal with the LHS summand of the above bound.

\paragraph{Sparsity Analysis.} To bound the LHS of the summand, let us first manipulate the expression
\begin{align*}
     \norm*{ E^{R,x}_D  Q^{\SPO}_{XYD} (I - \Pi_D^{R,x})(I - \proj{+_x}_{D_x}) \ket\phi }
     &= \norm*{ E^{R,x}_D (I - \Pi_D^{R,x}) Q^{\SPO}_{XYD} (I - \proj{+_x}_{D_x}) \ket\phi } \\
     &\leq \norm*{ E^{R,x}_D (I - \Pi_D^{R,x}) Q^{\SPO}_{XYD}} \cdot\norm*{(I - \proj{+_x}_{D_x}) \ket\phi }\\
     &= \norm*{ E^{R,x}_D (I - \Pi_D^{R,x})} \cdot \norm*{(I - \proj{+_x}_{D_x}) \ket\phi }\\
     &\leq \sqrt{\frac{r_{\max}}{x}} \norm*{(I - \proj{+_x}_{D_x}) \ket\phi }
\end{align*}
where the first inequality follows by the submultiplicativity of the operator norm and the second inequality is obtained by observing that
\[
\norm*{E^{R,x}_D (I - \Pi_D^{R,x})} = \norm*{\Pi_D^{R,x} (I-\ketbra{+_x}{+_x}_{D_x}) (I - \Pi_D^{R,x})} = \norm*{\Pi_D^{R,x} \ketbra{+_x}{+_x}_{D_x} (I - \Pi_D^{R,x})} \leq \norm*{\Pi_D^{R,x} \ketbra{+_x}{+_x}_{D_x}}
\]
which can be bound to $\sqrt{r_{\max}/x}$ using the same argument as in the fundamental lemma. Thus, bounding this term boils down to bounding the amount of locations read by the adversary, i.e., the number of \emph{active} registers in the database. Once again, the number of active locations on the initial state of the database is zero, so bounding this term involves analyzing the effect of the query unitary $Q_{XYD}$ on the joint database-adversary state. A delicate analysis leads to a bound of
\begin{align*}
	\frac{1}{N} \sum_{x=1}^N\frac{\abs{R}}x \norm*{ \parens[\big]{ I - \proj{+_x}_{D_x} } \ket{\psi} }^2 \le  O\left(\frac {q^3r_{\max}\ln (N)}{N^2}\right)
\end{align*}
in expectation over $x$, which we can once again turn into a worst-case bound on every $x$ at the cost of an extra $1/N$ factor.
%An interesting aspect of this analysis that we highlight here is that the bound on the forward and inverse queries are obtained using different strategies: For the former, we use more standard analysis on the effect of the unitary on the state, and we bound the norm of the projectors by connecting them to the probability of picking the ``right'' transposition in a random sampling. On the other hand, for inverse queries we use the fact that inverse queries can be simulated using forward queries \emph{plus post-selection}. Although the success probability of this post-selection is quite small (about $1/N$) this factor can be absorbed by slightly increasing the \emph{number of active registers}, plus a standard concentration bound.
We refer the reader to the technical sections for more details.

\paragraph{Putting Things Together.} Overall, the above analysis allow us to bound, via our \emph{progress measure}, how the success probability of the predicate $E^{R,x}$ evolves as the adversary performs more queries.
The final but crucial observation is that if~$x$ is the output of the algorithm $\mathcal A$, the measurement~$E^{R,x}$ accepts precisely if Exp.~II accepts.
%\MW{Is it an equality or are there factors of $N$?}
Applying the fundamental lemma to this bound, we obtain a bound on the acceptance probability of Exp.~I (as defined above), which is the quantity that we are interested in.

This outline ignores many subtle aspects of the proof, but it contains the main ideas.
Putting the bounds on the two main terms together, we obtain \cref{thm:intro main} and hence \cref{cor:intro sponge,cor:intro zero search}.

\subsection{Concurrent Independent Work}
%---------------------------------------------------------------
A recent manuscript by Carolan and Poremba \cite{cryptoeprint:2024/414}, developed concurrently and independently from our work, also shows a proof for the double-sided zero-search conjecture of Unruh and the one-wayness of the one-round sponge.
The techniques used to prove the bound are quite different and their work achieves the result using a worst-case to average-case reduction for random permutation, then appealing to known bounds.
The advantage of their approach is that the bound obtained is \emph{tight}.
% In contrast, our paper proposes a new framework to analyze permutation oracles and theorem that we obtain is more general, since it is a bound applicable to all one-input relations $R$.
In contrast, our paper proposes a new framework to analyze permutation oracles and a theorem that applies to arbitrary relations on (single) input-output pairs, promising opportunities for generalization.

% !TEX root = compressed-pi.tex
%=============================================================================
\section{Preliminaries}
%=============================================================================
We abbreviate $[N] := \{1,2,\dots,N\}$. For convenience, we assume that $N=2^n$, so we can identify $[N] \cong \{0,1\}^n$ and use~$\op$ to mean bitwise addition.
For a set~$S$, a probability distribution~$\mu$ and a (classical or quantum) algorithm~$\adver$, we write $x\leftarrow S$, $x\leftarrow D$, and~$x\leftarrow \adver$ for sampling a uniformly random element~$x$ from~$S$, sampling~$x$ according to the distribution~$D$, or running an algorithm~$\adver$ to produce an output~$x$.
If the output is quantum, we usually use upper-case letters, following our conventions for quantum registers discussed below.

%-----------------------------------------------------------------------------
\paragraph{Combinatorics.}
%-----------------------------------------------------------------------------
We will also use the fact that for the \emph{harmonic numbers}
\begin{align}\label{eq:H_N}
  H_N := \sum_{k=1}^N \frac 1 k,
\end{align}
the quantity~$H_N - \ln N$ is monotonically decreasing with~$N$.
In particular, it holds that
% \begin{align*}
% H_N - \ln(N) \leq H_1 - \ln(1) = 1
\begin{align}\label{eq:harmonic bound}
  H_N \leq \ln(N) + 1
\end{align}
% \end{align*}
for every~$N\geq1$.

%-----------------------------------------------------------------------------
\paragraph{Quantum Information.}%\label{subsec:quantum}
%-----------------------------------------------------------------------------
Here we provide some preliminary background on quantum mechanics and quantum information.
For more in-depth accounts we refer the reader to \cite{NC00,watrous2018theory}.

We will label quantum systems by~$A,B,X,Y,$ etc.
Any quantum system~$A$ is characterized by a Hilbert space~$\mathcal H_A$.
When $\mathcal H_A = \C^{\Sigma_A}$ for some finite set~$\Sigma_A$, we call~$A$ a \emph{quantum register}.
This means that~$\mathcal H_A$ has an orthonormal \emph{standard basis}~$\ket a$ labeled by the elements~$a\in \Sigma_A$.
When $\Sigma = [N]$, we can identify $\mathcal H = \C^N$ with its standard basis~$\ket a$ for $a \in \Z_N$.
If we have quantum system composed of two registers, say $A$ and $B$, then the corresponding Hilbert space is the tensor product of the individual Hilbert spaces~$\mathcal H_{AB} = \mathcal H_A \ot \mathcal H_B \cong \C^{\Sigma_{AB}}$, with~$\Sigma_{AB} = \Sigma_A \times \Sigma_B$ labeling the standard product basis.
Accordingly, we may think of the composite system~$AB=(A,B)$ as a register, and similarly for any collection~$S$ of registers.

States of a quantum system are given by \emph{density operators}, that is, positive semi-definite operators of trace one.
We call a state \emph{pure} if this operator is a rank-one orthogonal projector, i.e., equal to $\proj{\phi}$, where~$\ket\phi$ is a unit vector.
We will often identify pure states with unit vectors.
The \emph{trace distance} between two states~$\rho$ and~$\tau$, denoted by $\Td(\rho, \tau)$ is defined as
\[
\Td(\rho, \tau) = \frac{1}{2}\norm*{\rho -\tau}_1 = \frac{1}{2}\tr\left(\sqrt{(\rho-\tau)^\dagger (\rho-\tau)}\right).
\]
The operational meaning of the trace distance is that $\frac12(1+\Td(\rho,\tau))$ is the maximal probability that two states~$\rho$ and~$\tau$ can be distinguished by any (possibly unbounded) quantum channel, when given one or the other with equal probability.

There are two basic kinds of quantum operations.
The first is to apply unitary operators, or \emph{unitaries}.
If~$U$ is a unitary on~$\mathcal H$ and we apply it to a state~$\rho$, the result is $U \rho U^\dagger$, which is again a state.
We denote by~$\mathcal U(\mathcal H)$ the group of unitary operators on~$\mathcal H$, and abbreviate~$\mathcal U(N) = \mathcal U(\C^N)$.
The second is to measure the quantum state.
We will only require \emph{projective measurements}, which are given by a family of orthogonal projections~$\{P_\omega\}_{\omega\in\Omega}$, labeled by some finite index set~$\Omega$, such that~$\sum_{\omega\in\Omega} P_\omega = I$.
If one applies such a measurement to a system in state~$\rho$, then the probability of seeing outcome~$\omega\in\Omega$ is~$p_\omega = \tr(P_\omega\rho)$,  in which case the state changes to~$P_\omega \rho P_\omega / p_\omega$.
If $\{P_a\} = \{\proj a\}$ consists of the projections onto the standard basis of some register, this is called a \emph{standard basis measurement}.

We will use subscripts to denote the corresponding quantum system or tensor factors, e.g., $\rho_{AB}$ denotes a density operator on~$\mathcal H_{AB} = \mathcal H_A \ot \mathcal H_B$, and~$\rho_{AB} = \proj{\Psi}_{AB}$ in the case of a pure state, with~$\ket\Psi_{AB} \in \mathcal H_{AB}$.
Similarly, we write~$U_A$ in the case of a unitary on~$\mathcal H_A$.
In particular, $I_A$ denotes the identity operator on~$\mathcal H_A$.

For unitaries and measurement operators (but never for states), it will be useful to identify operators on some Hilbert space with operators on some any other Hilbert space which includes the former as a tensor factor, by tensoring with the identity operator.
For example, we will often abbreviate the operator~$U_A \ot I_B$ on $\mathcal H_{AB}$ simply by~$U_A$ if no confusion can arise.
This is useful if we want to quantum operations to a subset of registers.
For example, if $U_A$ is a unitary and $\rho_{AB}$ a state, then $U_A \rho_{AB} U_A^\dagger = (U_A \ot I_B) \rho_{AB} (U_A^\dagger \ot I_B)$ is the result of applying the unitary~$U_A$ to the first register~$A$ when the overall system starts out in state~$\rho_{AB}$, and similarly for measurements.

We also recall the gentle measurement lemma~\cite[Lemma~9.4.2]{wilde2019quantum}.

\begin{fact}[Gentle measurement]\label{lem:gentle}
Let $\rho$ be a quantum state, and let $\{ P, I - P\}$ be any projective measurement with two outcomes.
If $\tr(P \rho) \geq 1-\delta$, the post-measurement state $\rho' := P \rho P / \tr(P \rho)$ satisfies
\[
  \Td(\rho, \rho') \leq \sqrt{\delta}.
  %\text{\MW{$\sqrt\delta$, since we use the normalized trace distance?}\gm{Remove the 2, and fix the sparsity lemma}}
\]
\end{fact}

\paragraph{Quantum Algorithms.}
In this work we consider the query complexity of algorithms with quantum access to oracles.
An oracle is modeled by one or more unitaries~$O$ operating on an input/output register~$Z$ and possibly some internal register~$D$ (which will always be initialized explicitly).
% (for us this will usually be the database register~$D$).
A quantum algorithm $\mathcal A$ making queries to this oracle has,
without loss of generality, two registers -- the oracle's input/output register~$Z$ and an internal work register~$A$. It takes the following form:
First, the algorithm's registers are initialized in the initial state~$\ket0_{AZ}$. 
Then the algorithm alternatingly applies oracle-independent unitaries and query unitaries:
\begin{align}\label{eq:generic oracle algo}
  U^{(q)}_{AZ} O_{ZD} U^{(q-1)}_{AZ} O_{ZD} \ldots U^{(1)}_{AZ} O_{ZD} U^{(0)}_{AZ}
\end{align}
Note that $\mathcal A$ can be given access to any oracle with input/output register $Z$.
Finally, some (sub)registers might be measured or returned directly to obtain the classical and quantum outcomes of the algorithm.
We write $\mathcal A^O$ for such an application of a query algorithm to an oracle~$O$.
In the above situation, we say that the algorithm makes~$q$ queries to the oracle.
Thus we only consider the query complexity of an algorithm, but not the time complexity of the unitaries~$U^{(k)}$.
In other words, these unitaries need not be efficient.
In particular, any advice state can be placed in the adversary's quantum memory by using the first unitary $U^{(0)}$.

% !TEX root = compressed-pi.tex
%=============================================================================
\section{Random Permutations}\label{sec:permutations}
%=============================================================================
We denote by $S_N$ the permutation group on~$N$ elements, that is, the group of bijections of~$[N]$.
We have a chain of subgroups $S_1 \subset S_2 \subset \dots \subset S_N$, where for each $j\in[N\!-\!1]$ we think of~$S_j$ as those permutations in~$S_{j+1}$ that fix the element~$j+1$.
For $k, l \in [N]$, we denote by
$\tau = \tp k l$
the \emph{transposition} that sends~$\tau(k) = l$ and~$\tau(l) = k$.
It will be convenient to allow~$k=l$, in which case~$\tp k l=\tp k k$ is the identity permutation.

%---------------------------------------------------------------
\subsection{Random Permutations from Independent Transpositions}\label{subsec:subgroup-tower}
%---------------------------------------------------------------
The starting point for our work is the following representation of permutations.

\begin{lem}\label{lem:towerdecomposition}
For any $\pi\in S_N$, there exist unique numbers $t_k \in [k]$ for $k\in[N]$ such that
\begin{align}\label{eq:tower-decomposition}
	\pi = \tp N {t_N} \tp {N\!-\!1} {t_{N-1}} \cdots \tp 2 {t_2} \tp 1 {t_1}.
\end{align}
While we always have $t_1 = 1$, it is useful to include this term to obtain simpler formulas below.
\end{lem}
\begin{proof}
Any permutation~$\pi \in S_N$ has a unique decomposition of the form
\begin{align}\label{eq:pi into t and sigma}
	\pi = \tp N t \sigma,
\end{align}
where $t \in [N]$ and $\sigma \in S_{N-1}$.
Indeed, for \eqref{eq:pi into t and sigma} to hold we must have~$t = \pi(N)$ and hence~$\sigma = \tp N t \pi$, but these formulas always define a valid decomposition of the form of \cref{eq:pi into t and sigma}.
The lemma follows by induction.
\end{proof}

If one leaves out trivial transpositions in \cref{eq:tower-decomposition} (i.e., the factors with~$k = t_k$), one obtains a so-called \emph{strictly monotone} factorization.
It is well-known that any permutation~$\pi \in S_N$ has a unique strictly monotone factorization, with number of terms equal to the Cayley distance between~$\pi$ to the identity permutation, that is, the minimum number of transpositions (of arbitrary type) in any factorization of~$\pi$.

\begin{cor}\label{cor:independent}
A random permutation~$\pi \in S_N$ is uniformly random if and only if the numbers~$t_k \in [k]$ for~$k\in[N]$ in \cref{eq:tower-decomposition} are independent and uniformly random.
\end{cor}

Given a permutation in the form of \cref{eq:tower-decomposition}, it is easy to compute the inverse:
\begin{align}\label{eq:inverse tower}
	\pi^{-1} =  \tp 1 {t_1} \tp 2 {t_2} \cdots \tp {N\!-\!1} {t_{N-1}} \tp N {t_N}.
\end{align}
Note however that this decomposition is in general \emph{not} of the form of \cref{eq:tower-decomposition}.
It will also be convenient to introduce the following notation:
\begin{equation}\label{eq:pi gt lt k}
\begin{aligned}
	\pi_{>k} &= \tp N {t_N} \tp {N\!-\!1} {t_{N-1}} \cdots \tp {k+1} {t_{k+1}}, \\
	\pi_{<k} &= \tp {k-1} {t_{k-1}} \cdots \tp 2 {t_2} \tp 1 {t_1}.
\end{aligned}
\end{equation}
Note that $\pi = \pi_{>k} \tp k {t_k} \pi_{<k}$ and $\pi^{-1} = (\pi_{<k})^{-1} \tp k {t_k} (\pi_{>k})^{-1}$.
We use the convention that the subset-of-transpositions subscripts take precedence before other operations modifying a permutation, e.g.\ $\pi_{<k}^{-1}\coloneqq \left(\pi_{<k}\right)^{-1}$.

The following lemma will be useful in \cref{sec:expectation}.

\begin{lem}\label{lem:wrong side}
Let $\pi \in S_N$ be uniformly random and $k \in [N]$.
\begin{enumerate}
\item\label{it:wrong side 1} If $\xi \in S_k$ is uniformly random and independent from~$\pi$, then $\pi_{>k} \xi$ is uniformly random in~$S_N$.
\item\label{it:wrong side 2} If $\xi \in S_k$ is uniformly random and independent from~$\pi$, then $\xi \pi_{>k}$ is uniformly random in~$S_N$.
\item\label{it:wrong side 3} For every $\ell\in\{k+1,\dots,N\}$, it holds that $\Pr\mleft( \pi_{>k}(k) = \ell \mright) = \frac1N$.
\end{enumerate}
\end{lem}
\begin{proof}
\begin{enumerate}
\item This is clear from \cref{cor:independent}.
\item Using the notation of \cref{eq:pi gt lt k}, we have
\begin{align*}
    \xi \pi_{>k}
&= \xi \tp N {t_N} \xi^{-1} \xi \tp {N\!-\!1} {t_{N-1}} \xi^{-1} \cdots \xi \tp {k+1} {t_{k+1}} \xi^{-1} \xi \\
&= \tp N {\xi(t_N)} \tp {N\!-\!1} {\xi(t_{N-1})} \cdots \tp {k+1} {\xi(t_{k+1})} \xi,
\end{align*}
where we used that $\xi \in S_k$ and hence it fixes~$k+1,\dots,N$,
Now, the~$t_\ell \in [\ell]$ for~$\ell>k$ are uniformly random and independent from~$\xi$, so the same is true for the~$\xi(t_\ell) \in [\ell]$ for~$\ell>k$.
Thus we see that the distribution of~$\xi \pi_{>k}$ is the same as the distribution of~$\pi_{>k} \xi$, and the claim follows from part~\ref{it:wrong side 1}.
\item This follows from the preceding:
\begin{align*}
  \Pr\mleft( \pi_{>k}(k) = \ell \mright)
= \Pr\mleft( (\xi \pi_{>k})(k) = \ell \mright)
= \frac1N.
\end{align*}
with $\xi$ as in part~\ref{it:wrong side 2}.
\end{enumerate}
\end{proof}

%-----------------------------------------------------------------------------
\subsection{Active Transpositions}\label{subsec:active}
%-----------------------------------------------------------------------------
In the following we characterize the transpositions that, in the above decomposition of a permutation, are relevant to determine the action of the permutation on a given element.
The results of this section will not be needed to prove our main theorem, but we include them here for the purpose of building up an intuition.

Given a permutation~$\pi\in S_N$, consider its unique decomposition as in \cref{lem:towerdecomposition}:
\begin{align*}
  	\pi = \tp N {t_N} \tp {N\!-\!1} {t_{N-1}} \cdots \tp 2 {t_2} \tp 1 {t_1},
\end{align*}
When does a given transposition $\tp k {t_k}$ impact the computation of~$\pi(x)$ for some given~$x\in[N]$?
To study this we introduce the following definitions:

\begin{defn}[Active sets]\label{defn:active}
Given a permutation~$\pi \in S_N$ and $x,k \in [N]$, we say~$k$ is \emph{active for~$\pi$ and~$x$} if~$\pi_{<k}(x) \in \{ k, t_k \}$.
Similarly, for $y\in[N]$ we say that $k$ is \emph{inverse-active for~$\pi$ and~$y$} if~$(\pi_{>k})^{-1}(y) \in \{ k, t_k \}$.%
\footnote{Note that this is in general \emph{not} equivalent to saying that $k$ is active for $\pi^{-1}$ and~$y$, as the latter would refer to the decomposition of $\pi^{-1}$ in the form of \cref{eq:tower-decomposition}, rather than to \cref{eq:inverse tower}.}
We denote by $A(\pi,x), A^\mathrm{inv}(\pi,y) \subseteq [N]$ the set of active and inverse-active~$k$, respectively, defined as above.
\end{defn}

If~$k$ is active for $\pi$ and~$x$ then changing~$t_k$ to any other value will lead to a different~$\pi(x)$.
Similarly, if $k$ is inverse-active for~$\pi$ and~$y$ then changing~$t_k$ (in the decomposition of~$\pi$) to any other value will lead to a different~$\pi^{-1}(y)$.
(The converses of these statements are in general not true.)
Note that~$\pi_{<k}(x) = k$ if and only if~$x = k$, since~$\pi_{<k} \in S_{k-1}$.

It is clear that the action of a permutation or its inverse on some element only depends on the corresponding active set:

\begin{lem}\label{lem:active-makes-sense}
Let~$\pi \in S_N$ be a permutation.
For any $x\in [N]$, let $A(\pi,x) = \{ k_1 < \cdots < k_\ell \}$ denote the active locations, sorted in increasing order.
Then we have:
\begin{align*}
  \pi(x) =
  \tp {k_\ell} {t_{k_\ell}}
  \tp {k_{\ell-1}} {t_{k_{\ell-1}}}
  \cdots
  \tp {k_2} {t_{k_2}}
  \tp {k_1} {t_{k_1}} (x).
\end{align*}
Similarly, if $y\in[N]$ and $A^\mathrm{inv}(\pi,y) = \{ k_1 < \cdots < k_\ell \}$ are the corresponding inverse-active locations, then:
\begin{align*}
  \pi^{-1}(y) =
  \tp {k_1} {t_{k_1}}
  \tp {k_2} {t_{k_2}}
  \cdots
  \tp {k_{\ell-1}} {t_{k_{\ell-1}}}
  \tp {k_\ell} {t_{k_\ell}} (y).
\end{align*}
\end{lem}
% \begin{proof}
%   We prove this lemma by induction over $N$. For $N=1$ the statement is trivially true. Suppose now the statement holds for $N=k-1$, and let $\pi\in S_k$ and $x\in[k]$. We distinguish two cases.
%     \begin{itemize}
%         \item \emph{Case 1: $(k,t_k)$ is not active for $\pi$ on $x$.} In this case, $j_\ell<k$ and by the induction hypothesis we have
%     \[
%    \sigma_k(x)=(j_\ell t_{j_\ell})(j_{\ell-1} t_{j_{\ell-1}})\ldots (j_1t_{j_1})(x).
%    \]
%    As $(k,t_k)$ is not active for $\pi$ on $x$, we have $\sigma_k(x)\notin\{k, t_k\}$ and thus
%    \begin{align*}
%     \pi(x) &= (kt_k)\Big(\sigma_k(x)\Big)\\
%     &= \sigma_k(x)\\
%     &= (j_\ell t_{j_\ell})(j_{\ell-1} t_{j_{\ell-1}})\ldots (j_1t_{j_1})(x).
%    \end{align*}
%   \item \emph{Case 2: $(k,t_k)$ is active for $\pi$ on $x$.} Again by the induction hypothesis we have
%     \[
%    \sigma_k(x)=(j_{\ell-1} t_{j_{\ell-1}})(j_{\ell-2} t_{j_{\ell-2}})\ldots (j_1t_{j_1})(x)
%    \]
%   for some $j_{\ell-1} < k$.
%     Since $(k,t_k)$ is active for $\pi$ on $x$, by definition we have that $\sigma_k(x)\in\{k, t_k\}$ and therefore
%      \begin{align*}
%     \pi(x)&= (kt_k)\Big(\sigma_k(x)\Big)\\
%     &= (k t_k)(j_{\ell-1} t_{j_{\ell-1}})(j_{\ell-2} t_{j_{\ell-2}})\ldots (j_1t_{j_1})(x)\\
%     &= (j_\ell t_{j_\ell})(j_{\ell-1} t_{j_{\ell-1}})(j_{\ell-2} t_{j_{\ell-2}})\ldots (j_1t_{j_1})(x)
%    \end{align*}
%   setting $(k t_k) = (j_\ell t_{j_\ell})$.
%     \end{itemize}
% \end{proof}

We now show that the event that~$k$ is active for a random permutation (and fixed~$x$) is independent of the numbers~$t_1,\dots,t_{k-1}$, and compute the probability with which this happens:

\begin{lem}\label{lem:active indepe}
Let $x\in[N]$.
Then, for a uniformly random~$\pi \in S_N$, we have:
\begin{align*}
  \Pr\mleft( k \in A(\pi,x) \mid t_1,\dots,t_{k-1} \mright)
= \begin{cases}
  \frac1k & \text{ if $x < k$,} \\
  1 & \text{ if $x = k$,} \\
  0 & \text{ if $x > k$,}
\end{cases}
\end{align*}
hence this is also equal to~$\Pr( k \in A(\pi,x) )$.
In particular, the events $k \in A(\pi,x)$ for $k\in[N]$ are independent.
\end{lem}
\begin{proof}
Recall that $k \in A(\pi,x)$ means that~$\pi_{<k}(x) \in \{ k, t_k \}$.
For~$x < k$, we have~$\pi_{<k}(x) \in [k-1]$, hence the event~$\pi_{<k}(x) \in \{k, t_k\}$ is equivalent to~$\pi_{<k}(x) = t_k$.
For a uniformly random~$\pi$, the numbers~$t_1,t_2,\dots,t_k$ are independent and uniformly random (by \cref{cor:independent}).
Hence~$t_k\in[k]$ is uniformly random given~$t_1,\dots,t_{k-1}$.
As $\pi_{<k}(x) \in [k-1]$ only depends on the latter, it follows that it coincides with~$t_k$ with probability~$\frac1k$.
For~$x \geq k$ we have~$\pi_{<k}(x) = x$, hence the event~$\pi_{<k}(x) \in \{k, t_k\}$ is equivalent to~$x=k$.
\end{proof}

Next we are going to bound the expected number of active locations for a random choice of the permutation.

\begin{lem}\label{lem:active forward}
Let $x\in[N]$.
Then, for a uniformly random~$\pi \in S_N$, we have:
$\E \abs{A(\pi,x)} \leq 1 + \ln (N/x)$.
\end{lem}
\begin{proof}
Recall that we denote by~$H_n = \sum_{k=1}^n \frac 1 k$ the harmonic numbers.
Then, using \cref{lem:active indepe}, we have
\begin{align*}
  \E \abs{A(\pi,x)}
% % &= \sum_{k=1}^N \Pr\mleft( k \in A(\pi, x) \mright) \\
% &= \sum_{k=1}^N \Pr\mleft( \pi_{<k}(x) \in \{k, t_k\} \mright) \\
% % &= 1 + \sum_{k=x+1}^N \Pr\mleft( \pi_{<k}(x) \in \{k, t_k\} \mright) \\
% &= 1 + \sum_{k=x+1}^N \Pr\mleft( \pi_{<k}(x) = t_k \mright) \\
&= 1 + \sum_{k=x+1}^N \frac 1 k
= 1 + H_N - H_x
\leq % 1 + H_N - H_x + (H_x - \ln x) - (H_N - \ln N) =
1 + \ln N - \ln x.
\end{align*}
where the inequality holds as~$H_n - \ln n$ is monotonically decreasing with~$n$.
This completes the proof.
\end{proof}

We also provide a bound on the expected number of inverse-active locations.
Since this is not the same as the expected number of active locations for the inverse, the bound differs from \cref{lem:active forward}.

\begin{lem}\label{lem:active inverse}
Let $y\in[N]$.
Then, for a uniformly random~$\pi \in S_N$, we have:
$\E \abs{A^\mathrm{inv}(\pi,y)} \leq 1 + \frac{2y-2}N < 3$.
\end{lem}
\begin{proof}
% Recall that
% \begin{align*}
% 	(\pi_{>k})^{-1} &= \tp {k\!+\!1} {t_{k+1}} \cdots \tp {N\!-\!1} {t_{N-1}} \tp N {t_N}.
% \end{align*}
Suppose first that there is some~$k \in [N]$ such that~$t_k = y$.
Let~$k^* \in [N]$ denote the largest such~$k$.
Since $t_{k^*} \in [k^*]$, we clearly must have~$k^* \geq y$.
Then we have $(\pi_{>k})^{-1}(y) = y \notin \{ k, t_k \}$ for all~$k > k^*$,
$(\pi_{>k^*})^{-1}(y) = y = t_{k^*}$,
and~$(\pi_{>k})^{-1}(y) = k^* \notin \{ k, t_k \}$ for all~$k < k^*$.
Together, we see that~$A^\mathrm{inv}(\pi,y) = \{ k^* \}$.
On the other hand, if~$t_k \neq y$ for all~$k\in[N]$ then $(\pi_{>k})^{-1}(y) = y$ for all~$k \geq y$ and hence~$A^\mathrm{inv}(\pi,y) = \{ y \} \cup A^\mathrm{inv}(\pi_{<y},t_y)$.
Combining the above observations, we find that
\begin{align*}
	\E \abs{A^\mathrm{inv}(\pi,y)}
% &= 1 + \Pr(t_1 \neq y, \dots, t_N \neq y) \, \E\mleft[ A^\mathrm{inv}(\pi_{<y},t_y) \mid t_1 \neq y, \dots, t_N \neq y\mright] \\
&= 1 + \Pr(t_y \neq y, \dots, t_N \neq y) \, \E\mleft[ A^\mathrm{inv}(\pi_{<y},t_y) \mid t_y \neq y, \dots, t_N \neq y\mright] \\
&= 1 + \frac{y-1}N \E\mleft[ A^\mathrm{inv}(\pi_{<y},t_y) \mid t_y \neq y \mright] \\
&= 1 + \frac1N \sum_{z=1}^{y-1} \E\mleft[ A^\mathrm{inv}(\pi_{<y},z) \mright],
\end{align*}
where we first used that~$\pi_{<y}$ is independent from~$t_y,\dots,t_N$, and the last step holds because~$t_y$ is uniformly random in~$[y-1]$ conditional on~$t_y \neq y$.
Since~$\pi_{<y}$ is uniformly random in~$S_{y-1}$, we obtain the following recurrence for $e(N,y) := \E_{\pi \leftarrow S_N} \abs{A^\mathrm{inv}(\pi,y)}$:
\begin{align}\label{eq:e reccurrence}
	e(N,y)
= 1 + \frac1N \sum_{z=1}^{y-1} e(y-1, z)
= 1 + \frac1N f(y-1),
\end{align}
where we introduced the notation~$f(n) := \sum_{y=1}^n e(n,y)$, with~$f(0) = 0$, which in turn satisfies the recurrence
\begin{align*}
	f(n)
= \sum_{y=1}^n \parens*{ 1 + \frac1n f(y-1) }
= n + \frac1n \sum_{k=0}^{n-1} f(k)
\end{align*}
for all~$n > 0$.
It is easy to see that~$f(n) \leq 2n$ for all~$n \geq 0$ by using induction.
Indeed, this is clearly true for~$n=0$, and if it holds that~$f(k) \leq 2k$ for all~$k < n$ then also
$f(n)
\leq n + \frac1n \sum_{k=0}^{n-1} 2k
= n + n - 1 \leq 2n$.
Using this estimate in \cref{eq:e reccurrence} we obtain the desired result:
\begin{equation*}
	\E \abs{A^\mathrm{inv}(\pi,y)}
= e(N, y)
= 1 + \frac1N f(y-1)
\leq 1 + \frac{2y-2}N.
\qedhere
\end{equation*}
\end{proof}

% !TEX root = compressed-pi.tex
%=============================================================================
\section{Quantum Random Permutations Oracles}\label{sec:oracle} %from independent transpositions} % oracles from subgroup tower}
%=============================================================================
The decomposition of a random permutation introduced in \cref{subsec:subgroup-tower} provides a way of sampling a random permutation by sampling many \emph{independent and smaller} random data, namely the individual transpositions that make up the permutation (\cref{cor:independent}).
Importantly, typical input-output pairs of the random permutation only depend on a few of them (\cref{lem:active forward,lem:active inverse}).
In this section we will use this idea to construct an oracle that exactly simulate a quantum-accessible random permutation, but has an internal state that can be used to analyze quantum query algorithms.
We first introduce some notation.
Given any permutation~$\pi\in S_N$, we denote by~$U^\pi$ the corresponding permutation operator on~$\C^{N} \otimes \C^N$, that is,
\begin{align*}
    U^\pi \ket{x,y} = \ket{x,y\oplus \pi(x)}
\qquad \forall x,y\in[N].
\end{align*}
This defines a unitary representation of~$S_N$ on~$\C^{N}\otimes \C^N$.

\begin{defn}[Quantum-accessible random permutation]\label{defn:standard random perm}
A \emph{quantum-accessible random permutation} consists of query access to~$U^\pi$ and to~$U^{\pi^{-1}}$, for a permutation~$\pi\in S_N$ chosen uniformly at random.
\end{defn}

When $\mathcal A$ is a query algorithm that gets query access to two oracles that act on~$\C^{N} \otimes \C^N$ and $\pi\in S_N$ is a permutation, we write~$\mathcal A^{U^\pi,U^{\pi^{-1}}}$ to indicate that we use the unitaries~$U^\pi$ and~$U^{\pi^{-1}}$ as the two oracles.

Above we defined $U^\pi$ by the usual formula for a quantum oracle corresponding to a Boolean function, but because $\pi$ is a bijection we could also instead work with oracles that modify their input in-place, that is,
\begin{align*}
  V^\pi \ket x = \ket{\pi(x)}
\qquad \forall x\in[N].
\end{align*}
However, to prove a query lower bound we will be able to consider either type of oracles.
This is because the standard and the in-place variants can simulate each other at the cost of doubling the number of queries, if one is given access to the permutation as well as its inverse:
it holds that
\begin{equation*}
\begin{aligned}
    U^\pi_{XY} &= V^{\pi^{-1}}_X \CNOT_{X\to Y} V^{\pi}_X, \\
    U^{\pi^{-1}}_{XY} &= V^{\pi}_X \CNOT_{X\to Y} V^{\pi^{-1}}_X,
\end{aligned}
\end{equation*}
as well as
\begin{equation*}
\begin{aligned}
    V^{\pi}_X \ket0_Y &= U^{\pi^{-1}}_{XY} \SWAP_{X\leftrightarrow Y} U^\pi_{XY} \ket0_Y, \\
    V^{\pi^{-1}}_X \ket0_Y &= U^\pi_{XY}  \SWAP_{X\leftrightarrow Y}U^{\pi^{-1}}_{XY} \ket0_Y.
\end{aligned}
\end{equation*}

%-----------------------------------------------------------------------------
\subsection{Superposition Permutation Oracle}
%-----------------------------------------------------------------------------
We first construct a \emph{superposition oracle} for random permutations.
Like those for random functions, it is obtained by replacing the random (classical) choice of permutation by a uniform (quantum) superposition.
Our oracle is specified by an internal quantum state space, an initial state, query unitaries (one for the random permutation and one for its inverse), and a recovery routine.
The query unitaries will be constructed by applying the transpositions~$\tp k {t_k}$ in the right order, with each $t_j$ obtained from the internal state of the oracle.
Generalizing the approach of \cite{CFHL21}, we propose the following definition.

\begin{defn}[Superposition permutation oracle]\label{def:SPO}
The \emph{superposition permutation oracle} (\SPO) is defined as follows:
\begin{itemize}
\item The state space, called the \emph{database}, consists of $N$ registers, $D = D_1 \cdots D_N$ with the $k$-th register~$D_k$ having dimension~$k$ and computational basis~$\ket 1,\dots,\ket k$.
Any permutation~$\pi \in S_N$, determines a basis state~$\ket\pi_D = \ket{t_1,\dots,t_N}_D$, where the numbers~$t_k\in[k]$ are chosen as in \cref{eq:tower-decomposition}.
\item The initialization routine~$\Init^\SPO_D$ initializes each register in a uniform superposition over the basis states.
That is, the initial state of the database is
\begin{align*}
	\ket{\Phi_{\SPO}}_D
= \frac1{\sqrt{N!}} \sum_{\pi \in S_N} \ket{\pi}_D
= \bigotimes_{k=1}^N \ket{+_k}_{D_k},
\text{ where }
\ket{+_k}_{D_k} = \frac1{\sqrt k} \sum_{t=1}^k \ket{t}_{D_k}.
\end{align*}
\item There are two query unitaries,~$\OSPO_{XYD}$ and~$\OinvSPO_{XYD}$, that define the two interfaces available to the query algorithm and that simulate oracles for a random permutation and its inverse.
They are defined as follows:
For all $x,y\in[N]$ and $\pi\in S_N$,
\begin{align*}
    \OSPO_{XYD} \ket{x,y,\pi}_{XYD} &= \ket{x,y\op\pi(x),\pi}_{XYD}, \\
    \OinvSPO_{XYD} \ket{x,y,\pi}_{XYD} &= \ket{x,y\op\pi^{-1}(x),\pi}_{XYD}.
\end{align*}
%\gm{I wasn't sure how to express the oracle in terms of controlled transpositions, so I commented out the sentence. We do not seem to be using this decomposition anywhere, but please double check.}
%\MW{Good to still make a comment because we use it e.g.\ when we say that $\OSPO_{XYD} \ket z_X$ is independent of register $D_x$ for $x<z$. Actually shall we make this a lemma and cite it at the appropriate places?}
\item The recovery routine $\Rec^\SPO_D$ simplify measures all registers~$D_1,\dots,D_N$ in the computational basis to obtain~$t_k\in[k]$ for $k\in[N]$.
It outputs the corresponding permutation according to \cref{eq:tower-decomposition}.
\end{itemize}
\end{defn}
When $\mathcal A$ is a query algorithm that gets query access to two oracles acting on two $N$-dimensional register~$X$ and $Y$, and if~$D$ is the database register of a superposition permutation oracle, we write~$\mathcal A^{\SPO_D}$ to indicate that 
%before $\mathcal A$'s runtime, $\Init^\SPO_D$ is executed \MW{$\leftarrow$ This was added but does not match the current text. I would propose removing it again.}\gm{Fine for me.}, and 
we use the interfaces~$\OSPO_{XYD}$ and~$\OinvSPO_{XYD}$, respectively, to implement the two types of oracles queries.
It is straightforward to verify that the \SPO then exactly simulates a quantum-accessible random permutation.

\begin{lem}\label{prop:randompi vs spo}
  Let \adver be a query algorithm that gets query access to two oracles that act on~$\C^{N}\otimes \C^N$.
  Then the joint state of the classical random variable~$\pi$ and the quantum register~$B$ is the same for the following two experiments:
  \begin{enumerate}
    \item\label{it:classical random choice} Sample $\pi\leftarrow S_N$ uniformly at random and run $B \leftarrow \adver^{U^\pi, U^{\pi^{-1}}}$.
    \item\label{it:superposition measure last} Run~$\Init^\SPO_D$, then $B \leftarrow \adver^{\SPO_D}$, and finally $\pi \leftarrow \Rec^\SPO_D$.
  \end{enumerate}
\end{lem}
\begin{proof}
We show that both experiments give rise to the same joint state as the following:
\begin{enumerate}[start=3]
    \item\label{it:superposition measure first} Run~$\Init^\SPO_D$, then $\pi \leftarrow \Rec^\SPO_D$, and finally $B \leftarrow \adver^{\SPO_D}$.
\end{enumerate}
Indeed, \ref{it:classical random choice} and~\ref{it:superposition measure first} are equivalent by \cref{cor:independent} and the fact that measuring a uniform superposition yields a uniformly random sample, while~\ref{it:superposition measure first} and~\ref{it:superposition measure last} are equivalent because measuring in the computational basis commutes with unitaries that are controlled on this basis.
\end{proof}

A consequence of \cref{{lem:towerdecomposition}} is that $\pi(x)=\pi_{\ge x}(x)$, so in particular $\pi(x)$ does not depend on $t_{x'}$ for $x'<x$. For the \SPO, this means that when a query is made with input~$x$, the registers $D_{x'}$ for~$x'<x$ are not used, i.e., the query operator acts as the identity on them.

\begin{lem}\label{lem:small-x-not-touched}
	The \SPO query operator fulfils the equation
	\begin{align*}
		\OSPO_{XYD}\ket x_X\ket y_Y\ket{\pi_{\ge x}}_{D_{\ge x}}=\ket x_X\ket{y\oplus \pi_{\ge}(x)}_{Y}\ket{\pi_{\ge x}}_{D_{\ge x}}\otimes I_{D_{<x}}.
	\end{align*}
\end{lem}
\begin{proof}
	This follows directly from \cref{def:SPO,{lem:towerdecomposition}}.
\end{proof}

%-----------------------------------------------------------------------------
\subsection{Twirled Superposition Permutation Oracle}
%-----------------------------------------------------------------------------
Just like is the case for Zhandry's compressed oracle for random functions, we would like to be able to inspect the internal state of the oracle (that is, the database) to gain partial, approximate knowledge about the queries made by the algorithm.
However, there are two important caveats in the permutation case.

First, Zhandry's compressed oracle satisfies the extremely useful property that the compressed oracle stores (superpositions of) input-output pairs.
This means that in order to learn something about the value~$H(x)$ of the random function~$H$ at some point~$x\in[N]$, and whether this value is known to the adversary, it is only necessary to measure one register.
In contrast, we represent permutations as a product of transpositions and hence the database of our permutation oracle stores the analogous information in a less localized fashion.
If we want to determine the value~$\pi(x)$ of the random permutation~$\pi$ at some point~$x\in[N]$, in general we may need to inspect all registers~$D_k$ for~$k\geq x$.
On the other hand, recall that for any fixed~$x$, typically only~$\tilde O(1)$ permutations are active and suffice to determine~$\pi(x)$ (see \cref{lem:active-makes-sense,lem:active forward,lem:active inverse}).

Second, Zhandry's compressed oracle has the desirable feature that %computational basis measurements on the database registers should almost commute with the compression unitary.
%Roughly speaking, this allows
one can ``jointly measure'' whether a query algorithm has accessed a register \emph{and} what function value it holds, with only a small error.
% In particular, if a register has not been access then a computational basis measurement of the oracle register should yield outcome~$\bot$ up to a small error.
In our permutation oracle the analogous procedure is to apply the binary measurement~$\mathcal M_k := \{I-\proj{+_k}_{D_k}, \proj{+_k}_{D_k}\}$ to learn whether the $k$-th transposition has been accessed and, if so, measure in the computational basis to learn what its value is.
However, since the size of the support of the uniform superposition~$\ket{+_k}_{D_k}$ depends on~$k$, the error in this ``joint measurement'' depends on~$k$.
For example, suppose an algorithm managed through some combination of queries and measurements to learn~$t_k$ with certainty for some particular~$k\in[N]$.
Then, conditional on this event, the database register~$D_k$ is in state~$\ket{t_k}$.
Applying the measurement~$\mathcal M_k$ in this state will, however, return outcome~$\ket{+_k}_{D_k}$ with probability~$\abs{\braket{+_k|t_k}}^2 = \frac1k$, which need not be small!

Both challenges can be addressed by pre- and post-composing, or ``twirling'', the \SPO with two random permutations.
This yields our final construction, which we call the \emph{twirled superposition permutation oracle}.
For convenience we first define a version where the twirls are fixed.

\begin{defn}[Twirled superposition permutation oracle]
For any two fixed permutations~$\sigma,\tau\in S_N$, the \emph{$(\sigma,\tau)$-twirled superposition permutation oracle} ($\FTSPOsigmatau$) is defined as follows:
\begin{itemize}
\item The state space consists of the same \emph{database}~$D$ as the superposition permutation oracle.

\item The initialization routine is the same as for the superposition permutation oracle and hence we will continue to denote it by~$\Init^\SPO_D$.
% ~$\Init^\SPO_D$
% As this is independent of the choice of~$\sigma$ and~$\tau$, we suppress the superscripts in the notation.

\item There are two query unitaries,~$\OFTSPOsigmatau_{XYD}$ and~$\OinvFTSPOsigmatau_{XYD}$, that define the two interfaces available to the query algorithm and that simulate oracles for a random permutation and its inverse. They are defined as follows. For all $x,y\in[N]$ and $\pi\in S_N$
\begin{align*}
  \OFTSPOsigmatau_{XYD} \ket{x,y,\pi}_{XYD}
&= \ket{x,y \op \tau^{-1}(\pi(\sigma(x))),\pi}_{XYD}, \\
  \OinvFTSPOsigmatau_{XYD} \ket{x,y,\pi}_{XYD}
&= \ket{x,y \op \sigma^{-1}(\pi^{-1}(\tau(x))),\pi}_{XYD}.
\end{align*}
where $X$ and $Y$ are the $N$-dimensional target register of the oracles.

\item The recovery routine $\Rec^{\FTSPOsigmatau}_D$ first applies the recovery routine~$\Rec^\SPO$ to obtain a permutation~$\pi\in S_N$, and then returns~$\tau^{-1} \pi \sigma$.
\end{itemize}
% The \emph{twirled superposition permutation oracle} (\FTSPO) is defined as above, but for permutations~$\sigma,\tau\in S_N$ chosen independently uniformly at random.
\end{defn}
When $\mathcal A$ is a query algorithm that gets query access to two oracles acting on two $N$-dimensional register~$X$ and $Y$, and if~$D$ is the database register of a superposition permutation oracle, we write~$\mathcal A^{\FTSPOsigmatau_D}$ to indicate that 
%before $\mathcal A$'s runtime, $\Init^\SPO_D$ is executed \MW{$\leftarrow$ This was added but does not match the current text. I would propose removing it again.}\gm{Fine for me.}\cm{I am confused. What is it that doesn't match?}, and 
we use the interfaces~$\OFTSPOsigmatau_{XYD}$ and~$\OinvFTSPOsigmatau_{XYD}$, respectively, to implement the two types of oracles queries.
% \MW{Define the coherent version where the randomly sampled $\sigma,\tau$ are picked from an auxiliary database register (for~6.3), add as item~(iv) to the lemma below, and also define the in-place version in 4.3.}

It is easy to see that the twirled superposition permutation oracles exactly simulate the ordinary one (whether the twirls are fixed or randomly sampled).
Hence it also exactly simulates a quantum-accessible random permutation.

\begin{lem}\label{lem:spo vs ftspo}
Let \adver be a query algorithm that gets query access to two oracles that act on two $N$-dimensional registers, and let~$\sigma_0,\tau_0\in S_N$ be fixed permutations.
Then the joint classical-quantum state of the random variable~$\pi$ and the register~$B$ is the same for the following three experiments:
\begin{enumerate}
\item\label{it:spo vs ftspo:spo}  Run~$\Init^\SPO_D$, then $B \leftarrow \adver^{\SPO_D}$, and finally $\pi \leftarrow \Rec^\SPO_D$.
\item\label{it:spo vs ftspo:ftspo fixed} Run~$\Init^\SPO_D$, then $B \leftarrow \adver^{\FTSPOsigmaztauz_D}$, and finally $\pi \leftarrow \Rec^{\FTSPOsigmaztauz}_D$.
\item\label{it:spo vs ftspo:ftspo random} Sample $\sigma \leftarrow S_N$ and $\tau \leftarrow S_N$, run~$\Init^\SPO_D$, then $B \leftarrow \adver^{\FTSPOsigmatau_D}$, and finally $\pi \leftarrow \Rec^{\FTSPOsigmatau}_D$.
\end{enumerate}
Moreover, in part~\ref{it:spo vs ftspo:ftspo random}, $\sigma$, $\tau$, and $(\pi,B)$ are independent, and the three permutations~$\sigma, \tau, \pi$ are independent and uniformly distributed.
\end{lem}
\begin{proof}
    It suffices to argue that~\ref{it:spo vs ftspo:spo} and~\ref{it:spo vs ftspo:ftspo fixed} result in the same joint state.
    By using \cref{prop:randompi vs spo} twice, we see that it suffices to compare the following two experiments:
    \begin{enumerate}
        \item[(i')] Sample $\pi\leftarrow S_N$ uniformly at random and run $B \leftarrow \adver^{U^\pi, U^{\pi^{-1}}}$.
        \item[(ii')] Sample $\pi\leftarrow S_N$ uniformly at random, run $B \leftarrow \adver^{U^{\tau_0^{-1}} U^\pi U^{\sigma_0}, U^{\sigma_0^{-1}} U^{\pi^{-1}} U^{\tau_0}} = \adver^{U^{\tau_0^{-1} \pi \sigma_0}, U^{(\tau_0^{-1} \pi \sigma_0)^{-1}}}$, and update $\pi\leftarrow\sigma_0^{-1}\pi\sigma_0$.
    \end{enumerate}
    These are indeed equivalent, since if~$\pi\in S_N$ is uniformly random then so is~$\tau_0^{-1} \pi \sigma_0$, for fixed~$\sigma_0,\tau_0\in S_N$.
\end{proof}

It will be convenient to relate the twirled superposition oracle to the untwirled one by viewing the twirling as an action on the database~$D$.
To this end, define the left and right actions of~$S_N$ on~$D$ as
\begin{align*}
	L^\tau \ket\pi &= \ket{\tau\pi},\\
	R^\sigma \ket\pi &= \ket{\pi\sigma^{-1}}.
\end{align*}
Then we have the following lemma, which states that the superposition oracle can be expressed in terms of the untwirled one, sandwiched by a basis change implemented by the operators~$L_\tau$ and $R_\sigma$ as defined above.

\begin{lem}\label{lem:twirled spo via spo}
For all $\sigma, \tau\in S_N$, it holds that
\begin{align*}
	\OFTSPOsigmatau_{XYD} &= \parens[\big]{ L^\tau_D R^\sigma_D } \OSPO_{XYD} \parens[\big]{ L^{\tau^{-1}}_D R^{\sigma^{-1}}_D }, \\
	\OinvFTSPOsigmatau_{XYD} &= \parens[\big]{ L^\tau_D R^\sigma_D } \OinvSPO_{XYD} \parens[\big]{ L^{\tau^{-1}}_D R^{\sigma^{-1}}_D }.
\end{align*}
\end{lem}
\begin{proof}
For all $\pi\in S_N$ and $x,y\in[N]$, we have
\begin{align*}
  \parens[\big]{ L^\tau_D R^\sigma_D } \OSPO_{XYD} \parens[\big]{ L^{\tau^{-1}}_D R^{\sigma^{-1}}_D } \ket{x,y,\pi}
&= \parens[\big]{ L^\tau_D R^\sigma_D } \OSPO_{XYD} \ket{x,y,\tau^{-1}\pi\sigma} \\
&= \parens[\big]{ L^\tau_D R^\sigma_D } \ket{x,y \op \tau^{-1}\pi\sigma(x),\tau^{-1}\pi\sigma} \\
&= \ket{x,y \op \tau^{-1}(\pi(\sigma(x))),\pi}
= \OFTSPOsigmatau_{XYD} \ket{x,y,\pi},
\end{align*}
which establishes the first equation.
The second one is proved in the same way.
\end{proof}

Clearly, these two operations commute with each other, and they leave the initial state~$\ket{\Phi_{\SPO}}_D$ of the oracle invariant:
\begin{equation}\label{eq:initial invariant}
  L_\sigma\ket{\Phi_{\SPO}}=R_\sigma\ket{\Phi_{\SPO}}_D = \ket{\Phi_{\SPO}}_D.
\end{equation}

Thus we obtain the following lemma that allows us to compare the behavior of an algorithm when using either the twirled or the ordinary oracle, strengthening \cref{lem:spo vs ftspo}.

\begin{lem}\label{lem:output state twisted vs not}
Let $\mathcal A$ be a unitary query algorithm on registers~$AXY$, where~$X$ and~$Y$ are $N$-dimensional registers, that gets query access to two oracles that each act on~$XY$.
For every~$\sigma,\tau\in S_N$, let~$\ket{\phi^{\sigma,\tau}}_{AXYD}$ be the joint state of algorithm and oracle defined by running~$\Init^\SPO_D$ and then~$\mathcal A^{\FTSPOsigmatau_D}$.
Moreover, let $\ket\phi_{AXYD}$ denote the joint state of algorithm and oracle defined by running~$\Init^\SPO_D$ and then~$\mathcal A^{\SPO_D}$.
Then,
\begin{align*}
  \ket{\phi^{\sigma,\tau}}_{AXYD} = L^\tau_D R^\sigma_D \ket{\phi}_{AXYD}.
\end{align*}
\end{lem}
\begin{proof}
Without loss of generality, the quantum query algorithm takes the form (cf.\ \cref{eq:generic oracle algo})
\begin{align*}
  U^{(q)}_{AXY} Q^{(q)}_{XYD} U^{(q-1)}_{AXY} Q^{(q-1)}_{XYD} \ldots U^{(1)}_{AXY} Q^{(1)}_{XYD} U^{(0)}_{AXY},
\end{align*}
where each $Q_{XYD}^{(j)}$ is either a forward or an inverse query, and is applied to the initial state~$\ket0_{AXY} \ot \ket{\Phi_{\SPO}}_D$.
When expressing the twirled oracles in terms of the ordinary ones using \cref{lem:twirled spo via spo}, we see that the ``twirls''~$L^\tau_D R^\sigma_D$ and~$L^{\tau^{-1}}_D R^{\sigma^{-1}}_D$ in-between any pair of queries cancel (note that they commute with the unitaries~$U^{(j)}_{AXY}$).
Moreover, by \cref{eq:initial invariant} the initial twirl leaves the initial state~$\ket{\Phi_{\SPO}}_D$ of the database invariant.
Accordingly, the output state in the two scenarios only differs by an application of $L^\tau_D R^\sigma_D$, as claimed.
\end{proof}

\Cref{lem:twirled spo via spo} simulates queries to the twirled superposition oracle (for known~$\sigma$ and~$\tau$) by a single query to the ordinary one but requires access to the database.
This can also be achieved by acting on the input/output registers~$X$ and~$Y$, but in this case more than one query is required.
The following lemma shows that an algorithm can always be converted into a ``standard form'' such that an analogous replacement is possible, at the cost of doubling the number of queries, which will be useful in \cref{sec:progress}.

\begin{lem}\label{lem:std preprocess}
Let $\mathcal A$ be a unitary query algorithm on registers~$AXY$, where~$X$ and~$Y$ are $N$-dimensional registers, that gets query access to two oracles that each act on~$XY$.
Then there exists a unitary query algorithm $\mathcal C$ on registers~$BXY$, with $B=AZ$ and $Z$ an $N$-dimensional register, that gets query access to \emph{four} oracles acting on~$XY$ such that, for every~$\sigma, \tau\in S_N$ (and for any initial state of the database register~$D$), the following three experiments result in the same state of the registers~$BXYD$:
\begin{enumerate}
\item Run $\mathcal A^{\FTSPOsigmatau_D}$ and initialize the register~$Z$ in state~$\ket0_Z$.
% = \mathcal A^{\OFTSPOsigmatau_{XYD}, \OinvFTSPOsigmatau_{XYD}}
\item Run $\mathcal B^{\FTSPOsigmatau_D}$, where the query algorithm $\mathcal B$ gets access to two oracles on~$XY$ and is defined as follows:
\begin{align*}
  \mathcal B^{O_{XY}, O_{XY}^{\mathrm{inv}}}
:= \mathcal C^{O_{XY},\; O_{XY},\; O_{XY}^{\mathrm{inv}},\; O_{XY}^{\mathrm{inv}}}.
\end{align*}
\item Run $\mathcal B_{\sigma,\tau}^{\SPO_D}$,
% = \mathcal B_{\sigma,\tau}^{\OSPO_{XYD}, \; \OinvSPO_{XYD}},
where the query algorithm $\mathcal B_{\sigma,\tau}$ gets access to two oracles on~$XY$ and is defined as follows:
\begin{align*}
   \mathcal B_{\sigma,\tau}^{O_{XY}, O_{XY}^{\mathrm{inv}}}
:= \mathcal C^{(V_X^{\sigma^{-1}} V_Y^{\tau^{-1}} O_{XY} V_X^\sigma), \; (V_X^{\sigma^{-1}} O_{XY} V_X^\sigma V_Y^\tau), \; (V_X^{\tau^{-1}} V_Y^{\sigma^{-1}} O_{XY}^{\mathrm{inv}} V_X^\tau), \; (V_X^{\tau^{-1}} O_{XY}^{\mathrm{inv}} V_X^\tau V_Y^\sigma)}
\end{align*}
\end{enumerate}
Moreover, if $\mathcal A$ makes in total $q$ oracle queries then~$\mathcal C$ (and hence~$\mathcal B$ and~$\mathcal B_{\sigma,\tau}$) makes in total~$2q$ oracle queries.
\end{lem}
\noindent
\begin{proof}
Note that
\begin{align*}
  \OFTSPOsigmatau_{XYD} \ket0_Z
&= \SWAP_{Y \leftrightarrow Z} \OFTSPOsigmatau_{XYD} \CNOT_{Y\to Z} \OFTSPOsigmatau_{XYD} \SWAP_{Y \leftrightarrow Z} \ket0_Z \\
&= \SWAP_{Y \leftrightarrow Z}  \parens*{ V_X^{\sigma^{-1}} \OSPO_{XYD} V_X^\sigma V_Y^\tau } \CNOT_{Y\to Z} \parens*{ V_X^{\sigma^{-1}} V_Y^{\tau^{-1}} \OSPO_{XYD} V_X^\sigma } \SWAP_{Y \leftrightarrow Z} \ket0_Z
\end{align*}
and
\begin{align*}
   \OinvFTSPOsigmatau_{XYD} \ket0_Z
&= \SWAP_{Y \leftrightarrow Z} \OinvFTSPOsigmatau_{XYD} \CNOT_{Y\to Z} \OinvFTSPOsigmatau_{XYD} \SWAP_{Y \leftrightarrow Z} \ket0_Z \\
&= \SWAP_{Y \leftrightarrow Z}  \parens*{ V_X^{\tau^{-1}} \OinvSPO_{XYD} V_X^\tau V_Y^\sigma } \CNOT_{Y\to Z} \parens*{ V_X^{\tau^{-1}} V_Y^{\sigma^{-1}} \OinvSPO_{XYD} V_X^\tau } \SWAP_{Y \leftrightarrow Z} \ket0_Z
\end{align*}
Thus we see that if we define the query algorithm $\mathcal C$ as follows,
\begin{align*}
&\mathcal C^{O_{XY}, O'_{XY}, O_{XY}^{\mathrm{inv}}, {O'}_{XY}^{\mathrm{inv}}} := \\
&\qquad \mathcal A^{(\SWAP_{Y \leftrightarrow Z} O'_{XY} \CNOT_{Y\to Z} O_{XY} \SWAP_{Y \leftrightarrow Z}), \; (\SWAP_{Y \leftrightarrow Z} {O'}_{XY}^{\mathrm{inv}} \CNOT_{Y\to Z} O_{XY}^{\mathrm{inv}} \SWAP_{Y \leftrightarrow Z})} \ket0_Z,
\end{align*}
then the claim follows.
\end{proof}

% !TEX root = compressed-pi.tex
%=============================================================================
\section{The Fundamental Lemma of the Permutation Oracle}\label{sec:fundamental}
%=============================================================================

We know from the preceding section that the superposition permutation oracles exactly simulate a quantum-accessible permutation, with the permutation being obtained by measuring the database in the computational basis.
However, to learn about queries made by the adversary, we wish to also measure whether database registers are in the uniform superposition states.
The following result, which we call the \emph{Fundamental Lemma}, shows that this only slightly changes the statistics.
It resembles \cite[Corollary~4.2]{CFHL21} which goes back to Zhandry, but our result applies to random permutations rather than random functions.
We state and prove it for arbitrary relations involving a single input-output pair.

\begin{lem}[Fundamental Lemma of the Permutation Oracle]\label{lem:fundamental}
Let $R \subseteq [N] \times [N]$ be a relation.
Let~\adver be a quantum algorithm that gets query access to two oracles that each act on~$\C^N \ot \C^N$, and which returns a pair~$(x,y) \in [N] \times [N]$.
We consider the following two experiments:
\begin{enumerate}
\item Sample~$\pi \leftarrow S_N$ uniformly at random, and run~$(x,y) \leftarrow \mathcal A^{U^\pi,U^{\pi^{-1}}}$. \\
If~$\pi(x) = y$ and $(x,y) \in R$ then return 1. Otherwise return 0.
\item Sample~$\sigma \leftarrow S_N$ and $\tau\leftarrow S_N$, run~$\Init^\SPO_D$, and then~$(x,y) \leftarrow \mathcal A^{\FTSPOsigmatau_D}$.
Next, apply to register~$D_{\sigma(x)}$ the projective measurement $\{\proj{+_{\sigma(x)}}, I - \proj{+_{\sigma(x)}}\}$.
If the second outcome is observed, run~$\pi \leftarrow \Rec^{\FTSPOsigmatau}_D$.
If~$\pi(x) = y$ and $(x,y) \in R$, return~1.
In all other cases, return~0.
\end{enumerate}
Let $p_\text{(i)}$, $p_\text{(ii)}$ denote the probability that first or second experiment returns~1, respectively.
Then:
\begin{align*}
	\sqrt{p_\text{(i)}} \leq \sqrt{p_\text{(ii)}} + \sqrt{\frac{\ln(N) + 1}N}.
\end{align*}
\end{lem}
Thus, if we want to upper bound the probability that the algorithm learned a pair~$(x,y)$ such that~$\pi(x)=y$ satisfying some relation~$R$, then we can just imagine first measuring whether $D_{\sigma(x)}$ is not~$\ket{+_{\sigma(x)}}$, without significantly increasing the error -- we will see that this typically yields a quantity that is easier to upper bound.
This bound is essentially identical to the one known for random functions, except for the extra term $\ln(N)+1$, which is due to the varying dimensions of the database registers.

The idea of the proof is to recall that, by \cref{prop:randompi vs spo,lem:spo vs ftspo}, the first experiment is exactly simulated by the following:
\begin{enumerate}
	\item[{\crtcrossreflabel{(i')}[it:real prime]}] Sample~$\sigma \leftarrow S_N$ and~$\tau\leftarrow S_N$, run~$\Init^\SPO_D$, and then~$(x,y) \leftarrow \mathcal A^{\FTSPOsigmatau_D}$.
	Measure the entire database in the computational basis and interpret the outcome as a permutation~$\pi \in S_N$.
	If $\pi(\sigma(x)) = \tau(y)$ and $(x,y) \in R$, then return 1. Otherwise return 0.
\end{enumerate}
For comparison, expanding the definition of the recovery routine~$\Rec^{\FTSPOsigmatau}_D$, the second experiment can be written as follows:
\begin{enumerate}
    \item[{\crtcrossreflabel{(ii')}[it:forensic prime]}] Sample~$\sigma \leftarrow S_N$ and~$\tau\leftarrow S_N$, run~$\Init^\SPO_D$, and then~$(x,y) \leftarrow \mathcal A^{\FTSPOsigmatau_D}$.
    Next, apply to register~$D_{\sigma(x)}$ the projective measurement $\{\proj{+_{\sigma(x)}}, I - \proj{+_{\sigma(x)}}\}$.
    If the second outcome is observed, measure the entire database in the computational basis and interpret the outcome as a permutation~$\pi \in S_N$.
    If $\pi(\sigma(x)) = \tau(y)$ and $(x,y) \in R$, then return 1. In all other cases, return 0.
\end{enumerate}
Note that~\ref{it:real prime} and \ref{it:forensic prime} only differ in that the latter contains the additional measurement of the register~$D_{\sigma(x)}$ and subsequent check that the desired (second) outcome occurred.
To prove the fundamental lemma, we therefore need to argue that this ``postselection'' does not impact the probability of acceptance much.
We first state and prove a technical lemma that contains the core argument, and then use it establish \cref{lem:fundamental}.

\begin{lem}\label{lem:help}
% Let~$\ket\Delta \in \mathcal H_D$ be an arbitrary pure state of the database (which should be thought of as ``compressed''), and let~$x,y\in[N]$.
Let $x \in [N]$ and $Y \subseteq [N]$.
Then it holds that:
\begin{align*}
	 \norm*{ \sum_{\pi \in S_N : \pi(x) \in Y} \proj\pi_D - \sum_{\pi \in S_N : \pi(x) \in Y} \proj\pi_D \parens*{ I - \proj{+_x}_{D_x} } }
= \norm*{ \sum_{\pi \in S_N : \pi(x) \in Y} \proj\pi_D \ket{+_x}_{D_x} }
\leq \sqrt{\frac{\abs Y}x}.
\end{align*}
% where we recall that by convention the operator~$\proj{+_x}_{D_x}$ acts as the identity operator on all registers but~$D_x$.
\end{lem}
\begin{proof}
The first equality is clear.
% We start by computing
% \begin{align*}
% &\quad \norm*{ \sum_{\pi \in S_N : \pi(x) \in Y} \proj\pi_D - \sum_{\pi \in S_N : \pi(x) \in Y} \proj\pi_D \parens*{ I - \proj{+_x}_{D_x} } }
% = \norm*{ \sum_{\pi \in S_N : \pi(x) \in Y} \proj\pi_D \proj{+_x}_{D_x} } \\
% &= \norm*{ \sum_{\pi \in S_N : \pi(x) \in Y} \proj\pi_D \ket{+_x}_{D_x} }.
% \end{align*}
Now, for any operator~$X$ the operator norm can be computed as~$\norm X = \max_{\norm\phi = 1} \norm{X\ket\phi}$.
Thus there exists a vector~$\ket\Delta \in \mathcal H_{D_{x^c}} = \bigotimes_{k=1 : k \neq x}^n \mathcal H_{D_k}$ such that
\begin{align*}
	\norm*{ \sum_{\pi \in S_N : \pi(x) \in Y} \proj\pi_D \ket{+_x}_{D_x} }^2
&= \norm*{ \sum_{\pi \in S_N : \pi(x) \in Y} \proj\pi_D \parens*{ \ket{+_x}_{D_x} \ot \ket\Delta_{D_{x^c}} } }^2 \\
&= \sum_{\pi \in S_N : \pi(x) \in Y} \abs*{ \bra\pi_D \parens*{ \ket{+_x}_{D_x} \ot \ket\Delta_{D_{x^c}} } }^2.
\end{align*}
This can be upper bounded by
\begin{align*}
	\Pr\bigl( \pi(x) \in Y \mid t_x \leftarrow [x], \; (t_1,\dots,t_{x-1},t_{x+1},\dots,t_N) \leftarrow Q \bigr),
\end{align*}
where the permutation~$\pi \in S_N$ is defined in terms of the numbers~$t_1,\dots,t_N$ via \cref{eq:tower-decomposition}, and where~$Q$ is some probability distribution on $\prod_{k=1 : k \neq x}^n [k]$ (namely the distribution obtained by measuring~$\ket\Delta_{D_{x^c}}$ in the standard basis).
Now,
\begin{align*}
	\pi(x) \in Y
\quad\Leftrightarrow\quad \tp N {t_N} \tp {N\!-\!1} {t_{N-1}} \cdots \tp 2 {t_2} \tp 1 {t_1} (x) \in Y
\quad\Leftrightarrow\quad \pi_{>x} (t_x) \in Y
\quad\Leftrightarrow\quad t_x \in \pi_{>x}^{-1}(Y),
\end{align*}
where we recall the notation~$\pi_{>x} = \tp N {t_N} \tp {N\!-\!1} {t_{N-1}} \cdots \tp {x+1} {t_{x+1}}$.
Thus, whatever~$\pi_{>x}$, there are at most~$\abs Y$ choices of~$t_x \in [x]$ such that~$\pi(x) \in Y$.
Since~$\pi_{>x}$ and $t_x$ are independent and the latter is chosen uniformly at random in~$[x]$, we have
\begin{align*}
	\Pr\bigl( \pi(x) \in Y \;\big|\; t_x \leftarrow [x], \; (t_1,\dots,t_{x-1},t_{x+1},\dots,t_N) \leftarrow Q \bigr) \leq \frac{\abs Y}x.
\end{align*}
This concludes the proof.
\end{proof}

We now prove the fundamental lemma.

\begin{proof}[Proof of \cref{lem:fundamental}]
As discussed it suffices to compare the two experiments~\ref{it:real prime} and~\ref{it:forensic prime}.
Let $p(\sigma,\tau,x,y)$ denote the joint distribution of the uniformly random choices of~$\sigma,\tau\in S_N$ and the output~$(x,y)$ of~$\mathcal A^{\FTSPOsigmatau_D}$ and choose, for each $\sigma,\tau,x,y$, a purification~$\ket{\Delta(\sigma,\tau,x,y)}_{DE}$ of the corresponding state of the database.
Then:
\begin{align*}
    p_\text{(i')} = \hspace{-0.5cm} \sum_{\sigma,\tau \in S_N, (x,y) \in R} \hspace{-0.5cm} p(\sigma,\tau,x,y) \, p_\text{(i')}(\sigma,\tau,x,y)
\qquad\text{and}\qquad
    p_\text{(ii')} = \hspace{-0.5cm} \sum_{\sigma,\tau \in S_N, (x,y) \in R} \hspace{-0.5cm} p(\sigma,\tau,x,y) \, p_\text{(ii')}(\sigma,\tau,x,y),
\end{align*}
where
\begin{align*}
    p_\text{(i')}(\sigma,\tau,x,y) &:= \norm*{\sum_{\pi \in S_N : \pi(\sigma(x)) = \tau(y)} \hspace{-0.7cm} \proj\pi_D \ket{\Delta(\sigma,\tau,x,y)}_{DE} }^2, \\
    p_\text{(ii')}(\sigma,\tau,x,y) &:= \norm*{\sum_{\pi \in S_N : \pi(\sigma(x)) = \tau(y)} \hspace{-0.7cm} \proj\pi_D \parens*{ I - \proj{+_{\sigma(x)}}_{D_{\sigma(x)}} } \ket{\Delta(\sigma,\tau,x,y)}_{DE} }^2.
\end{align*}
% \begin{align*}
%   p_\text{(i')}
% &= \hspace{-0.5cm} \sum_{\sigma,\tau \in S_N, (x,y) \in R} \hspace{-0.5cm} p(\sigma,x,y) \, p_\text{(i')}(\sigma,\tau,x,y),
% \quad p_\text{(i')}(\sigma,\tau,x,y) := \norm*{\sum_{\pi \in S_N : \pi(\sigma(x)) = \tau(y)} \hspace{-0.7cm} \proj\pi_D \ket{\Delta(\sigma,\tau,x,y)}_{DE} }^2, \\
%   p_\text{(ii')}
% &= \hspace{-0.5cm} \sum_{\sigma,\tau \in S_N, (x,y) \in R} \hspace{-0.5cm} p(\sigma,x,y) \, p_\text{(ii')}(\sigma,\tau,x,y),
% \quad p_\text{(ii')}(\sigma,\tau,x,y) := \norm*{\sum_{\pi \in S_N : \pi(\sigma(x)) = \tau(y)} \hspace{-0.7cm} \proj\pi_D \parens*{ I - \proj{+_{\sigma(x)}}_{D_{\sigma(x)}} } \ket{\Delta(\sigma,\tau,x,y)}_{DE} }^2.
% \end{align*}
Using \cref{lem:help} (with $Y = \{\tau(y)\}$) and the Cauchy-Schwarz inequality, it follows that
\begin{align*}
	p_\text{(i')}
&\leq \sum_{\sigma,\tau \in S_N, (x,y) \in R} p(\sigma,\tau,x,y)
\parens*{
	\sqrt{ p_\text{(ii')}(\sigma,\tau,x,y) }
+ \sqrt{\frac1{\sigma(x)}}
}^2 \\
&= p_\text{(ii')}
+ 2 \sum_{\sigma,\tau \in S_N, (x,y) \in R} p(\sigma,\tau,x,y)
    \sqrt{ p_\text{(ii')}(\sigma,\tau,x,y) }
    \sqrt{ \frac1{\sigma(x)} }
  + \sum_{\sigma,\tau \in S_N, (x,y) \in R} p(\sigma,\tau,x,y) \frac1{\sigma(x)} \\
&\leq p_\text{(ii')}
+ 2 \sqrt{p_\text{(ii')}} \sqrt{ \sum_{\sigma,\tau \in S_N, (x,y) \in R} p(\sigma,\tau,x,y) \frac1{\sigma(x)} }
  + \sum_{\sigma,\tau \in S_N, (x,y) \in R} p(\sigma,\tau,x,y) \frac1{\sigma(x)}.
\end{align*}
Thus,
\begin{align*}
  \sqrt{p_\text{(i')}}
\leq \sqrt{p_\text{(ii')}} + \sqrt{\sum_{\sigma,\tau \in S_N, (x,y) \in R} p(\sigma,\tau,x,y) \frac1{\sigma(x)}}.
% \leq \sqrt{p_\text{(ii')}} + \sqrt{\sum_{\sigma \in S_N, (x, y) \in [N]^2} p(\sigma,x,y) \frac1{\sigma(x)}}.
\end{align*}
Finally, we note that as a consequence of \cref{lem:spo vs ftspo}, $\sigma$, $\tau$, and $(x,y)$ are independent with respect to the distribution~$p(\sigma,\tau,x,y)$.
Thus,
\begin{align*}
  \sum_{\sigma,\tau \in S_N, (x,y) \in R} p(\sigma,\tau,x,y) \frac1{\sigma(x)}
\leq \frac1 N \sum_{k=1}^N \frac 1 k
\leq \frac {\ln(N) + 1} N.
\end{align*}
\end{proof}

%\MW{Added:}\cm{Checked this, looks good.}
In order to apply the fundamental lemma, it is useful to upper bound the probability~$p_\text{(ii)}$ in way that only refers to the state of the database.
This is achieved by the following lemma.

\begin{lem}\label{lem:fund lem rhs upper bound}
Let $R \subseteq [N] \times [N]$ be a relation.
Let~\adver be a quantum algorithm that gets query access to two oracles that each act on~$\C^N \ot \C^N$, and which returns a pair~$(x,y) \in [N] \times [N]$.
For~$\sigma,\tau\in S_N$, let $\ket{\phi^{\sigma,\tau}}$ denote a purification of the state of the database after running $\Init^\SPO_D$ and then~$\mathcal A^{\FTSPOsigmatau_D}$.
Then, the quantity~$p_\text{(ii)}$ in \cref{lem:fundamental} can be upper bounded as
\begin{align*}
    p_\text{(ii)}
\leq \E_{\sigma, \tau \leftarrow S_N}
\sum_{(x,y) \in R}
\sum_{\substack{\pi \in S_N \text{ s.th.}\\ \tau^{-1}(\pi(\sigma(x))) = y}}
\norm*{ \bra\pi_D \parens*{ I - \proj{+_{\sigma(x)}}_{D_{\sigma(x)}} }
\ket{\phi^{\sigma,\tau}} }^2.
\end{align*}
\end{lem}
\begin{proof}
Without loss of generality we can assume that $\mathcal A$ is a unitary query algorithm on registers $AXY$ such that the classical outcomes~$x$ and $y$ can be obtained by measuring the~$X$ and~$Y$ registers.
For every~$\sigma,\tau\in S_N$, let~$\ket{\phi^{\sigma,\tau}}_{AXYD}$ be the joint state of algorithm and oracle defined by running~$\Init^\SPO_D$ and then~$\mathcal A^{\FTSPOsigmatau_D}$.
Then:
\begin{align*}
    p_{\text{(ii)}}
&=
\E_{\sigma, \tau \leftarrow S_N}
\sum_{(x,y) \in R}
\sum_{\substack{\pi \in S_N \text{ s.th.}\\ \tau^{-1}(\pi(\sigma(x))) = y}}
\norm*{ \bra\pi_D \parens*{ I - \proj{+_{\sigma(x)}}_{D_{\sigma(x)}} }
\bra{xy}_{XY} \ket{\phi^{\sigma,\tau}}_{AXYD}
}^2.
\end{align*}
Because the projection $\bra{xy}$ commutes with the operators on $D$ and never increases the norm, we can upper bound the above as
\begin{align*}
    p_{\text{(ii)}}
&\leq \E_{\sigma, \tau \leftarrow S_N}
\sum_{(x,y) \in R}
\sum_{\substack{\pi \in S_N \text{ s.th.}\\ \tau^{-1}(\pi(\sigma(x))) = y}}
\norm*{ \bra\pi_D \parens*{ I - \proj{+_{\sigma(x)}}_{D_{\sigma(x)}} }
\ket{\phi^{\sigma,\tau}}_{AXYD}
}^2.
\end{align*}
Since this expression only depends on the reduced state on~$D$, the claim follows.
\end{proof}

% !TEX root = compressed-pi.tex
%=============================================================================
\section{Bounding the Success Probability for Search}\label{sec:progress}
%=============================================================================
%The compressed oracle for random functions developed in~\cite{Zhandry2019} has turned out to be a powerful tool for %proving the hardness of oracle search problems such as, e.g., pre-image finding or collision finding.
%The first theorems of this type have been proven in \cite{Zhandry2019}, followed by an entire line of work focusing on them.

In this section we generalize Zhandry's compressed oracle technique \cite{Zhandry2019} to the case of random permutations, but using our (twirled) permutation oracle. The high-level strategy for proving such theorems is as follows:
For the compressed oracle for a random function~$H$, it can be approximately determined whether a certain input~$x$ has been queried, and if so whether~$x$, together with the corresponding function output~$y=H(x)$, fulfils a certain relation.
This test is performed by a measurement acting on the database register~$D_x$ only.
For different inputs~$x$, these tests commute, and hence there exists a projective measurement answering the question whether there \emph{exists} an input $x$ such that the described test would trigger.
The probability that this measurement outputs yes after~$q$ queries then serves as a convenient \emph{progress measure}.
This is reminiscent to the progress measures used in the so-called ``hybrid method'' that was introduced earlier to prove the query lower bound for the unstructured search problem~\cite{BBBV}.

We would like to generalize this technique to the case of random permutations.
We follow a similar strategy, and begin by devising a generalization of the test for a single input.
To test whether a certain input~$x$ has been queried in forward direction or output in an inverse query, and if so, whether, together with the corresponding function output~$y=\pi(x)$, it fulfils a certain relation $R\subseteq [N]\times [N]$, we define the following operators:
\begin{align*}
    \Pi^{R,x}_D &:= \sum_{\pi \in S_N : \pi(x) \in R_x} \proj\pi_D,&
    E^{R,x}_D &:= \Pi^{R,x}_D (I - \proj{+_x}_{D_x}).&
\end{align*}
where we have defined $R_x := \{ y \in [N] : (x,y) \in R \}$.
For later use, we also set~$R^\mathrm{inv}_y := \{ x \in [N] : (x,y) \in R \}$ and
\begin{align*}
    r_{\max} := \max \braces*{ \max_{x\in[N]} \, \abs{R_x}, \max_{y\in[N]} \, \abs{R^\mathrm{inv}_y} }.
\end{align*}

Unfortunately, the operators~$E^{R,x}$ for different~$x$ do not commute, so we cannot simply use these to construct a measurement answering the existence question.
Instead, we will check whether a \emph{random} $x$ has this property.
To lift the worst-case bounds established in the previous section to the average case, we further consider running the query algorithm with the \emph{twirled} permutation oracle for uniformly random~$\sigma,\tau \in S_N$.
Because the permutation in the database is now twirled as compared to the action of the oracle, we must also consider the twirled relation~$R^{\sigma,\tau}$ that is defined as follows in terms of~$R$:
\begin{align}\label{eq:twirled relation}
    (x,y) \in R^{\sigma,\tau} :\!\!\iff (\sigma^{-1}(x),\tau^{-1}(y)) \in R.
\end{align}
Thus we are led to consider the following natural progress measure:
\begin{align*}
    \E_{\substack{x \leftarrow [N], \\ \sigma, \tau \leftarrow S_N}} \norm*{ E^{R^{\sigma,\tau},x}_D \ket{\phi^{\sigma,\tau}} }^2,
\end{align*}
where $\ket{\phi^{\sigma,\tau}}$ denotes the joint state of (a unitary realization of) the algorithm and database obtained by running the algorithm with the twirled permutation oracle $\FTSPOsigmatau$.
Remarkably, this progress measure not only has an intuitive operational interpretation, but it is also directly related to the upper bound furnished by the fundamental lemma (cf.\ \cref{lem:fund lem rhs upper bound}):

% In order to apply the fundamental lemma, it is useful to upper bound the probability~$p_\text{(ii)}$ in way that that makes no reference to the algorithm.
% This is achieved by the following result.

\begin{lem}\label{lem:progress vs fundamental}
For any relation $R \subseteq [N] \times [N]$, it holds that
\begin{align*}
    N \E_{\substack{x \leftarrow [N], \\ \sigma, \tau \leftarrow S_N}} \norm*{ E^{R^{\sigma,\tau},x}_D \ket{\phi^{\sigma,\tau}} }^2
= \E_{\sigma, \tau \leftarrow S_N}
\sum_{(x,y) \in R}
\sum_{\substack{\pi \in S_N \text{ s.th.}\\ \tau^{-1}(\pi(\sigma(x))) = y}}
\norm*{ \bra\pi_D \parens*{ I - \proj{+_{\sigma(x)}}_{D_{\sigma(x)}} }
\ket{\phi^{\sigma,\tau}} }^2.
\end{align*}
\end{lem}
\begin{proof}
We calculate:
% This follows from \cref{lem:fund lem rhs upper bound}, since
\begin{align*}
&\quad \E_{\sigma, \tau \leftarrow S_N}
\sum_{(x,y) \in R}
\sum_{\substack{\pi \in S_N \text{ s.th.}\\ \tau^{-1}(\pi(\sigma(x))) = y}}
\norm*{ \bra\pi_D \parens*{ I - \proj{+_{\sigma(x)}}_{D_{\sigma(x)}} }
\ket{\phi^{\sigma,\tau}}
}^2 \\
% &= \E_{\sigma, \tau \leftarrow S_N}
% \sum_{x=1}^N
% \sum_{\substack{\pi \in S_N \text{ s.th.}\\ (x, \tau^{-1}(\pi(\sigma(x)))) \in R}}
% \norm*{ \bra\pi_D \parens*{ I - \proj{+_{\sigma(x)}}_{D_{\sigma(x)}} }
% \ket{\phi^{\sigma,\tau}}
% }^2 \\
% &= \E_{\sigma, \tau \leftarrow S_N}
% \sum_{x=1}^N
% \sum_{\substack{\pi \in S_N \text{ s.th.}\\ (\sigma^{-1}(\sigma(x)), \tau^{-1}(\pi(\sigma(x)))) \in R}}
% \norm*{ \bra\pi_D \parens*{ I - \proj{+_{\sigma(x)}}_{D_{\sigma(x)}} }
% \ket{\phi^{\sigma,\tau}}
% }^2 \\
&= \E_{\sigma, \tau \leftarrow S_N}
\sum_{x=1}^N
\sum_{\substack{\pi \in S_N \text{ s.th.}\\ (\sigma(x), \pi(\sigma(x))) \in R^{\sigma,\tau}}}
\norm*{ \bra\pi_D \parens*{ I - \proj{+_{\sigma(x)}}_{D_{\sigma(x)}} }
\ket{\phi^{\sigma,\tau}}
}^2 \\
&= \E_{\sigma, \tau \leftarrow S_N}
\sum_{x=1}^N
\sum_{\substack{\pi \in S_N \text{ s.th.}\\ (x, \pi(x)) \in R^{\sigma,\tau}}}
\norm*{ \bra\pi_D \parens*{ I - \proj{+_x}_{D_x} }
\ket{\phi^{\sigma,\tau}}
}^2 \\
% &= \E_{\sigma, \tau \leftarrow S_N}
% \sum_{x=1}^N
% \norm*{ \sum_{\substack{\pi \in S_N \text{ s.th.}\\ (x, \pi(x)) \in R^{\sigma,\tau}}} \proj\pi_D \parens*{ I - \proj{+_x}_{D_x} }
% \ket{\phi^{\sigma,\tau}}
% }^2 \\
% &= \E_{\sigma, \tau \leftarrow S_N}
% \sum_{x=1}^N
% \norm*{ E^{R^{\sigma,\tau},x}_D \ket{\phi^{\sigma,\tau}} }^2 \\
&= N \E_{\substack{x \leftarrow [N], \\ \sigma, \tau \leftarrow S_N}} \norm*{ E^{R^{\sigma,\tau},x}_D \ket{\phi^{\sigma,\tau}} }^2.
\qedhere
\end{align*}
\end{proof}

We can now summarize the plan for the remainder of this section:
% In more details, our proof strategy can be summarized in the following steps:
\begin{enumerate}
    \item We bound the effect of a single query on any fixed location $x$ of the database in \cref{sec:worst_case}.
    \item We derive a bound on the success probability of the adversary in \cref{sec:expectation}, but in expectation over the random choice of the database location and the choice of the twirling.
    \item We bound the (weighted) average probability that a database register is no longer in the uniform state, which is necessary to complete the proof, in \cref{sec:bounc_active}.
    \item We combine (ii) and (iii) with the fundamental lemma to obtain our main theorem in \cref{sec:main_thm}.
\end{enumerate}

%-----------------------------------------------------------------------------
\subsection{Bounding the Success Probability in the Worst Case}\label{sec:worst_case}
%-----------------------------------------------------------------------------
%In this subsection and the next, it will be convenient to work with the standard variant of the oracles (see \cref{subsec:std vs inplace}).
We start with a useful lemma that follows readily from \cref{lem:help}.

\begin{lem}\label{lem:easy}
Let $ Q^{\SPO}_{XYD} \in \{ \OSPO_{XYD}, \OinvSPO_{XYD} \}$, and~$x\in[N]$.
Then:
\begin{enumerate}
\item $\norm[\big]{ E^{R,x}_D  Q^{\SPO}_{XYD} \parens[\big]{ I - \Pi^{R,x}_D } } \leq \sqrt{ \frac{\abs{R_x}}x }$.
\item $\norm[\big]{ E^{R,x}_D  Q^{\SPO}_{XYD} \ket\phi } - \norm[\big]{ E^{R,x}_D \ket\phi } \leq \sqrt{\frac{\abs{R_x}}x} \norm[\big]{ \parens[\big]{ I - \proj{+_x}_{D_x} } \ket\phi } + \norm[\big]{ E^{R,x}_D  Q^{\SPO}_{XYD} \proj{+_x}_{D_x} \ket\phi }$ for any pure state~$\ket\phi_{AXYD}$.
\item $\norm[\big]{ E^{R,x}_D  Q^{\SPO,z}_{YD} \proj{+_x}_{D_x} } \leq 2 \sqrt{ \frac{\abs{R_x}}x }$ for all $z\in[N]$, where $Q^{\SPO,z}_{YD} := \bra z_X  Q^{\SPO}_{XYD} \ket z_X$.
\end{enumerate}
\end{lem}
\begin{proof}
(i)~
Since~$ Q^{\SPO}_{XYD}$ is a unitary controlled on register~$D$, it commutes with computational basis projections, and hence with~$\Pi^{R,x}_D$.
Therefore, and using unitary invariance, we obtain
\begin{align*}
  \norm*{ E^{R,x}_D  Q^{\SPO}_{XYD} \parens[\big]{ I - \Pi^{R,x}_D } }
&= \norm*{ E^{R,x}_D \parens[\big]{ I - \Pi^{R,x}_D }  Q^{\SPO}_{XYD} } \\
% &= \norm*{ \Pi^{R,x}_D \parens[\big]{ I - \proj{+_x}_{D_x} } \parens[\big]{ I - \Pi^{R,x}_D }  Q^{\SPO}_{XYD} } \\
&= \norm*{ \Pi^{R,x}_D \parens[\big]{ I - \proj{+_x}_{D_x} } \parens[\big]{ I - \Pi^{R,x}_D } } \\
&= \norm*{ \Pi^{R,x}_D \proj{+_x}_{D_x} \parens[\big]{ I - \Pi^{R,x}_D } } \\
&\leq \norm*{ \Pi^{R,x}_D \ket{+_x}_{D_x} } % \norm*{ \bra{+_x}_{D_x} \parens[\Big]{ I - \Pi^{R,x}_D } }
\leq \sqrt{\frac{\abs{R_x}}{x}}.
\end{align*}
The last inequality holds due to \cref{lem:help}.

(ii)~Since
\begin{align*}
    E^{R,x}_D  Q^{\SPO}_{XYD} \ket\phi
=   E^{R,x}_D  Q^{\SPO}_{XYD} E^{R,x}_D \ket\phi
+   E^{R,x}_D  Q^{\SPO}_{XYD} \parens*{ I - E^{R,x}_D } \ket\phi,
\end{align*}
and $\norm{E^{R,x}_D  Q^{\SPO}_{XYD} E^{R,x}_D \ket\phi} \leq \norm{E^{R,x}_D \ket\phi}$, we have, by the triangle inequality,
\begin{align*}
    \norm*{ E^{R,x}_D  Q^{\SPO}_{XYD} \ket\phi } - \norm*{ E^{R,x}_D \ket\phi }
\leq \norm*{ E^{R,x}_D  Q^{\SPO}_{XYD} \parens*{ I - E^{R,x}_D } \ket\phi }.
\end{align*}
Since $I - E^{R,x}_D = (I - \Pi^{R,x}_D)(I - \proj{+_x}_{D_x}) \;+\; \proj{+_x}_{D_x}$, we can bound the right-hand side by
\begin{align*}
&\quad \norm*{ E^{R,x}_D  Q^{\SPO}_{XYD} \parens*{ I - E^{R,x}_D } \ket\phi } \\
&\leq \norm*{ E^{R,x}_D  Q^{\SPO}_{XYD} (I - \Pi^{R,x}_D)(I - \proj{+_x}_{D_x}) \ket\phi }
+ \norm*{ E^{R,x}_D  Q^{\SPO}_{XYD} \proj{+_x}_{D_x} \ket\phi } \\
&\leq \sqrt{\frac{\abs{R_x}}x} \norm*{ \parens[\big]{ I - \proj{+_x}_{D_x} } \ket\phi }
+ \norm*{ E^{R,x}_D  Q^{\SPO}_{XYD} \proj{+_x}_{D_x} \ket\phi },
\end{align*}
using another triangle inequality and, in the last step, part~(i).

(iii)~We can proceed similarly as in part~(i):
\begin{align*}
\norm*{ E^{R,x}_D  Q^{\SPO,z}_{YD} \proj{+_x}_{D_x} }
&= \norm*{ \Pi^{R,x}_D \parens*{ I - \proj{+_x}_{D_x} }  Q^{\SPO,z}_{YD} \proj{+_x}_{D_x} } \\
&\leq \norm*{ \Pi^{R,x}_D  Q^{\SPO,z}_{YD} \proj{+_x}_{D_x} } + \norm*{ \Pi^{R,x}_D \proj{+_x}_{D_x}  Q^{\SPO,z}_{YD} \proj{+_x}_{D_x} } \\
&= \norm*{  Q^{\SPO,z}_{YD} \Pi^{R,x}_D \proj{+_x}_{D_x} } + \norm*{ \Pi^{R,x}_D \proj{+_x}_{D_x}  Q^{\SPO,z}_{YD} \proj{+_x}_{D_x} } \\
&\leq 2 \norm{ \Pi^{R,x}_D \ket{+_x}_{D_x} }
\leq 2 \sqrt{\frac{\abs{R_x}}x}.
\qedhere
\end{align*}
\end{proof}

We now bound the effect of a single forward query on the probability amplitude.

\begin{lem}[Forward query]\label{lem:forward}
For any pure state~$\ket\phi_{AXYD}$ such that $\norm{\bra\pi_D \ket\phi_{AXYD}}^2 = \frac1{N!}$ for all~$\pi\in S_N$, and for any~$x\in[N]$, we have
\begin{align*}
\norm*{ E^{R,x}_D \OSPO_{XYD} \ket\phi } - \norm*{ E^{R,x}_D \ket\phi } &\leq \sqrt{\frac{\abs{R_x}}x} \ \norm*{ \parens[\big]{ I - \proj{+_x}_{D_x} } \ket\phi }
+ 2 \sqrt{\uglyterm_{\ket\phi, x}}
\end{align*}
where $\ket{\phi_x}_{AYD} := \bra x_X \ket\phi_{AXYD}$ and
\begin{align*}
\uglyterm_{\ket\phi,x} :=	\frac{\abs{R_x}}x \, \norm{ \ket{\phi_x} }^2
	+ \frac{\abs{R_x}}{x^2 N}
	+ \frac1x \sum_{z=1}^{x-1} \sum_{\pi_{x^c} : \pi_{x^c}(z) \in R_x} \norm[\Big]{ \bra{\pi_{x^c}}_{D_{x^c}} \ket{\phi_z} }^2
\end{align*}
\end{lem}
%\cm{I introduced the notation $\uglyterm_{\ket\phi, x}$ for future benefit.}
\begin{proof}
By \cref{lem:easy}~(ii), we have
\begin{align}\label{eq:forward rhs}
\norm*{ E^{R,x}_D \OSPO_{XYD} \ket\phi } - \norm[\big]{ E^{R,x}_D \ket\phi }
 \leq \sqrt{\frac{\abs{R_x}}x} \norm[\big]{ \parens[\big]{ I - \proj{+_x}_{D_x} } \ket\phi } + \norm[\big]{ E^{R,x}_D \OSPO_{XYD} \proj{+_x}_{D_x} \ket\phi },
\end{align}
To upper bound the right-hand side norm, we compute its square as follows:
\begin{align*}
  \norm*{ E^{R,x}_D \OSPO_{XYD} \proj{+_x}_{D_x} \ket\phi }^2
= \sum_{z\in[N]} \norm*{ \bra z_X E^{R,x}_D \OSPO_{XYD} \proj{+_x}_{D_x} \ket\phi }^2
= \sum_{z\in[N]} \norm*{ E^{R,x}_D \OSPOz_{YD} \proj{+_x}_{D_x} \ket{\phi_z} }^2,
\end{align*}
where we defined $\OSPOz_{YD} := \bra z_X \OSPO_{XYD} \ket z_X$.
We analyze the summands and distinguish three cases:
\begin{enumerate}
\item For $z > x$, $\OSPOz_{YD}$ acts trivially on~$D_x$ by \cref{lem:small-x-not-touched}, and hence commutes with~$\proj{+_x}_{D_x}$.
As~$E^{R,x}_D \proj{+_x}_{D_x} = 0$, we see that the corresponding summands vanish.

\item For $z = x$, then we have the following bound from \cref{lem:easy}~(iii):
\begin{align*}
\norm*{ E^{R,x}_D \OSPOx_{YD} \proj{+_x}_{D_x} \ket{\phi_x} }^2 \leq 4 \frac{\abs{R_x}}x \, \norm{ \ket{\phi_x} }^2.
\end{align*}

\item For $z < x$, the argument is more involved.
We begin by computing the action of~$\OSPOz_{YD}$ on computational basis states.
For $\pi \in S_N$, let us write $\pi = \pi_{>x} \tp x t \pi_{<x}$ as in \cref{subsec:active}, with $t\in[x]$, and define~$\pi_{x^c} := \pi_{>x} \pi_{<x}$.
We may identify~$\pi_{x^c}$ with the indices~$t_k$ for~$k\neq x$; accordingly we shall write~$\ket\pi_D = \ket t_{D_x} \ket{\pi_{x^c}}_{D_{x^c}}$.
Note that if~$\pi_{<x}(z) = t$ then $\pi(z) = \pi_{>x}(x)$, while otherwise~$\pi(z) = \pi_{x^c}(z)$.
Thus, for every $\pi \in S_N$ and~$y \in [N]$, we have
\begin{align*}
  \OSPOz_{YD} \ket{y,\pi}_{YD}
= \OSPOz_{YD} \ket{y,t,\pi_{x^c}}_{YD_xD_{x^c}}
= \begin{cases}
  \ket{y \op \pi_{>x}(x),t,\pi_{x^c}}_{YD_xD_{x^c}} & \text{if $\pi_{<x}(z) = t$}, \\
  \ket{y \op \pi_{x^c}(z),t,\pi_{x^c}}_{YD_xD_{x^c}} & \text{otherwise.}
\end{cases}
\end{align*}
It follows that
\begin{align*}
  \OSPOz_{YD} \ket y_Y \ket{+_x}_{D_x} \ket{\pi_{x^c}}_{D_{x^c}}
&= \ket{y \op \pi_{x^c}(z)}_Y \ket{+_x}_{D_x} \ket{\pi_{x^c}}_{D_{x^c}} \\
&+ \frac1{\sqrt x} \parens[\Big]{ \ket{y \op \pi_{>x}(x)}_Y - \ket{y \op \pi_{x^c}(z)}_Y} \ket{\pi_{<x}(z)}_{D_x} \ket{\pi_{x^c}}_{D_{x^c}},
\end{align*}
hence
\begin{align*}
  \parens[\big]{ I - \proj{+_x}_{D_x} } \OSPOz_{YD} \ket y_Y \ket{+_x}_{D_x} \ket{\pi_{x^c}}_{D_{x^c}}
&= \frac1{\sqrt x} \parens[\Big]{ \ket{y \op \pi_{>x}(x)}_Y - \ket{y \op \pi_{x^c}(z)}_Y} \\
&\ot \parens[\Big]{ \ket{\pi_{<x}(z)}_{D_x} - \frac1{\sqrt x} \ket{+_x}_{D_x} } \ot \ket{\pi_{x^c}}_{D_{x^c}}
\end{align*}
and finally
\begin{align*}
&\quad E^{R,x}_D \OSPOz_{YD} \ket y_Y \ket{+_x}_{D_x} \ket{\pi_{x^c}}_{D_{x^c}} \\
&= \frac1{\sqrt x} \parens[\Big]{ \ket{y \op \pi_{>x}(x)}_Y - \ket{y \op \pi_{x^c}(z)}_Y}
\ot \parens[\Big]{ \boldsymbol{1}_{\pi_{x^c}(z) \in R_x} \ket{\pi_{<x}(z)}_{D_x} - \frac1x \sum_{t \in [x] \cap \pi_{>x}^{-1}(R_x)} \ket t_{D_x} } \ot \ket{\pi_{x^c}}_{D_{x^c}} \\
&= M_{YD_{x^c}}^{(z)} \parens*{ \frac1{\sqrt x} \ket y_Y \ot \parens[\Big]{ \boldsymbol{1}_{\pi_{x^c}(z) \in R_x} \ket{\pi_{<x}(z)}_{D_x} - \frac1x \sum_{t \in [x] \cap \pi_{>x}^{-1}(R_x)} \ket t_{D_x} } \ot \ket{\pi_{x^c}}_{D_{x^c}} },
\end{align*}
where the operator~$M_{YD_{x^c}}^{(z)}$ is defined by~$M_{YD_{x^c}}^{(z)} \ket y_Y \ket{\pi_{x^c}}_{D_{x^c}} = \parens*{ \ket{y \op \pi_{>x}(x)}_Y - \ket{y \op \pi_{x^c}(z)}_Y } \ket{\pi_{x^c}}_{D_{x^c}}$ and has operator norm~$\leq\sqrt2$.
Thus:
\begin{align*}
  E^{R,x}_D \OSPOz_{YD} \ket{+_x}_{D_x} \ket{\pi_{x^c}}_{D_{x^c}}
= M_{YD_{x^c}}^{(z)} \frac1{\sqrt x} \parens[\Big]{ \boldsymbol{1}_{\pi_{x^c}(z) \in R_x} \ket{\pi_{<x}(z)}_{D_x} - \frac1x \sum_{t \in [x] \cap \pi_{>x}^{-1}(R_x)} \ket t_{D_x} } \ket{\pi_{x^c}}_{D_{x^c}}.
\end{align*}
We can now bound the desired norm:
\begin{align*}
&\quad \norm*{ E^{R,x}_D \OSPOz_{YD} \proj{+_x}_{D_x} \ket{\phi_z} }^2 \\
&= \norm[\Big]{ E^{R,x}_D \OSPOz_{YD} \proj{+_x}_{D_x} \sum_{\pi_{x^c}} \proj{\pi_{x^c}}_{D_{x^c}} \ket{\phi_z} }^2 \\
&= \norm[\Big]{ \sum_{\pi_{x^c}} E^{R,x}_D \OSPOz_{YD} \ket{+_x}_{D_x} \ket{\pi_{x^c}}_{D_{x^c}} \bra{+_x}_{D_x} \bra{\pi_{x^c}}_{D_{x^c}} \ket{\phi_z} }^2 \\
&= \norm[\Big]{ M_{YD_{x^c}}^{(z)} \sum_{\pi_{x^c}} \frac1{\sqrt x} \parens[\Big]{ \boldsymbol{1}_{\pi_{x^c}(z) \in R_x} \ket{\pi_{<x}(z)}_{D_x} - \frac1x \sum_{t \in [x] \cap \pi_{>x}^{-1}(R_x)} \ket t_{D_x} } \ket{\pi_{x^c}}_{D_{x^c}} \bra{+_x}_{D_x} \bra{\pi_{x^c}}_{D_{x^c}} \ket{\phi_z} }^2 \\
&\leq \frac2x \sum_{\pi_{x^c}} \norm[\Big]{ \parens[\Big]{ \boldsymbol{1}_{\pi_{x^c}(z) \in R_x} \ket{\pi_{<x}(z)}_{D_x} - \frac1x \sum_{t \in [x] \cap \pi_{>x}^{-1}(R_x)} \ket t_{D_x} } \bra{+_x}_{D_x} \bra{\pi_{x^c}}_{D_{x^c}} \ket{\phi_z} }^2 \\
% &\leq \frac2x \sum_{\pi_{x^c}} \parens*{
    % \norm[\Big]{ \boldsymbol{1}_{\pi_{x^c}(z) \in R_x} \ket{\pi_{<x}(z)}_{D_x} \bra{+_x}_{D_x} \bra{\pi_{x^c}}_{D_{x^c}} \ket{\phi_z} }
% + \norm[\Big]{ \frac1x \sum_{t \in [x] \cap \pi_{>x}^{-1}(R_x)} \ket t_{D_x} \bra{+_x}_{D_x} \bra{\pi_{x^c}}_{D_{x^c}} \ket{\phi_z} }
% }^2 \\
&\leq \frac4x \sum_{\pi_{x^c}} \parens[\bigg]{
    \boldsymbol{1}_{\pi_{x^c}(z) \in R_x} \norm[\Big]{ \bra{+_x}_{D_x} \bra{\pi_{x^c}}_{D_{x^c}} \ket{\phi_z} }^2
+ \frac1{x^2} \norm[\Big]{ \sum_{t \in [x] \cap \pi_{>x}^{-1}(R_x)} \ket t_{D_x} \bra{+_x}_{D_x} \bra{\pi_{x^c}}_{D_{x^c}} \ket{\phi_z} }^2
} \\
% &= \frac4x \sum_{\pi_{x^c}} \parens[\bigg]{
%     \boldsymbol{1}_{\pi_{x^c}(z) \in R_x} \norm[\Big]{ \bra{+_x}_{D_x} \bra{\pi_{x^c}}_{D_{x^c}} \ket{\phi_z} }^2
% + \frac1{x^2} \sum_{t \in [x] \cap \pi_{>x}^{-1}(R_x)} \norm[\big]{ \bra{+_x}_{D_x} \bra{\pi_{x^c}}_{D_{x^c}} \ket{\phi_z} }^2
% } \\
&= \frac4x \sum_{\pi_{x^c} : \pi_{x^c}(z) \in R_x} \norm[\Big]{ \bra{+_x}_{D_x} \bra{\pi_{x^c}}_{D_{x^c}} \ket{\phi_z} }^2
+ \frac4{x^3} \sum_{\pi_{x^c}} \sum_{t \in [x] \cap \pi_{>x}^{-1}(R_x)} \norm[\big]{ \bra{+_x}_{D_x} \bra{\pi_{x^c}}_{D_{x^c}} \ket{\phi_z} }^2 \\
&\leq \frac4x \sum_{\pi_{x^c} : \pi_{x^c}(z) \in R_x} \norm[\Big]{ \bra{\pi_{x^c}}_{D_{x^c}} \ket{\phi_z} }^2
+ \frac4{x^3} \sum_{\pi_{x^c}} \sum_{t \in [x] \cap \pi_{>x}^{-1}(R_x)} \norm[\big]{ \bra{\pi_{x^c}}_{D_{x^c}} \ket{\phi_z} }^2.
\end{align*}
Here the first inequality follows from the bound on the operator norm of $M_{YD_{x^c}}^{(z)}$, and the fact that all $\pi_{x^c}$ are orthogonal, the second inequality follows by Cauchy-Schwarz, and the last inequality follows by the fact that $\ket{+_x}$ is a unit vector.  It follows that
\begin{align*}
&\quad \sum_{z=1}^{x-1} \norm*{ E^{R,x}_D \OSPOz_{YD} \proj{+_x}_{D_x} \ket{\phi_z} }^2 \\
&\leq \frac4x \sum_{z=1}^{x-1} \sum_{\pi_{x^c} : \pi_{x^c}(z) \in R_x} \norm[\Big]{ \bra{\pi_{x^c}}_{D_{x^c}} \ket{\phi_z} }^2
+ \frac4{x^3} \sum_{z=1}^{x-1} \sum_{\pi_{x^c}} \sum_{t \in [x] \cap \pi_{>x}^{-1}(R_x)} \norm[\big]{ \bra{\pi_{x^c}}_{D_{x^c}} \ket{\phi_z} }^2 \\
&\leq \frac4x \sum_{z=1}^{x-1} \sum_{\pi_{x^c} : \pi_{x^c}(z) \in R_x} \norm[\Big]{ \bra{\pi_{x^c}}_{D_{x^c}} \ket{\phi_z} }^2
+ \frac4{x^3} \sum_{\pi_{x^c}} \sum_{t \in [x] \cap \pi_{>x}^{-1}(R_x)} \norm[\big]{ \bra{\pi_{x^c}}_{D_{x^c}} \ket\phi }^2 \\
&= \frac4x \sum_{z=1}^{x-1} \sum_{\pi_{x^c} : \pi_{x^c}(z) \in R_x} \norm[\Big]{ \bra{\pi_{x^c}}_{D_{x^c}} \ket{\phi_z} }^2
+ \frac4{x^2} \sum_{\pi_{x^c}} \sum_{t \in [x] \cap \pi_{>x}^{-1}(R_x)} \frac1{N!},
\end{align*}
where we used that~$\norm{\bra{\pi}_{D} \ket\phi}^2 = \frac1{N!}$ for all $\pi\in S_N$ and hence~$\norm{\bra{\pi_{x^c}}_{D_{x^c}} \ket\phi}^2 = \frac x{N!}$.
Since
\begin{align*}
  \sum_{\pi_{x^c}} \sum_{t \in [x] \cap \pi_{>x}^{-1}(R_x)} \frac1{N!}
% = \frac1{N!} \sum_{t \in [x]} \sum_{\pi_{x^c} : \pi_{>x}(t) \in R_x} 1
= \frac1{N!} \sum_{\pi \in S_N : \pi(x) \in R_x} 1
= \Pr\bigl( \pi(x) \in R_x \mid \pi \leftarrow S_N \bigr)
\leq \frac{\abs{R_x}}N,
\end{align*}
we find that
\begin{align*}
    \sum_{z=1}^{x-1} \norm*{ E^{R,x}_D \OSPOz_{YD} \proj{+_x}_{D_x} \ket{\phi_z} }^2
\leq \frac4x \sum_{z=1}^{x-1} \sum_{\pi_{x^c} : \pi_{x^c}(z) \in R_x} \norm[\Big]{ \bra{\pi_{x^c}}_{D_{x^c}} \ket{\phi_z} }^2
+ \frac4{x^2} \frac{\abs{R_x}}N.
\end{align*}
\end{enumerate}
Altogether, we obtain the following bound for the right-hand side term in \cref{eq:forward rhs}:
\begin{align*}
    \norm*{ E^{R,x}_D \OSPO_{XYD} \proj{+_x}_{D_x} \ket\phi }
&\leq \sqrt{
4 \frac{\abs{R_x}}x \, \norm{ \ket{\phi_x} }^2
+ \frac4x \sum_{z=1}^{x-1} \sum_{\pi_{x^c} : \pi_{x^c}(z) \in R_x} \norm[\Big]{ \bra{\pi_{x^c}}_{D_{x^c}} \ket{\phi_z} }^2
+ \frac4{x^2} \frac{\abs{R_x}}N
} \\
&= 2 \sqrt{
\frac{\abs{R_x}}x \, \norm{ \ket{\phi_x} }^2
+ \frac{\abs{R_x}}{x^2 N}
+ \frac1x \sum_{z=1}^{x-1} \sum_{\pi_{x^c} : \pi_{x^c}(z) \in R_x} \norm[\Big]{ \bra{\pi_{x^c}}_{D_{x^c}} \ket{\phi_z} }^2
}
\end{align*}
and the lemma follows.
\end{proof}

Next, we now bound the effect of an inverse query on the probability amplitude.

\begin{lem}[Inverse query]\label{lem:inverse}
For any pure state~$\ket\phi_{AXYD}$ such that $\norm{\bra\pi_D \ket\phi_{AXYD}}^2 = \frac1{N!}$, and for any~$x\in[N]$, we have
\begin{align*}
 \norm*{ E^{R,x}_D \OinvSPO_{XYD} \ket\phi } - \norm*{ E^{R,x}_D \ket\phi }
&\leq \sqrt{\frac{\abs{R_x}}x} \norm*{ \parens[\big]{ I - \proj{+_x}_{D_x} } \ket\phi }
+ 4 \sqrt{ \uglyterm^{\mathrm{inv}}_{\ket\phi,x}},
\end{align*}
where $\ket{\phi_x}_{AYD} := \bra x_X \ket\phi_{AXYD}$ and
\begin{align*}
\uglyterm^{\mathrm{inv}}_{\ket\phi,x} := \;
		&\frac{\abs{R_x}}x \, \norm{ \ket{\phi_x} }^2 + \frac{\abs{R_x}}{x^2N}
		+ \frac1x \sum_{z \in R_x} \sum_{\pi_{x^c} : \pi_{>x}^{-1}(z) < x} \norm[\Big]{ \bra{\pi_{x^c}}_{D_{x^c}} \ket{\phi_z} }^2 \\
		&+ \frac1x \sum_{z=x+1}^N \sum_{\pi_{x^c} : \pi_{>x}^{-1}(z) = x} \abs{[x] \cap \pi_{>x}^{-1}(R_x)} \, \norm[\Big]{ \bra{\pi_{x^c}}_{D_{x^c}} \ket{\phi_z} }^2
\end{align*}
\end{lem}
\begin{proof}
By \cref{lem:easy}~(ii), we have
\begin{align}\label{eq:inverse rhs}
\norm*{ E^{R,x}_D \OinvSPO_{XYD} \ket\phi } - \norm*{ E^{R,x}_D \ket\phi }
\leq \sqrt{\frac{\abs{R_x}}x} \norm*{ \parens[\big]{ I - \proj{+_x}_{D_x} } \ket\phi }
+ \norm*{ E^{R,x}_D \OinvSPO_{XYD} \proj{+_x}_{D_x} \ket\phi },
\end{align}
and we need to bound the right-hand side term, which we can rewrite as
\begin{align*}
  \norm*{ E^{R,x}_D \OinvSPO_{XYD} \proj{+_x}_{D_x} \ket\phi }^2
= \sum_{z\in[N]} \norm*{ E^{R,x}_D \OinvSPOz_{YD} \proj{+_x}_{D_x} \ket{\phi_z} }^2,
\end{align*}
where $\OinvSPOz_{YD} := \bra z_X \OinvSPO_{XYD} \ket z_X$.
% We pause to introduce some notation.
% For any permutation~$\pi\in S_N$, we consider the decomposition $\pi = \tp N {t_N} \tp {N\!-\!1} {t_{N-1}} \cdots \tp 2 {t_2} \tp 1 {t_1}$ from \cref{eq:tower-decomposition}.
% Recall that~$\pi_{>x} = \tp N {t_N} \tp {N\!-\!1} {t_{N-1}} \cdots \tp {x+1} {t_{x+1}}$ and $\pi_{<x} = \tp {x-1} {t_{x-1}} \cdots \tp 2 {t_2} \tp 1 {t_1}$.
% For any $t\in[x]$, we denote by $\pi_{\tp x t} := \pi_{>x} \tp x t \pi_{<x}$ the permutation obtained by replacing the $x$-th transposition with~$\tp x t$.
% Finally, we let
% \begin{align*}
% b^k_{x,z} := \frac1{N!} \abs*{
% \braces[\Big]{
%     \pi \in S_N
% \;\Big|\;
%     \abs*{ \braces*{ t \in [x] \mid x \text{ is inverse-active for $\pi_{\tp x t}$ and $z$} } } = k
% }
% }.
% \end{align*}
% In other words, $b^k_{x,z}$ is the fraction of permutations such that there are exactly~$k$ choices for the transposition~$\tp x t$ that render $x$ is inverse-active for~$\pi$ and~$z$.
We analyze the the right-hand side summands and distinguish three cases:
\begin{enumerate}
\item For $z = x$, we have the following bound from \cref{lem:easy}~(iii):
\begin{align*}
    \norm*{ E^{R,x}_D \OinvSPOz_{YD} \proj{+_x}_{D_x} \ket{\phi_x} }^2 \leq 4 \frac{\abs{R_x}}x \, \norm{ \ket{\phi_x} }^2.
\end{align*}
% For $z = x$, note that~$x$ is inverse-active for~$\pi_{\tp x t}$ and~$x$ if and only if~$t_w \neq x$ for all~$w > x$ (regardless of the choice of~$t \in [x]$).
% We thus get:
% \begin{align*}
%     b^k_{x,x} = \begin{cases}
%         \prod_{w=x+1}^N \frac {w-1} w = \frac x N & \text{ if $k = x$, } \\
%         1 - \frac x N & \text{ if $k = 0$, } \\
%         0 & \text{ otherwise.}
%     \end{cases}
% \end{align*}

\item For $z < x$, we proceed similarly as in the corresponding case of \cref{lem:forward}.
For $\pi \in S_N$, let us again write $\pi = \pi_{>x} \tp x t \pi_{<x}$ as in \cref{subsec:active}, with $t\in[x]$, define~$\pi_{x^c} := \pi_{>x} \pi_{<x}$, identify~$\pi_{x^c}$ with the indices~$t_k$ for~$k\neq x$, and write~$\ket\pi_D = \ket t_{D_x} \ket{\pi_{x^c}}_{D_{x^c}}$.
% It follows that
% \begin{align*}
%     b^1_{x,z} = \begin{cases}
%     \prod_{w=x+1}^N \frac{w-1} w = \frac x N & \text{ if $k = 1$, } \\
%     1 - \frac x N & \text{ if $k = 0$, } \\
%     0 & \text{ otherwise.}
%     \end{cases}
% \end{align*}
% \begin{align*}
%     E^{R,x}_D \OinvSPOz_{YD} \proj{+_x}_{D_x} \ket{\phi_z}
% \end{align*}
Let
\begin{align*}
    S_{x,z}
:= \braces[\big]{ \pi_{x^c} \;:\; \pi_{>x}^{-1}(z) \leq x }
= \braces[\big]{ \pi_{x^c} \;:\; \pi_{>x}^{-1}(z) < x }
= \braces[\big]{ \pi_{x^c} \;:\; \pi_{>x}^{-1}(z) = z }
= \braces[\big]{ \pi_{x^c} \;:\; t_w \neq z \; \forall w > x }.
\end{align*}
If $\pi_{x^c} \not\in S_{x,z}$, then $\pi^{-1}(z) = z$ for any choice of~$t$. Hence $\OinvSPOz_{YD}$ acts trivially on~$D_x$ and we have
\begin{align}\label{eq:inv ii easy}
  E^{R,x}_D \OinvSPOz_{YD} \ket{+_x}_{D_x} \ket{\pi_{x^c}}_{D_{x^c}}
= E^{R,x}_D  \ket{+_x}_{D_x} \OinvSPOz_{YD} \ket{\pi_{x^c}}_{D_{x^c}}
= 0.
\end{align}
Now suppose that~$\pi_{x^c} \in S_{x,z}$.
Then,
\begin{align*}
  \OinvSPOz_{YD} \ket{y,t,\pi_{x^c}}_{YD_xD_{x^c}}
= \ket{y \op \pi^{-1}(z),t,\pi_{x^c}}_{YD_xD_{x^c}}
% = \ket{y \op \pi_{<x}^{-1} \tp x t (z)}_Y \ot \ket t_{D_x} \ot \ket{\pi_{x^c}}_{D_{x^c}}
= \begin{cases}
  \ket{y \op x,t,\pi_{x^c}}_{YD_xD_{x^c}} & \text{ if $t = z$, } \\
  \ket{y \op \pi_{<x}^{-1}(z),t,\pi_{x^c}}_{YD_xD_{x^c}} & \text{ otherwise.}
\end{cases}
\end{align*}
It follows that
\begin{align*}
  \OinvSPOz_{YD} \ket y_Y \ket{+_x}_{D_x} \ket{\pi_{x^c}}_{D_{x^c}}
&= \ket{y \op \pi_{<x}^{-1}(z)}_Y \ket{+_x}_{D_x} \ket{\pi_{x^c}}_{D_{x^c}} \\
&+ \frac1{\sqrt x} \parens[\Big]{ \ket{y \op x}_Y - \ket{y \op \pi_{<x}^{-1}(z)}_Y} \ket{z}_{D_x} \ket{\pi_{x^c}}_{D_{x^c}},
\end{align*}
hence
\begin{align*}
  \parens[\big]{ I - \proj{+_x}_{D_x} } \OinvSPOz_{YD} \ket y_Y \ket{+_x}_{D_x} \ket{\pi_{x^c}}_{D_{x^c}}
&= \frac1{\sqrt x} \parens[\Big]{ \ket{y \op x}_Y - \ket{y \op \pi_{<x}^{-1}(z)}_Y} \\
&\ot \parens[\Big]{ \ket{z}_{D_x} - \frac1{\sqrt x} \ket{+_x}_{D_x} } \ot \ket{\pi_{x^c}}_{D_{x^c}}
\end{align*}
and finally
\begin{align*}
&\quad E^{R,x}_D \OinvSPOz_{YD} \ket y_Y \ket{+_x}_{D_x} \ket{\pi_{x^c}}_{D_{x^c}} \\
&= \frac1{\sqrt x} \parens[\Big]{ \ket{y \op x}_Y - \ket{y \op \pi_{<x}^{-1}(z)}_Y}
\ot \parens[\Big]{ \boldsymbol{1}_{z \in R_x} \ket{z}_{D_x} - \frac1x \sum_{t \in [x] \cap \pi_{>x}^{-1}(R_x)} \ket t_{D_x} } \ot \ket{\pi_{x^c}}_{D_{x^c}} \\
&= M_{YD_{x^c}}^{(z)} \parens*{ \frac1{\sqrt x} \ket{y}_Y \ot \parens[\Big]{ \boldsymbol{1}_{z \in R_x} \ket{z}_{D_x} - \frac1x \sum_{t \in [x] \cap \pi_{>x}^{-1}(R_x)} \ket t_{D_x} } \ot \ket{\pi_{x^c}}_{D_{x^c}} },
\end{align*}
where the operator~$M_{YD_{x^c}}^{(z)}$ is defined by~$M_{YD_{x^c}}^{(z)} \ket y_Y \ket{\pi_{x^c}}_{YD_{x^c}} = \parens*{ \ket{y \op x}_Y - \ket{y \op \pi_{<x}^{-1}(z)}_Y } \ket{\pi_{x^c}}_{D_{x^c}}$ and has operator norm~$\leq\sqrt2$.
Thus:
\begin{align*}
  E^{R,x}_D \OinvSPOz_{YD} \ket{+_x}_{D_x} \ket{\pi_{x^c}}_{D_{x^c}}
= M_{YD_{x^c}}^{(z)} \frac1{\sqrt x} \parens[\Big]{ \boldsymbol{1}_{z \in R_x} \ket{z}_{D_x} - \frac1x \sum_{t \in [x] \cap \pi_{>x}^{-1}(R_x)} \ket t_{D_x} } \ot \ket{\pi_{x^c}}_{D_{x^c}}
\end{align*}
Using the above and \cref{eq:inv ii easy}, we can now bound the desired norm:
\begin{align*}
&\quad \norm*{ E^{R,x}_D \OinvSPOz_{YD} \proj{+_x}_{D_x} \ket{\phi_z} }^2 \\
&= \norm[\Big]{ E^{R,x}_D \OinvSPOz_{YD} \proj{+_x}_{D_x} \sum_{\pi_{x^c}} \proj{\pi_{x^c}}_{D_{x^c}} \ket{\phi_z} }^2 \\
&= \norm[\Big]{ \sum_{\pi_{x^c} \in S_{x,z}} E^{R,x}_D \OinvSPOz_{YD} \ket{+_x}_{D_x} \ket{\pi_{x^c}}_{D_{x^c}} \bra{+_x}_{D_x} \bra{\pi_{x^c}}_{D_{x^c}} \ket{\phi_z} }^2 \\
&= \norm[\Big]{ M_{YD_{x^c}}^{(z)} \sum_{\pi_{x^c} \in S_{x,z}} \frac1{\sqrt x} \parens[\Big]{ \boldsymbol{1}_{z \in R_x} \ket{z}_{D_x} - \frac1x \sum_{t \in [x] \cap \pi_{>x}^{-1}(R_x)} \ket t_{D_x} } \ket{\pi_{x^c}}_{D_{x^c}} \bra{+_x}_{D_x} \bra{\pi_{x^c}}_{D_{x^c}} \ket{\phi_z} }^2 \\
&\leq \frac2x \sum_{\pi_{x^c} \in S_{x,z}} \norm[\Big]{ \parens[\Big]{ \boldsymbol{1}_{z \in R_x} \ket{z}_{D_x} - \frac1x \sum_{t \in [x] \cap \pi_{>x}^{-1}(R_x)} \ket t_{D_x} } \bra{+_x}_{D_x} \bra{\pi_{x^c}}_{D_{x^c}} \ket{\phi_z} }^2 \\
&\leq \frac4x \sum_{\pi_{x^c} \in S_{x,z}} \parens[\bigg]{
  \boldsymbol{1}_{z \in R_x} \norm[\Big]{ \bra{+_x}_{D_x} \bra{\pi_{x^c}}_{D_{x^c}} \ket{\phi_z} }^2
+ \frac1{x^2} \norm[\Big]{ \sum_{t \in [x] \cap \pi_{>x}^{-1}(R_x)} \ket t_{D_x} \bra{+_x}_{D_x} \bra{\pi_{x^c}}_{D_{x^c}} \ket{\phi_z} }^2
} \\
&= \frac4x \sum_{\pi_{x^c} \in S_{x,z}} \boldsymbol{1}_{z \in R_x} \norm[\Big]{ \bra{+_x}_{D_x} \bra{\pi_{x^c}}_{D_{x^c}} \ket{\phi_z} }^2
+ \frac4{x^3} \sum_{\pi_{x^c} \in S_{x,z}} \sum_{t \in [x] \cap \pi_{>x}^{-1}(R_x)} \norm[\Big]{ \bra{+_x}_{D_x} \bra{\pi_{x^c}}_{D_{x^c}} \ket{\phi_z} }^2 \\
&\leq \frac4x \sum_{\pi_{x^c} \in S_{x,z}} \boldsymbol{1}_{z \in R_x} \norm[\Big]{ \bra{\pi_{x^c}}_{D_{x^c}} \ket{\phi_z} }^2
+ \frac4{x^3} \sum_{\pi_{x^c} \in S_{x,z}} \sum_{t \in [x] \cap \pi_{>x}^{-1}(R_x)} \norm[\Big]{ \bra{\pi_{x^c}}_{D_{x^c}} \ket{\phi_z} }^2
\end{align*}
following the same reasoning as the proof of \cref{lem:forward}.
% \gm{Here the first inequality follows from the bound on the operator norm of $M_{YD_{x^c}}^{(z)}$, and the fact that all $\pi_{x^c}$ are orthogonal, the second inequality follows by Cauchy-Schwarz, and the the last inequality follows by the fact that $\ket{+_x}$ is a unit vector. (Correct? BTW should we recall in the prelim Cauchy-Schwarz, or is it a noob thing to do that?)}
% \MW{Correct! Perhaps good to explain this reasoning when we use it for the first time in the proof of Lemma~6.1, and then just say something like ``similarly as in the proof of Lemma~6.1'' when we use it again here.}
It follows that
\begin{align*}
&\quad \sum_{z=1}^{x-1} \norm*{ E^{R,x}_D \OinvSPOz_{YD} \proj{+_x}_{D_x} \ket{\phi_z} }^2 \\
&\leq \frac4x \sum_{z=1}^{x-1} \sum_{\pi_{x^c} \in S_{x,z}} \boldsymbol{1}_{z \in R_x} \norm[\Big]{ \bra{\pi_{x^c}}_{D_{x^c}} \ket{\phi_z} }^2
+ \frac4{x^3} \sum_{z=1}^{x-1} \sum_{\pi_{x^c} \in S_{x,z}} \sum_{t \in [x] \cap \pi_{>x}^{-1}(R_x)} \norm[\Big]{ \bra{\pi_{x^c}}_{D_{x^c}} \ket{\phi_z} }^2 \\
&\leq \frac4x \sum_{z=1}^{x-1} \sum_{\pi_{x^c} \in S_{x,z}} \boldsymbol{1}_{z \in R_x} \norm[\Big]{ \bra{\pi_{x^c}}_{D_{x^c}} \ket{\phi_z} }^2
+ \frac4{x^3} \sum_{\pi_{x^c}} \sum_{t \in [x] \cap \pi_{>x}^{-1}(R_x)} \norm[\Big]{ \bra{\pi_{x^c}}_{D_{x^c}} \ket\phi }^2 \\
&\leq \frac4x \sum_{z=1}^{x-1} \sum_{\pi_{x^c} \in S_{x,z}} \boldsymbol{1}_{z \in R_x} \norm[\Big]{ \bra{\pi_{x^c}}_{D_{x^c}} \ket{\phi_z} }^2
+ \frac4{x^2} \frac{\abs{R_x}}N \\
&= \frac4x \sum_{z=1}^{x-1} \sum_{\pi_{x^c} : \pi_{>x}^{-1}(z) < x} \boldsymbol{1}_{z \in R_x} \norm[\Big]{ \bra{\pi_{x^c}}_{D_{x^c}} \ket{\phi_z} }^2
+ \frac4{x^2} \frac{\abs{R_x}}N,
\end{align*}
where the last inequality follows from the same argument as in the proof of \cref{lem:forward} and is using the assumption that~$\norm{\bra{\pi_{x^c}}_{D_{x^c}} \ket\phi}^2 = \frac x{N!}$ for all~$\pi\in S_N$.
\item For $z > x$, we use the same notation as above, but instead of~$S_{x,z}$ we consider
\begin{align*}
    S'_{x,z} &:= \braces[\big]{ \pi_{x^c} \;:\; \pi_{>x}^{-1}(z) < x }, \\
    S''_{x,z}& := \braces[\big]{ \pi_{x^c} \;:\; \pi_{>x}^{-1}(z) = x }.
\end{align*}
If $\pi_{x^c} \not\in S'_{x,z} \cup S''_{x,z}$, then we see as above that~$\pi^{-1}(z)$ does not depend on the choice of~$t$ and hence
\begin{align}\label{eq:inv iii easy}
  E^{R,x}_D \OinvSPOz_{YD} \ket{+_x}_{D_x} \ket{\pi_{x^c}}_{D_{x^c}}
= E^{R,x}_D  \ket{+_x}_{D_x} \OinvSPOz_{YD} \ket{\pi_{x^c}}_{D_{x^c}}
= 0.
\end{align}
Next, suppose that~$\pi_{x^c} \in S'_{x,z}$.
Then,
\begin{align*}
  \OinvSPOz_{YD} \ket{y,t,\pi_{x^c}}_{YD_xD_{x^c}}
% = \ket{y \op \pi^{-1}(z),t,\pi_{x^c}}_{YD_xD_{x^c}}
% = \ket{y \op \pi_{<x}^{-1} \tp x t \pi_{>x}^{-1}(z),t,\pi_{x^c}}_{YD_xD_{x^c}} \\
= \begin{cases}
  \ket{y \op x,t,\pi_{x^c}}_{YD_xD_{x^c}} & \text{ if $t = \pi_{>x}^{-1}(z)$, } \\
  \ket{y \op \pi_{x^c}^{-1}(z),t,\pi_{x^c}}_{YD_xD_{x^c}} & \text{ otherwise.}
\end{cases}
\end{align*}
It follows that
\begin{align*}
  \OinvSPOz_{YD} \ket y_Y \ket{+_x}_{D_x} \ket{\pi_{x^c}}_{D_{x^c}}
&= \ket{y \op \pi_{x^c}^{-1}(z)}_Y \ket{+_x}_{D_x} \ket{\pi_{x^c}}_{D_{x^c}} \\
&+ \frac1{\sqrt x} \parens[\Big]{ \ket{y \op x}_Y - \ket{y \op \pi_{x^c}^{-1}(z)}_Y} \ket{\pi_{>x}^{-1}(z)}_{D_x} \ket{\pi_{x^c}}_{D_{x^c}},
\end{align*}
hence
\begin{align*}
  \parens[\big]{ I - \proj{+_x}_{D_x} } \OinvSPOz_{YD} \ket y_Y \ket{+_x}_{D_x} \ket{\pi_{x^c}}_{D_{x^c}}
&= \frac1{\sqrt x} \parens[\Big]{ \ket{y \op x}_Y - \ket{y \op \pi_{x^c}^{-1}(z)}_Y} \\
&\ot \parens[\Big]{ \ket{\pi_{>x}^{-1}(z)}_{D_x} - \frac1{\sqrt x} \ket{+_x}_{D_x} } \ot \ket{\pi_{x^c}}_{D_{x^c}}
\end{align*}
and finally
\begin{align*}
&\quad E^{R,x}_D \OinvSPOz_{YD} \ket y_Y \ket{+_x}_{D_x} \ket{\pi_{x^c}}_{D_{x^c}} \\
&= \frac1{\sqrt x} \parens[\Big]{ \ket{y \op x}_Y - \ket{y \op \pi_{x^c}^{-1}(z)}_Y}
\ot \parens[\Big]{ \boldsymbol{1}_{z \in R_x} \ket{\pi_{>x}^{-1}(z)}_{D_x} - \frac1x \sum_{t \in [x] \cap \pi_{>x}^{-1}(R_x)} \ket t_{D_x} } \ot \ket{\pi_{x^c}}_{D_{x^c}} \\
&= M_{YD_{x^c}}^{(z)} \parens*{ \frac1{\sqrt x} \ket{y}_Y \ot \parens[\Big]{ \boldsymbol{1}_{z \in R_x} \ket{\pi_{>x}^{-1}(z)}_{D_x} - \frac1x \sum_{t \in [x] \cap \pi_{>x}^{-1}(R_x)} \ket t_{D_x} } \ot \ket{\pi_{x^c}}_{D_{x^c}} },
\end{align*}
where the operator~$M'_{YD_{x^c},z}$ is defined by~$M'_{YD_{x^c},z} \ket y_Y \ket{\pi_{x^c}}_{YD_{x^c}} = \parens*{ \ket{y \op x}_Y - \ket{y \op \pi_{x^c}^{-1}(z)}_Y } \ket{\pi_{x^c}}_{D_{x^c}}$ and has operator norm~$\leq\sqrt2$.
Thus:
\begin{align}\label{eq:inv iii S'}
  E^{R,x}_D \OinvSPOz_{YD} \ket{+_x}_{D_x} \ket{\pi_{x^c}}_{D_{x^c}}
= M'_{YD_{x^c},z} \frac1{\sqrt x} \parens[\Big]{ \boldsymbol{1}_{z \in R_x} \ket{\pi_{>x}^{-1}(z)}_{D_x} - \frac1x \!\!\!\!\sum_{t \in [x] \cap \pi_{>x}^{-1}(R_x)}\!\!\!\!\!\!\!\! \ket t_{D_x} } \ot \ket{\pi_{x^c}}_{D_{x^c}}.
\end{align}
Finally, suppose that~$\pi_{x^c} \in S''_{x,z}$.
Then,
\begin{align*}
  \OinvSPOz_{YD} \ket{y,t,\pi_{x^c}}_{YD_xD_{x^c}}
% = \ket{y \op \pi^{-1}(z),t,\pi_{x^c}}_{YD_xD_{x^c}}
% = \ket{y \op \pi_{<x}^{-1} \tp x t \pi_{>x}^{-1}(z),t,\pi_{x^c}}_{YD_xD_{x^c}} \\
= \ket{y \op \pi_{<x}^{-1}(t),t,\pi_{x^c}}_{YD_xD_{x^c}},
\end{align*}
hence
\begin{align*}
    \parens[\big]{ I - \proj{+_x}_{D_x} } \OinvSPOz_{YD} \ket y_Y \ket{+_x}_{D_x} \ket{\pi_{x^c}}_{D_{x^c}}
% = \parens[\big]{ I - \proj{+_x}_{D_x} } \frac1{\sqrt x}\sum_{t=1}^x \ket{y \op \pi_{<x}^{-1}(t),t,\pi_{x^c}}_{YD_xD_{x^c}} \\
&= \frac1{\sqrt x}\sum_{t=1}^x \ket{y \op \pi_{<x}^{-1}(t)}_Y \\
&\ot \parens*{ \ket t_{D_x} - \frac1{\sqrt x}\ket{+_x}_{D_x} }  \ot \ket{\pi_{x^c}}_{D_{x^c}},
\end{align*}
and finally
\begin{align*}
&\quad E^{R,x}_D \OinvSPOz_{YD} \ket y_Y \ket{+_x}_{D_x} \ket{\pi_{x^c}}_{D_{x^c}} \\
&= \frac1{\sqrt x}\sum_{t=1}^x \ket{y \op \pi_{<x}^{-1}(t)}_Y \ot \parens*{ \boldsymbol{1}_{\pi_{>x}(t) \in R_x} \ket t_{D_x} - \frac1x\sum_{t' \in [x] \cap \pi_{>x}^{-1}(R_x)} \ket{t'}_{D_x} }  \ot \ket{\pi_{x^c}}_{D_{x^c}} \\
&= M''_{YD} \parens*{ \frac1{\sqrt x} \ket y_Y \ot \sum_{t \in [x] \cap \pi_{>x}^{-1}(R_x)} \ket{t}_{D_x} \ot \ket{\pi_{x^c}}_{D_{x^c}} },
\end{align*}
with~$M''_{YD} \ket{y,t,\pi_{x^c}} := ( \ket{y \op \pi_{<x}^{-1}(t)}_Y - \frac1x \sum_{t'=1}^x \ket{y \op t'}_Y) \ket t_{D_x} \ket {\pi_{x^c}}_{D_{x^c}}$, an operator of norm~$\leq 2$, hence
\begin{align*}
  E^{R,x}_D \OinvSPOz_{YD} \ket{+_x}_{D_x} \ket{\pi_{x^c}}_{D_{x^c}}
= M''_{YD} \frac1{\sqrt x} \sum_{t \in [x] \cap \pi_{>x}^{-1}(R_x)} \ket{t}_{D_x} \ot \ket{\pi_{x^c}}_{D_{x^c}}.
\end{align*}
Together with \cref{eq:inv iii easy,eq:inv iii S'}, we obtain
\begin{align*}
&\quad \norm*{ E^{R,x}_D \OinvSPOz_{YD} \proj{+_x}_{D_x} \ket{\phi_z} }^2 \\
&= \norm[\Big]{ E^{R,x}_D \OinvSPOz_{YD} \proj{+_x}_{D_x} \sum_{\pi_{x^c}} \proj{\pi_{x^c}}_{D_{x^c}} \ket{\phi_z} }^2 \\
&= \norm[\Big]{ \sum_{\pi_{x^c} \in S'_{x,z} \cup S''_{x,z}} E^{R,x}_D \OinvSPOz_{YD} \ket{+_x}_{D_x} \ket{\pi_{x^c}}_{D_{x^c}} \bra{+_x}_{D_x} \bra{\pi_{x^c}}_{D_{x^c}} \ket{\phi_z} }^2 \\
&\leq 2 \, \norm[\Big]{ \sum_{\pi_{x^c} \in S'_{x,z}} E^{R,x}_D \OinvSPOz_{YD} \ket{+_x}_{D_x} \ket{\pi_{x^c}}_{D_{x^c}} \bra{+_x}_{D_x} \bra{\pi_{x^c}}_{D_{x^c}} \ket{\phi_z} }^2 \\
&+ 2 \, \norm[\Big]{ \sum_{\pi_{x^c} \in S''_{x,z}} E^{R,x}_D \OinvSPOz_{YD} \ket{+_x}_{D_x} \ket{\pi_{x^c}}_{D_{x^c}} \bra{+_x}_{D_x} \bra{\pi_{x^c}}_{D_{x^c}} \ket{\phi_z} }^2 \\
&= 2 \, \norm[\Big]{ M'_{YD_{x^c},z} \sum_{\pi_{x^c} \in S'_{x,z}} \frac1{\sqrt x} \parens[\Big]{ \boldsymbol{1}_{z \in R_x} \ket{\pi_{>x}^{-1}(z)}_{D_x} - \frac1x \!\!\!\!\sum_{t \in [x] \cap \pi_{>x}^{-1}(R_x)}\!\!\!\!\!\!\!\! \ket t_{D_x} } \ket{\pi_{x^c}}_{D_{x^c}} \bra{+_x}_{D_x} \bra{\pi_{x^c}}_{D_{x^c}} \ket{\phi_z} }^2 \\
&+ 2 \, \norm[\Big]{ M''_{YD} \sum_{\pi_{x^c} \in S''_{x,z}} \frac1{\sqrt x} \sum_{t \in [x] \cap \pi_{>x}^{-1}(R_x)} \ket{t}_{D_x} \ket{\pi_{x^c}}_{D_{x^c}} \bra{+_x}_{D_x} \bra{\pi_{x^c}}_{D_{x^c}} \ket{\phi_z} }^2 \\
&\leq \frac4x \sum_{\pi_{x^c} \in S'_{x,z}} \norm[\Big]{ \parens[\Big]{ \boldsymbol{1}_{z \in R_x} \ket{\pi_{>x}^{-1}(z)}_{D_x} - \frac1x \!\!\!\!\sum_{t \in [x] \cap \pi_{>x}^{-1}(R_x)}\!\!\!\!\!\!\!\! \ket t_{D_x} } \bra{+_x}_{D_x} \bra{\pi_{x^c}}_{D_{x^c}} \ket{\phi_z} }^2 \\
&+ \frac8x \sum_{\pi_{x^c} \in S''_{x,z}}  \norm[\Big]{ \sum_{t \in [x] \cap \pi_{>x}^{-1}(R_x)} \ket{t}_{D_x} \bra{+_x}_{D_x} \bra{\pi_{x^c}}_{D_{x^c}} \ket{\phi_z} }^2 \\
&\leq \frac8x \sum_{\pi_{x^c} \in S'_{x,z}} \boldsymbol{1}_{z \in R_x} \norm[\Big]{ \bra{+_x}_{D_x} \bra{\pi_{x^c}}_{D_{x^c}} \ket{\phi_z} }^2
+ \frac8{x^3} \sum_{\pi_{x^c} \in S'_{x,z}} \sum_{t \in [x] \cap \pi_{>x}^{-1}(R_x)} \norm[\Big]{ \bra{+_x}_{D_x} \bra{\pi_{x^c}}_{D_{x^c}} \ket{\phi_z} }^2 \\
&+ \frac8x \sum_{\pi_{x^c} \in S''_{x,z}} \sum_{t \in [x] \cap \pi_{>x}^{-1}(R_x)} \norm[\Big]{ \bra{+_x}_{D_x} \bra{\pi_{x^c}}_{D_{x^c}} \ket{\phi_z} }^2 \\
&\leq \frac8x \sum_{\pi_{x^c} \in S'_{x,z}} \boldsymbol{1}_{z \in R_x} \norm[\Big]{ \bra{\pi_{x^c}}_{D_{x^c}} \ket{\phi_z} }^2
+ \frac8{x^3} \sum_{\pi_{x^c} \in S'_{x,z}} \sum_{t \in [x] \cap \pi_{>x}^{-1}(R_x)} \norm[\Big]{ \bra{\pi_{x^c}}_{D_{x^c}} \ket{\phi_z} }^2 \\
&+ \frac8x \sum_{\pi_{x^c} \in S''_{x,z}} \sum_{t \in [x] \cap \pi_{>x}^{-1}(R_x)} \norm[\Big]{ \bra{\pi_{x^c}}_{D_{x^c}} \ket{\phi_z} }^2.
\end{align*}
By summing the above estimate over all~$z>x$, we obtain
\begin{align*}
&\quad \sum_{z=x+1}^N \norm*{ E^{R,x}_D \OinvSPOz_{YD} \proj{+_x}_{D_x} \ket{\phi_z} }^2 \\
&\leq \frac8x \sum_{z=x+1}^N \sum_{\pi_{x^c} \in S'_{x,z}} \boldsymbol{1}_{z \in R_x} \norm[\Big]{ \bra{\pi_{x^c}}_{D_{x^c}} \ket{\phi_z} }^2
+ \frac8{x^3} \sum_{z=x+1}^N \sum_{\pi_{x^c} \in S'_{x,z}} \sum_{t \in [x] \cap \pi_{>x}^{-1}(R_x)} \norm[\Big]{ \bra{\pi_{x^c}}_{D_{x^c}} \ket{\phi_z} }^2 \\
&+ \frac8x \sum_{z=x+1}^N \sum_{\pi_{x^c} \in S''_{x,z}} \sum_{t \in [x] \cap \pi_{>x}^{-1}(R_x)} \norm[\Big]{ \bra{\pi_{x^c}}_{D_{x^c}} \ket{\phi_z} }^2 \\
% &\leq \frac8x \sum_{z=x+1}^N \sum_{\pi_{x^c} \in S'_{x,z}} \boldsymbol{1}_{z \in R_x} \norm[\Big]{ \bra{\pi_{x^c}}_{D_{x^c}} \ket{\phi_z} }^2
% + \frac8{x^3} \sum_{\pi_{x^c}} \sum_{t \in [x] \cap \pi_{>x}^{-1}(R_x)} \norm[\Big]{ \bra{\pi_{x^c}}_{D_{x^c}} \ket{\phi} }^2 \\
% &+ \frac8x \sum_{z=x+1}^N \sum_{\pi_{x^c} \in S''_{x,z}} \sum_{t \in [x] \cap \pi_{>x}^{-1}(R_x)} \norm[\Big]{ \bra{\pi_{x^c}}_{D_{x^c}} \ket{\phi_z} }^2 \\
&\leq \frac8x \sum_{z=x+1}^N \sum_{\pi_{x^c} \in S'_{x,z}} \boldsymbol{1}_{z \in R_x} \norm[\Big]{ \bra{\pi_{x^c}}_{D_{x^c}} \ket{\phi_z} }^2
 + \frac8{x^2} \frac{\abs{R_x}}N
 + \frac8x \sum_{z=x+1}^N \sum_{\pi_{x^c} \in S''_{x,z}} \sum_{t \in [x] \cap \pi_{>x}^{-1}(R_x)} \norm[\Big]{ \bra{\pi_{x^c}}_{D_{x^c}} \ket{\phi_z} }^2
\\ &= \frac8x \sum_{z=x+1}^N \sum_{\pi_{x^c} : \pi_{>x}^{-1}(z) < x} \boldsymbol{1}_{z \in R_x} \norm[\Big]{ \bra{\pi_{x^c}}_{D_{x^c}} \ket{\phi_z} }^2
+ \frac8{x^2} \frac{\abs{R_x}}N \\
&+ \frac8x \sum_{z=x+1}^N \sum_{\pi_{x^c} : \pi_{>x}^{-1}(z) = x} \abs{[x] \cap \pi_{>x}^{-1}(R_x)} \, \norm[\Big]{ \bra{\pi_{x^c}}_{D_{x^c}} \ket{\phi_z} }^2.
\end{align*}
The last estimate follows as in part~(ii).
\end{enumerate}
By combining the results of~(i), (ii), and (iii), we obtain the following bound on the right-hand side \cref{eq:inverse rhs}:
\begin{align*}
  \norm*{ E^{R,x}_D \OinvSPO_{XYD} \proj{+_x}_{D_x} \ket\phi }
% = \sqrt{ \sum_{z\in[N]} \norm*{ E^{R,x}_D \OinvSPOz_{YD} \proj{+_x}_{D_x} \ket{\phi_z} }^2 }
&\leq
\sqrt{ \begin{aligned}
&4 \frac{\abs{R_x}}x \, \norm{ \ket{\phi_x} }^2
+ \frac4x \sum_{z=1}^{x-1} \sum_{\pi_{x^c} : \pi_{>x}^{-1}(z) < x} \boldsymbol{1}_{z \in R_x} \norm[\Big]{ \bra{\pi_{x^c}}_{D_{x^c}} \ket{\phi_z} }^2
+ \frac4{x^2} \frac{\abs{R_x}}N \\
&+ \frac8x \sum_{z=x+1}^N \sum_{\pi_{x^c} : \pi_{>x}^{-1}(z) < x} \boldsymbol{1}_{z \in R_x} \norm[\Big]{ \bra{\pi_{x^c}}_{D_{x^c}} \ket{\phi_z} }^2
+ \frac8{x^2} \frac{\abs{R_x}}N \\
&+ \frac8x \sum_{z=x+1}^N \sum_{\pi_{x^c} : \pi_{>x}^{-1}(z) = x} \abs{[x] \cap \pi_{>x}^{-1}(R_x)} \, \norm[\Big]{ \bra{\pi_{x^c}}_{D_{x^c}} \ket{\phi_z} }^2 \end{aligned}} \\
&\leq
4 \sqrt{ \begin{aligned}
&\frac{\abs{R_x}}x \, \norm{ \ket{\phi_x} }^2 + \frac1{x^2} \frac{\abs{R_x}}N
+ \frac1x \sum_{z=1}^N \sum_{\pi_{x^c} : \pi_{>x}^{-1}(z) < x} \boldsymbol{1}_{z \in R_x} \norm[\Big]{ \bra{\pi_{x^c}}_{D_{x^c}} \ket{\phi_z} }^2 \\
&+ \frac1x \sum_{z=x+1}^N \sum_{\pi_{x^c} : \pi_{>x}^{-1}(z) = x} \abs{[x] \cap \pi_{>x}^{-1}(R_x)} \, \norm[\Big]{ \bra{\pi_{x^c}}_{D_{x^c}} \ket{\phi_z} }^2 \end{aligned}} \\
&\leq
4 \sqrt{ \begin{aligned}
&\frac{\abs{R_x}}x \, \norm{ \ket{\phi_x} }^2 + \frac1{x^2} \frac{\abs{R_x}}N
+ \frac1x \sum_{z \in R_x} \sum_{\pi_{x^c} : \pi_{>x}^{-1}(z) < x} \norm[\Big]{ \bra{\pi_{x^c}}_{D_{x^c}} \ket{\phi_z} }^2 \\
&+ \frac1x \sum_{z=x+1}^N \sum_{\pi_{x^c} : \pi_{>x}^{-1}(z) = x} \abs{[x] \cap \pi_{>x}^{-1}(R_x)} \, \norm[\Big]{ \bra{\pi_{x^c}}_{D_{x^c}} \ket{\phi_z} }^2 \end{aligned}}.
\end{align*}
Now the lemma follows.
\end{proof}

As a consequence of the preceding lemmas, we obtain a bound that holds for any query algorithm that makes $q$~queries.
%\MW{Instead of taking the sum of both terms, I am now picking our the right one. This makes the proof of \cref{prop:hard-database} nicer.}

\begin{prop}\label{prop:progress}
Let $\mathcal A$ be a unitary query algorithm on registers~$AXY$, where~$X$ and~$Y$ are $N$-dimensional registers, that gets query access to two oracles that each act on~$XY$.
Let $\ket\phi_{AXYD}$ be the joint state of algorithm and oracle defined by running $\Init^\SPO_D$ and then $\mathcal A^{\SPO_D}$. %\OSPO_{XYD},\OinvSPO_{XYD}}$.
Suppose that~$\mathcal A$ makes in total~$q$ queries to its oracles.
Then for any $x\in[N]$,
\begin{align*}
    \norm*{ E^{R,x}_D \ket\phi }
&\leq \sum_{j=1}^q \parens*{
    \sqrt{\frac{\abs{R_x}}x} \ \norm*{ \parens[\big]{ I - \proj{+_x}_{D_x} } \ket{\phi^{(j)} } }
+ 4 \sqrt{\uglyterm^{(j)}_x } } \\
&\leq \sqrt{2q \sum_{j=1}^q \parens*{
    \frac{\abs{R_x}}x \ \norm*{ \parens[\big]{ I - \proj{+_x}_{D_x} } \ket{\phi^{(j)}} }^2
+ 16 \uglyterm^{(j)}_x } },
\end{align*}
where $\ket{\phi^{(j)}}_{AXYD}$ denotes the state right before the $j$-th query of~$\mathcal A^{\SPO_D}$ and where we defined~$\uglyterm^{(j)}_x := \uglyterm_{\ket{\phi^{(j)}},x}$ as in \cref{lem:forward} if the $j$-th query is a forward query and otherwise $\uglyterm^{(j)}_x := \uglyterm^{\mathrm{inv}}_{\ket{\phi^{(j)}},x}$, as in \cref{lem:inverse}.
\end{prop}
% \begin{prop}\label{prop:progress}
% Let $\mathcal A$ be a unitary query algorithm that gets query access to two oracles that each act on two $N$-dimensional registers~$X$ and~$Y$.
% Suppose that~$\mathcal A$ makes in total~$q$ queries to its oracles.
% Let $\ket\phi_{AXYD}$ be the joint state of algorithm and oracle defined by running $\Init^\SPO_D$ and then $\mathcal A^{\SPO_D}$. %\OSPO_{XYD},\OinvSPO_{XYD}}$.
% Then for any $x\in[N]$,
% \begin{align*}
%     \norm*{ E^{R,x}_D \ket\phi }
% &\leq \sum_{j=1}^q \parens*{
%     \sqrt{\frac{\abs{R_x}}x} \ \norm*{ \parens[\big]{ I - \proj{+_x}_{D_x} } \ket{\phi^{(j)} } }
% + \sqrt{4 \uglyterm_{\ket{\phi^{(j)}}, x} + 16 \uglyterm^{\mathrm{inv}}_{\ket{\phi^{(j)},x}} } } \\
% &\leq \sqrt{2q \sum_{j=1}^q \parens*{
%     \frac{\abs{R_x}}x \ \norm*{ \parens[\big]{ I - \proj{+_x}_{D_x} } \ket{\phi^{(j)}} }^2
% + 4 \uglyterm_{\ket{\phi^{(j)}}, x} + 16 \uglyterm^{\mathrm{inv}}_{\ket{\phi^{(j)},x}} } },
% \end{align*}
% where $\ket{\phi^{(j)}}_{AXYD}$ denotes the state right before the $j$-th query of~$\mathcal A^{\SPO_D}$ and where $\uglyterm_{\ket{\phi^{(j)}},x}$ and $\uglyterm^{\mathrm{inv}}_{\ket{\phi^{(j)}},x}$ are defined as in \cref{lem:forward,lem:inverse}, respectively.
% \end{prop}
\begin{proof}
We will prove the first inequality, since the second follows directly using the Cauchy-Schwartz inequality.
To this end, we first observe that $\norm{\bra\pi_D \ket{\phi^{(j)}}_{AXYD}}^2 = \frac1{N!}$ for every~$j\in[q]$ and $\pi\in S_N$.
This is because the database is initialized in a uniform superposition of the basis states~$\ket\pi_D$, the query unitaries~$\OSPO_{XYD}$ and $\OinvSPO_{XYD}$ are controlled on~$D$ in this basis and hence commute with a basis measurement, and all other unitaries applied by~$\mathcal A$ only act on registers~$AXY$.
Thus the states $\ket{\phi^{(j)}}_{AXYD}$ satisfy the requirements of \cref{lem:forward,lem:inverse} and hence we see that, for $ Q^{\SPO} \in \{ \OSPO, \OinvSPO \}$ and any~$x\in[N]$,
\begin{align}
% \nonumber
    \norm*{ E^{R,x}_D  Q^{\SPO}_{XYD} \ket{\phi^{(j)}} } - \norm*{ E^{R,x}_D \ket{\phi^{(j)}} }
&\leq \sqrt{\frac{\abs{R_x}}x} \ \norm*{ \parens[\big]{ I - \proj{+_x}_{D_x} } \ket{\phi^{(j)}} }
+ 4 \sqrt{\uglyterm^{(j)}_x}
% + \max \{ 2 \sqrt{\uglyterm_{\ket{\phi^{(j)}}, x}}, 4 \sqrt{ \uglyterm^{\mathrm{inv}}_{\ket{\phi^{(j)}},x}} \} \\
\label{eq:progress-1-unified}
% &\leq \sqrt{\frac{\abs{R_x}}x} \ \norm*{ \parens[\big]{ I - \proj{+_x}_{D_x} } \ket{\phi^{(j)}} }
% + \sqrt{4 \uglyterm_{\ket{\phi^{(j)}}, x}
% + 16 \uglyterm^{\mathrm{inv}}_{\ket{\phi^{(j)}},x}}
\end{align}
We will now show the following inequality for every~$k \in \{0,1,\dots,q\}$, with $\ket{\phi^{(q+1)}} := \ket\phi$:
\begin{align}\label{eq:desired ineq k}
    \norm*{ E^{R,x}_D \ket{\phi^{(k+1)}} }
&\leq \sum_{j=1}^k \parens*{
    \sqrt{\frac{\abs{R_x}}x} \ \norm*{ \parens[\big]{ I - \proj{+_x}_{D_x} } \ket{\phi^{(j)} } }
+ 4 \sqrt{\uglyterm^{(j)}_x } },
% + \sqrt{4 \uglyterm_{\ket{\phi^{(j)}}, x} + 16 \uglyterm^{\mathrm{inv}}_{\ket{\phi^{(j)},x}} } },
\end{align}
This will conclude the proof, since for~$k=q$ it is the desired inequality.
We use induction over~$k$.
For~$k=0$, the inequality holds trivially, since all database registers are initialized in a uniform superposition and hence~$(I - \proj{+_x}_{D_x}) \ket{\phi} = 0$ and~$E^{R,x}_D \ket{\phi} = 0$.
For the induction step, note that for any~$k>0$~we~can~write
\begin{align*}
    E^{R,x}_D \ket{\psi^{(k+1)}}_{AXYD}
= E^{R,x}_D U_{AXY} {Q}^{\SPO}_{XYD}\ket{\psi^{(k)}}_{AXYD}
= U_{AXY} E^{R,x}_D {Q}^{\SPO}_{XYD} \ket{\psi^{(k)}}_{AXYD},
\end{align*}
for some unitary $U_{AXY}$, where $ Q^{\SPO} \in \{ \OSPO, \OinvSPO \}$ depending on $\mathcal A$'s choice for the $k$-th query.
Here we have used that the operators $U_{AXY}$ and $E^{R,x}_D$ act on disjoint registers.
Using first the unitary invariance of the norm, then \cref{eq:progress-1-unified}, and finally the induction hypothesis, we get
\begin{align*}
  \norm*{ E^{R,x}_D \ket{\phi^{(k+1)}} }
% = \norm*{ U_{AXY} E^{R,x}_D {Q}^{\SPO}_{XYD} \ket{\psi^{(k)}}_{AXYD} }
&= \norm*{ E^{R,x}_D {Q}^{\SPO}_{XYD} \ket{\psi^{(k)}}_{AXYD} } \\
&\leq \norm*{ E^{R,x}_D \ket{\phi^{(k)}} }
+ \sqrt{\frac{\abs{R_x}}x} \ \norm*{ \parens[\big]{ I - \proj{+_x}_{D_x} } \ket{\phi^{(k)}} }
+ 4 \sqrt{\uglyterm^{(k)}_{\ket{\phi^{(k)}}, x} } \\
% + \sqrt{4 \uglyterm_{\ket{\phi^{(k)}}, x} + 16 \uglyterm^{\mathrm{inv}}_{\ket{\phi^{(k)}},x}} \\
&\leq \sum_{j=1}^k \parens*{
    \sqrt{\frac{\abs{R_x}}x} \ \norm*{ \parens[\big]{ I - \proj{+_x}_{D_x} } \ket{\phi^{(j)} } }
+ 4 \sqrt{\uglyterm^{(j)}_x } }.
% + \sqrt{4 \uglyterm_{\ket{\phi^{(j)}}, x} + 16 \uglyterm^{\mathrm{inv}}_{\ket{\phi^{(j)},x}} } }.
\end{align*}
This concludes the proof of \cref{eq:desired ineq k} and hence proof of the corollary.
\end{proof}

%-----------------------------------------------------------------------------
\subsection{Bounding the Success Probability in Expectation}\label{sec:expectation}
%-----------------------------------------------------------------------------
As discussed at the beginning of the section we will now lift the worst-case bounds established in the previous section to the average case by running the query algorithm with the \emph{twirled} permutation oracle for uniformly random~$\sigma,\tau \in S_N$ and also averaging over the choice of~$x\in[N]$, corresponding to the progress measure
\begin{align*}
    \E_{\substack{x \leftarrow [N], \\ \sigma, \tau \leftarrow S_N}} \norm*{ E^{R^{\sigma,\tau},x}_D \ket{\phi^{\sigma,\tau}} }^2,
\end{align*}
where $R^{\sigma,\tau}$ is the twirled relation defined in \eqref{eq:twirled relation}.
Note that $r_{\max} = \max \{ \max_{x\in[N]} \abs{R^{\sigma,\tau}_x}, \max_{y\in[N]} \abs{(R^{\sigma,\tau})^\mathrm{inv}_y} \}$.

The main result of this section is the following.
It shows that the randomized test for the twirled relation rarely succeeds for algorithms making not too many queries, up to an error term that captures the ``sparsity'' of the database and that will be bounded in \cref{sec:bounc_active}.

\begin{prop}\label{prop:hard-database}
Let $\mathcal A$ be a unitary query algorithm on registers~$AXY$, where~$X$ and~$Y$ are $N$-dimensional registers, that gets query access to two oracles that each act on~$XY$.
Suppose that~$\mathcal A$ makes in total~$q$ queries to its oracles.
For every~$\sigma,\tau\in S_N$, let~$\ket{\phi^{\sigma,\tau}}_{AXYD}$ be the joint state of algorithm and oracle defined by running~$\Init^\SPO_D$ and then~$\mathcal A^{\FTSPOsigmatau_D}$.
Let~$\mathcal B$ be the query algorithm on registers~$BXY$ as in \cref{lem:std preprocess}, with~$B=AZ$ and~$Z$ another $N$-dimensional register.
Then,
\begin{align*}
    \E_{\substack{x \leftarrow [N], \\ \sigma, \tau \leftarrow S_N}} \norm*{ E^{R^{\sigma,\tau},x}_D \ket{\phi^{\sigma,\tau}} }^2
&\leq 384 \frac {q^2 r_{\max} \parens[\big]{ \ln(N) + 2 }} {N^2}
% + 4q \sum_{j=1}^{2q} \E_{\substack{x \leftarrow [N], \\ \sigma, \tau \leftarrow S_N}} \frac{\abs{R^{\sigma,\tau}_x}}x \ \norm*{ \parens[\big]{ I - \proj{+_x}_{D_x} } \ket{\phi^{\sigma,\tau,(j)}} }^2
+ 4q r_{\max} \sum_{j=1}^{2q} \E_{\substack{x \leftarrow [N], \\ \sigma, \tau \leftarrow S_N}} \frac {\norm*{ \parens[\big]{ I - \proj{+_x}_{D_x} } \ket{\phi^{\sigma,\tau,(j)}} }^2}x
\end{align*}
where $\ket{\phi^{\sigma,\tau,(j)}}_{BXYD}$ is the state given by running~$\Init^\SPO_D$ and then~$\mathcal B^{\FTSPOsigmatau_D}$ until right before the~$j$-th query.
\end{prop}
\begin{proof}
By \cref{lem:std preprocess} for every $\sigma, \tau\in S_N$ it holds that
\begin{align*}
  \mathcal B^{\FTSPOsigmatau_D}
% = \mathcal C^{\OFTSPOsigmatau_{XYD},\; \OFTSPOsigmatau_{XYD},\;  \OinvFTSPOsigmatau_{XYD},\; \OinvFTSPOsigmatau_{XYD}}
= \mathcal B_{\sigma,\tau}^{\SPO_D}.
\end{align*}
with~$\mathcal B_{\sigma,\tau}$ defined in the statement of the lemma.
Moreover, the joint state of the algorithm and oracle defined by running $\Init^\SPO_D$ and then either of these two algorithms is given by~$\ket{\chi^{\sigma,\tau}}_{BXYD} = \ket{\phi^{\sigma,\tau}}_{AXYD} \ot \ket0_Z$.
Thus we can apply \cref{prop:progress} with the relation~$R^{\sigma,\tau}$ and the algorithm~$\mathcal B_{\sigma,\tau}$, which makes~$2q$ queries to the untwirled standard oracle, to obtain
\begin{align}\label{eq:hard B sigma tau bound}
    \norm*{ E^{R^{\sigma,\tau},x}_D \ket{\phi^{\sigma,\tau}} }^2
= \norm*{ E^{R^{\sigma,\tau},x}_D \ket{\chi^{\sigma,\tau}} }^2
\leq 4q \sum_{j=1}^{2q} \parens*{
    \frac{\abs{R^{\sigma,\tau}_x}}x \ \norm*{ \parens[\big]{ I - \proj{+_x}_{D_x} } \ket{\chi^{\sigma,\tau,(j)}} }^2
+ 16 \uglyterm^{\sigma,\tau,(j)}_x },
\end{align}
where $\ket{\chi^{\sigma,\tau,(j)}}_{BXYD}$ denotes the state right before the $j$-th query of~$\mathcal B_{\sigma,\tau}^{\SPO_D}$ and
\begin{align*}
  \uglyterm^{\sigma,\tau,(j)}_x
= \begin{cases}
\uglyterm_{\ket{\chi^{\sigma,\tau,(j)}}, x} & \text{if the $j$-th query of~$\mathcal B_{\sigma,\tau}$ is a forward query}, \\
\uglyterm^{\mathrm{inv}}_{\ket{\chi^{\sigma,\tau,(j)}}, x} & \text{if the $j$-th query of~$\mathcal B_{\sigma,\tau}$ is an inverse query};
\end{cases}
\end{align*}
the right-hand side quantities are defined in \cref{lem:forward,lem:inverse}.
In view of the relation between~$\mathcal B_{\sigma,\tau}$ and~$\mathcal B$ in \cref{lem:std preprocess}, we can express the pre-query states of the former in terms of the pre-query states of the latter: we have~$\ket{\chi^{\sigma,\tau,(j)}} \in \braces*{
V_X^\sigma \ket{\phi^{\sigma,\tau,(j)}},
V_X^\sigma V_Y^\tau \ket{\phi^{\sigma,\tau,(j)}}
}$
if the $j$-th query of~$\mathcal B_{\sigma,\tau}$ is a forward query, and otherwise~$\ket{\chi^{\sigma,\tau,(j)}} \in \braces*{
V_X^\tau \ket{\phi^{\sigma,\tau,(j)}},
V_X^\tau V_Y^\sigma \ket{\phi^{\sigma,\tau,(j)}}
}$.
Using the unitary invariance of the norm, we see that
\begin{align}\label{eq:hard norm inv}
  \norm*{ \parens[\big]{ I - \proj{+_x}_{D_x} } \ket{\chi^{\sigma,\tau,(j)}} }
= \norm*{ \parens[\big]{ I - \proj{+_x}_{D_x} } \ket{\phi^{\sigma,\tau,(j)}} }
\end{align}
and
\begin{align*}
  \uglyterm^{\sigma,\tau,(j)}_x
= \begin{cases}
\uglyterm_{V_X^\sigma \, \ket{\phi^{\sigma,\tau,(j)}}, x} & \text{if the $j$-th query of~$\mathcal B$ is a forward query}, \\
\uglyterm^{\mathrm{inv}}_{V_X^\tau \, \ket{\phi^{\sigma,\tau,(j)}}, x} & \text{if the $j$-th query of~$\mathcal B$ is an inverse query}.
\end{cases}
\end{align*}
To prove the proposition, we need to upper bound $\E_{x,\sigma,\tau} \uglyterm^{\sigma,\tau,(j)}_{\ket{\chi^{\sigma,\tau,(j)}}, x}$.
We distinguish the two cases:
\begin{enumerate}
\item If the~$j$-th query is a forward query, we have
\begin{align*}
  \uglyterm^{\sigma,\tau,(j)}_x
= \uglyterm_{V_X^\sigma \, \ket{\phi^{\sigma,\tau,(j)}}, x}
= \frac{\abs{R^{\sigma,\tau}_x}}x \, \norm{ \ket{\phi^{\sigma,\tau,(j)}_{\sigma^{-1}(x)}} }^2
    + \frac{\abs{R^{\sigma,\tau}_x}}{x^2 N}
    + \frac1x \sum_{z=1}^{x-1} \sum_{\pi_{x^c} : \pi_{x^c}(z) \in R^{\sigma,\tau}_x} \norm[\Big]{ \bra{\pi_{x^c}}_{D_{x^c}} \ket{\phi^{\sigma,\tau,(j)}_{\sigma^{-1}(z)}} }^2,
\end{align*}
where $\ket{\phi^{\sigma,\tau,(j)}_\xi}_{BYD} := \bra \xi_X \ket{\phi^{\sigma,\tau,(j)}}_{BXYD}$ (as in \cref{lem:forward}).
We will now upper bound the average of the above term by term, beginning with the first term.
Now, $p^{(j)}_\xi := \norm{ \ket{\phi^{\sigma,\tau,(j)}_\xi} }^2$ is a function of the reduced state of an algorithm that makes queries to the twirled standard oracle (the part of~$\mathcal B$ right up to the $j$-th query), so it follows from \cref{lem:spo vs ftspo} %\gm{Double check this reference} \MW{Could you elaborate? Is it wrong or should we explain better?} 
that this quantity is independent of~$\sigma,\tau\in S_N$.
Thus:
\begin{align}
\nonumber
  \E_{x,\sigma,\tau} \frac{\abs{R^{\sigma,\tau}_x}}x \, \norm{ \ket{\phi^{\sigma,\tau,(j)}_{\sigma^{-1}(x)}} }^2
% = \E_{x,\sigma,\tau} \frac{\abs{R^{\sigma,\tau}_x}}x \, p^{(j)}_{\sigma^{-1}(x)}
&\leq \E_{x,\sigma,\tau} \frac{r_{\max}}x \, p^{(j)}_{\sigma^{-1}(x)}
= \E_x \frac{r_{\max}}x \, \E_{\sigma,\tau} p^{(j)}_{\sigma^{-1}(x)}
= \E_x \frac{r_{\max}}x \, \E_{x'} p^{(j)}_{x'} \\
\nonumber
&= r_{\max} \parens*{ \frac1N \sum_{x=1}^N \frac1x } \parens*{ \frac1N \sum_{x'=1}^N p^{(j)}_{x'} }
\leq r_{\max} \frac {\ln(N) + 1} N \frac 1N \\
\label{eq:hard fwd 1}
&= \parens*{ \ln(N) + 1 } \frac {r_{\max}} {N^2}
\end{align}
since $\sum_{x'=1}^N p^{(j)}_{x'} = \norm{ \ket{\phi^{\sigma,\tau,(j)}} }^2 = 1$.
The second term can be bounded straightforwardly:
\begin{align}
\label{eq:hard fwd 2}
  \E_{x,\sigma,\tau} \frac{\abs{R^{\sigma,\tau}_x}}{x^2 N}
\leq \frac{r_{\max}}N \E_x \frac1{x^2}
= \frac{r_{\max}}{N^2} \sum_{x=1}^N \frac1{x^2}
\leq \frac{\pi^2}6 \frac{r_{\max}}{N^2}
\leq 2 \frac{r_{\max}}{N^2}.
\end{align}
We defer bounding the third term to \cref{lem:crucial}~\ref{it:crucial 1}, where we get
\begin{align}
\label{eq:hard fwd 3}
    \E_{x,\sigma,\tau} \frac1x \sum_{z=1}^{x-1} \sum_{\pi_{x^c} : \pi_{x^c}(z) \in R^{\sigma,\tau}_x} \norm[\Big]{ \bra{\pi_{x^c}}_{D_{x^c}} \ket{\phi^{\sigma,\tau,(j)}_{\sigma^{-1}(z)}} }^2
\leq \parens*{ \ln(N) + 3 } \frac{r_{\max}}{N^2}.
\end{align}
Combining \cref{eq:hard fwd 1,eq:hard fwd 2,eq:hard fwd 3}, we obtain
\begin{align}\label{eq:hard fwd}
    \E_{x,\sigma,\tau} \uglyterm^{\sigma,\tau,(j)}_x
\leq \parens[\big]{ 2 \ln(N) + 6 } \frac {r_{\max}} {N^2}.
\end{align}

\item If the~$j$-th query is an inverse query, we have
\begin{align*}
  \uglyterm^{\sigma,\tau,(j)}_x
= \uglyterm^{\mathrm{inv}}_{V_X^\tau \, \ket{\phi^{\sigma,\tau,(j)}}, x}
= &\frac{\abs{R^{\sigma,\tau}_x}}x \, \norm{\ket{\phi^{\sigma,\tau,(j)}_{\tau^{-1}(x)}} }^2 + \frac{\abs{R^{\sigma,\tau}_x}}{x^2N}
        + \frac1x \sum_{z \in R^{\sigma,\tau}_x} \sum_{\pi_{x^c} : \pi_{>x}^{-1}(z) < x} \norm[\Big]{ \bra{\pi_{x^c}}_{D_{x^c}} \ket{\phi^{\sigma,\tau,(j)}_{\tau^{-1}(z)}} }^2 \\
  &+ \frac1x \sum_{z=x+1}^N \sum_{\pi_{x^c} : \pi_{>x}^{-1}(z) = x} \abs{[x] \cap \pi_{>x}^{-1}(R^{\sigma,\tau}_x)} \, \norm[\Big]{ \bra{\pi_{x^c}}_{D_{x^c}} \ket{\phi^{\sigma,\tau,(j)}_{\tau^{-1}(z)}} }^2,
\end{align*}
where $\ket{\phi^{\sigma,\tau,(j)}_\xi}_{BYD} := \bra \xi_X \ket{\phi^{\sigma,\tau,(j)}}_{BXYD}$ (as in \cref{lem:inverse}).
The average can again be bounded term by term.
For the first two terms we proceed as above and for the last two we use \cref{lem:crucial}~\ref{it:crucial 2} and~\ref{it:crucial 3}.
Altogether we obtain
\begin{equation}\label{eq:hard inv}
\begin{aligned}
    \E_{x,\sigma,\tau} \uglyterm^{\sigma,\tau,(j)}_x
&\leq \parens*{ \ln(N) + 1 } \frac {r_{\max}} {N^2} + 2 \frac{r_{\max}}{N^2} + \parens*{ \ln(N)+1 }\frac {r_{\max}} {N^2} + \parens*{ \ln(N)+1 } \frac {r_{\max}} {N^2} \\
&= \parens[\big]{ 3\ln(N) + 5 } \frac {r_{\max}} {N^2}.
\end{aligned}
\end{equation}
\end{enumerate}
From \cref{eq:hard fwd,eq:hard inv} we see that in both the forward and the inverse case, we can bound
\begin{align*}
    \E_{x,\sigma,\tau} \uglyterm^{\sigma,\tau,(j)}_x
\leq 3 \parens[\big]{ \ln(N) + 2 } \frac {r_{\max}} {N^2}
\end{align*}
We can now use the above and \cref{eq:hard norm inv} to further bound \cref{eq:hard B sigma tau bound} and obtain
\begin{align*}
    \E_{x,\sigma,\tau} \norm*{ E^{R^{\sigma,\tau},x}_D \ket{\phi^{\sigma,\tau}} }^2
&\leq 64 q \sum_{j=1}^{2q} \E_{x,\sigma,\tau} \uglyterm^{\sigma,\tau,(j)}_x + 4q \sum_{j=1}^{2q} \E_{x,\sigma,\tau} \frac{\abs{R^{\sigma,\tau}_x}}x \ \norm*{ \parens[\big]{ I - \proj{+_x}_{D_x} } \ket{\chi^{\sigma,\tau,(j)}} }^2 \\
&\leq 384 q^2 \parens[\big]{ \ln(N) + 2 } \frac {r_{\max}} {N^2} + 4q \sum_{j=1}^{2q} \E_{x,\sigma,\tau} \frac{\abs{R^{\sigma,\tau}_x}}x \ \norm*{ \parens[\big]{ I - \proj{+_x}_{D_x} } \ket{\phi^{\sigma,\tau,(j)}} }^2 \\
&= 384 \frac {q^2 r_{\max} \parens[\big]{ \ln(N) + 2 }} {N^2} + 4q r_{\max} \sum_{j=1}^{2q} \E_{x,\sigma,\tau} \frac {\norm*{ \parens[\big]{ I - \proj{+_x}_{D_x} } \ket{\phi^{\sigma,\tau,(j)}} }^2}x,
\end{align*}
which is the desired result.
\end{proof}

\begin{lem}\label{lem:crucial}
In the situation of \cref{prop:hard-database} and with $\ket{\phi^{\sigma,\tau,(j)}_\xi}_{BYD} := \bra \xi_X \ket{\phi^{\sigma,\tau,(j)}}_{BXYD}$, we have:
\begin{enumerate}
\item\label{it:crucial 1}
$\E_{x \leftarrow [N], \; \sigma,\tau\leftarrow S_N} \frac1x \sum_{z=1}^{x-1} \sum_{\pi_{x^c} : \pi_{x^c}(z) \in R^{\sigma,\tau}_x} \norm[\Big]{ \bra{\pi_{x^c}}_{D_{x^c}} \ket{\phi^{\sigma,\tau,(j)}_{\sigma^{-1}(z)}} }^2 \leq \parens*{ \ln(N) + 3 } \frac{r_{\max}}{N^2}$.
\item\label{it:crucial 2}
$\E_{x \leftarrow [N], \; \sigma,\tau\leftarrow S_N} \frac1x \sum_{z \in R^{\sigma,\tau}_x} \sum_{\pi_{x^c} : \pi_{>x}^{-1}(z) < x} \norm[\Big]{ \bra{\pi_{x^c}}_{D_{x^c}} \ket{\phi^{\sigma,\tau,(j)}_{\tau^{-1}(z)}} }^2 \leq \parens*{ \ln(N)+1 }\frac {r_{\max}} {N^2}$.
\item\label{it:crucial 3}
$\E_{x \leftarrow [N], \; \sigma,\tau\leftarrow S_N} \frac1x \sum_{z=x+1}^N \sum_{\pi_{x^c} : \pi_{>x}^{-1}(z) = x} \abs{[x] \cap \pi_{>x}^{-1}(R^{\sigma,\tau}_x)} \, \norm[\Big]{ \bra{\pi_{x^c}}_{D_{x^c}} \ket{\phi^{\sigma,\tau,(j)}_{\tau^{-1}(z)}} }^2 \leq \parens*{ \ln(N)+1 }\frac {r_{\max}} {N^2}$.
\end{enumerate}
\end{lem}
\begin{proof}
Note that the quantity
%\MW{Note that this is defined differently than $\beta$ in the old proof, because now $\ket{\phi^{\sigma,\tau,(j)}}$ refers to the pre-query states of $\mathcal B$, not of $\mathcal B_{\sigma,\tau}$. Please check most carefully that this makes seense and that the below is right.}
\begin{align*}
    q_{\omega,\xi}
:= \norm[\Big]{ \bra{\tau \omega \sigma^{-1}}_D \ket{\phi^{\sigma,\tau,(j)}_\xi} }^2
= \norm[\Big]{ \bra{\tau \omega \sigma^{-1}}_D \bra\xi_X \ket{\phi^{\sigma,\tau,(j)}} }^2
\end{align*}
can be interpreted as the joint probability of the outcomes of the following procedure:
initialize the database, run an algorithm (namely,~$\mathcal B$ up to right before its~$j$-th query) that makes queries to the twirled standard oracle, measure the~$X$ register to obtain an outcome~$\xi\in[N]$, and also apply the recovery operation to the database to obtain an outcome~$\omega\in S_N$.
Accordingly, \cref{lem:spo vs ftspo} %\gm{Double check this reference} \MW{Could you elaborate?} 
shows that $q_{\omega,\xi}$ does not depend on the choice of~$\sigma,\tau\in S_N$ (which justifies the notation) and that the marginal distribution of~$\omega$ with respect to~$q_{\omega,\xi}$ is uniform, i.e., $\sum_{\xi\in[N]} q_{\omega,\xi} = \frac1{N!}$ for all~$\omega\in S_N$.
%\MW{Agreed?}
This observation will be used to establish all three parts of the lemma.
\begin{enumerate}
\item
We start by writing
\begin{align}\label{eq:crucial 1 start}
\E_{x,\sigma,\tau} \frac1x \sum_{z=1}^{x-1} \sum_{\pi_{x^c} : \pi_{x^c}(z) \in R^{\sigma,\tau}_x} \norm[\Big]{ \bra{\pi_{x^c}}_{D_{x^c}} \ket{\phi^{\sigma,\tau,(j)}_{\sigma^{-1}(z)}} }^2
% &= \E_{x,\sigma,\tau} \frac1x \sum_{z=1}^{x-1} \sum_{\pi : \pi_{x^c}(z) \in R^{\sigma,\tau}_x} \norm[\Big]{ \bra\pi_D \ket{\phi^{\sigma,\tau,(j)}_{\sigma^{-1}(z)}} }^2 \\
% &= \E_{x,\sigma,\tau} \frac1x \sum_{z=1}^{x-1} \sum_{\pi : \pi_{x^c}(z) \in R^{\sigma,\tau}_x} \norm[\Big]{ \bra{\tau(\tau^{-1}\pi\sigma)\sigma^{-1}}_D \ket{\phi^{\sigma,\tau,(j)}_{\sigma^{-1}(z)}} }^2 \\
= \E_{x,\sigma,\tau} \frac1x \sum_{z=1}^{x-1} \sum_{\pi : \pi_{x^c}(z) \in R^{\sigma,\tau}_x} \hspace{-0.5cm} q_{\tau^{-1}\pi\sigma,\sigma^{-1}(z)}.
\end{align}
Recall that $\pi = \pi_{>x} \tp x {t_x} \pi_{>x}$ for some~$t_x \in [x]$, as in \cref{eq:tower-decomposition,eq:pi gt lt k}.
We will write~$t_x(\pi) := t_x$ to make explicit the dependency of~$t_x$ on the permutation.
Because $z < x$, we have
\begin{align*}
    \pi_{x^c}(z) = \begin{cases}
        \pi(z) & \text{ if } \pi_{<x}(z) \neq t_x(\pi), \\
        \pi(x) & \text{ if } \pi_{<x}(z) = t_x(\pi),
    \end{cases}
\end{align*}
so we can write the right-hand side of \cref{eq:crucial 1 start} as a sum of two terms,
\begin{align*}
&\quad \E_{x,\sigma,\tau} \frac1x \sum_{z=1}^{x-1} \sum_{\pi : \pi_{x^c}(z) \in R^{\sigma,\tau}_x} q_{\tau^{-1}\pi\sigma,\sigma^{-1}(z)} \\
&= \E_{x,\sigma,\tau} \frac1x \sum_{z=1}^{x-1} \sum_{\pi : \pi(z) \in R^{\sigma,\tau}_x} \hspace{-0.5cm} \boldsymbol{1}_{\pi_{<x}(z) \neq t_x(\pi)} \, q_{\tau^{-1}\pi\sigma,\sigma^{-1}(z)}
+ \E_{x,\sigma,\tau} \frac1x \sum_{z=1}^{x-1} \sum_{\pi : \pi(x) \in R^{\sigma,\tau}_x} \hspace{-0.5cm} \boldsymbol{1}_{\pi_{<x}(z) = t_x(\pi)} \, q_{\tau^{-1}\pi\sigma,\sigma^{-1}(z)} \\
&= \E_{x',\sigma,\tau} \frac1{\sigma(x')} \sum_{z'\in\sigma^{-1}([1,\sigma(x')])} \sum_{\pi' : \pi'(z') \in R_{x'}} \boldsymbol{1}_{\pi_{<\sigma(x')}(\sigma(z')) \neq t_{\sigma(x')}(\tau \pi' \sigma^{-1})} \, q_{\pi',z'} \\
&+ \E_{x',\sigma,\tau} \frac1{\sigma(x')} \sum_{z'\in\sigma^{-1}([1,\sigma(x')])} \sum_{\pi' : \pi'(x') \in R_{x'}} \boldsymbol{1}_{\pi_{<\sigma(x')}(\sigma(z')) = t_{\sigma(x')}(\tau \pi' \sigma^{-1})} \, q_{\pi',z'},
\end{align*}
where the last step follows by substituting~$x = \sigma(x')$, $z = \sigma(z')$, and $\pi = \tau \pi' \sigma^{-1}$, noting that~$(x',\sigma,\tau)$ is still uniformly random, and using the relation~$\pi(\xi) \in R^{\sigma,\tau}_x \Leftrightarrow \tau^{-1}(\pi(\xi)) \in R_{\sigma^{-1}(x)}$.
As $\tau$ only appears in the indicator functions, we can rewrite and bound this as
\begin{align*}
&\quad \E_{x',\sigma} \frac1{\sigma(x')} \sum_{z'\in\sigma^{-1}([1,\sigma(x')])} \sum_{\pi' : \pi'(z') \in R_{x'}} \Pr_\tau\mleft(\pi_{<\sigma(x')}(\sigma(z')\mright) \neq t_{\sigma(x')}(\tau \pi' \sigma^{-1})) \, q_{\pi',z'} \\
&+ \E_{x',\sigma} \frac1{\sigma(x')} \sum_{z'\in\sigma^{-1}([1,\sigma(x')])} \sum_{\pi' : \pi'(x') \in R_{x'}} \Pr_\tau\mleft(\pi_{<\sigma(x')}(\sigma(z')\mright) = t_{\sigma(x')}(\tau \pi' \sigma^{-1})) \, q_{\pi',z'} \\
&= \E_{x',\sigma} \frac1{\sigma(x')} \sum_{z'\in\sigma^{-1}([1,\sigma(x')])} \sum_{\pi' : \pi'(z') \in R_{x'}} \Pr_{t' \leftarrow [\sigma(x')]}\mleft(\pi_{<\sigma(x')}(\sigma(z')) \neq t'\mright) \, q_{\pi',z'} \\
&+ \E_{x',\sigma} \frac1{\sigma(x')} \sum_{z'\in\sigma^{-1}([1,\sigma(x')])} \sum_{\pi' : \pi'(x') \in R_{x'}} \Pr_{t' \leftarrow [\sigma(x')]}\mleft(\pi_{<\sigma(x')}(\sigma(z')) = t'\mright) \, q_{\pi',z'} \\
&\leq \E_{x',\sigma} \frac1{\sigma(x')} \sum_{z'\in\sigma^{-1}([1,\sigma(x')])} \sum_{\pi' : \pi'(z') \in R_{x'}} q_{\pi',z'}
+ \E_{x',\sigma} \frac1{\sigma(x')} \sum_{z'\in\sigma^{-1}([1,\sigma(x')])} \sum_{\pi' : \pi'(x') \in R_{x'}} \frac1{\sigma(x')} q_{\pi',z'} \\
&= \E_{x',\sigma} \frac1{\sigma(x')} \sum_{z'\in\sigma^{-1}([1,\sigma(x')])} \sum_{\pi' : \pi'(z') \in R_{x'}} q_{\pi',z'}
+ \E_{x',\sigma} \frac1{(\sigma(x'))^2} \sum_{z'\in\sigma^{-1}([1,\sigma(x')])} \sum_{\pi' : \pi'(x') \in R_{x'}} q_{\pi',z'},
\end{align*}
since, for any fixed $x',z',\sigma,\pi'$, the permutation $\tau \pi' \sigma^{-1}$ is uniformly random in~$S_N$, so~$t_{\sigma(x')}(\tau \pi' \sigma^{-1})$ is uniformly random in~$[\sigma(x')]$ (by \cref{cor:independent}) and hence equal to any fixed integer in this interval with probability~$\frac1{\sigma(x')}$.
We can finally upper bound the above by
\begin{align*}
&\quad \E_{x',\sigma} \frac1{\sigma(x')} \sum_{z'\in\sigma^{-1}([1,\sigma(x')])} \sum_{\pi' : \pi'(z') \in R_{x'}} q_{\pi',z'}
+ \E_{x',\sigma} \frac1{(\sigma(x'))^2} \sum_{z'\in\sigma^{-1}([1,\sigma(x')])} \sum_{\pi' : \pi'(x') \in R_{x'}} q_{\pi',z'} \\
&\leq \E_{x',\sigma} \frac1{\sigma(x')} \sum_{z'=1}^N \sum_{\pi' : \pi'(z') \in R_{x'}} q_{\pi',z'}
+ \E_{x',\sigma} \frac1{(\sigma(x'))^2} \sum_{z'=1}^N \sum_{\pi' : \pi'(x') \in R_{x'}} q_{\pi',z'} \\
&= \E_{x'} \parens*{ \E_\sigma \frac1{\sigma(x')} } \sum_{z'=1}^N \sum_{\pi' : \pi'(z') \in R_{x'}} q_{\pi',z'}
+ \E_{x'} \parens*{ \E_\sigma \frac1{{\sigma(x')}^2} } \sum_{z'=1}^N \sum_{\pi' : \pi'(x') \in R_{x'}} q_{\pi',z'} \\
&= \E_{x'} \parens*{ \E_\sigma \frac1{\sigma(x')} } \sum_{z'=1}^N \sum_{\pi' : \pi'(z') \in R_{x'}} q_{\pi',z'}
+ \E_{x'} \parens*{ \E_\sigma \frac1{{\sigma(x')}^2} } \sum_{z'=1}^N \sum_{\pi' : \pi'(x') \in R_{x'}} q_{\pi',z'} \\
&\leq \frac{\ln(N)+1}N \E_{x'} \sum_{z'=1}^N \sum_{\pi' : \pi'(z') \in R_{x'}} q_{\pi',z'}
+ \frac{\pi^2}6 \frac 1N \E_{x'} \sum_{z'=1}^N \sum_{\pi' : \pi'(x') \in R_{x'}} q_{\pi',z'} \\
&= \frac{\ln(N)+1}N \sum_{\pi' \in S_N} \sum_{z'=1}^N \Pr_{x'}\mleft( x' \in R^\mathrm{inv}_{\pi'(z)} \mright) q_{\pi',z'}
+ \frac{\pi^2}6 \frac 1N \E_{x'} \Pr_{\pi'}(\pi'(x') \in R_{x'}) \\
&\leq \frac{\ln(N)+1}N \frac{r_{\max}}N + \frac{\pi^2}6 \frac 1N \frac{r_{\max}}N \\
&\leq \parens*{ \ln(N) + 3 } \frac{r_{\max}}{N^2},
\end{align*}
where we first enlarging the sum over~$z'$ to all of~$[N]$, then we bounded the expectation over~$\sigma$ by using that~$\sigma(x') \in [N]$ is uniformly random for any fixed~$x'$; in the last equality we also used that the marginal distribution of~$\pi'$ with respect to~$q_{\pi',z'}$ is uniform as discussed above.

\item Similarly as above, we begin by writing
\begin{align*}
  \E_{x,\sigma,\tau} \frac1x \sum_{z \in R^{\sigma,\tau}_x} \sum_{\pi_{x^c} : \pi_{>x}^{-1}(z) < x} \norm[\Big]{ \bra{\pi_{x^c}}_{D_{x^c}} \ket{\phi^{\sigma,\tau,(j)}_{\tau^{-1}(z)}} }^2
% = \E_{x,\sigma,\tau} \frac1x \sum_{z \in R^{\sigma,\tau}_x} \sum_{\pi : \pi_{>x}^{-1}(z) < x} \norm[\Big]{ \bra\pi_D \ket{\phi^{\sigma,\tau,(j)}_{\tau^{-1}(z)}} }^2
% = \E_{x,\sigma,\tau} \frac1x \sum_{z \in R^{\sigma,\tau}_x} \sum_{\pi : \pi_{>x}^{-1}(z) < x} \norm[\Big]{ \bra{\tau(\tau^{-1}\pi\sigma)\sigma^{-1}}_D \ket{\phi^{\sigma,\tau,(j)}_{\tau^{-1}(z)}} }^2
= \E_{x,\sigma,\tau} \frac1x \sum_{z \in R^{\sigma,\tau}_x} \sum_{\pi : \pi_{>x}^{-1}(z) < x} q_{\tau^{-1}\pi\sigma,\tau^{-1}(z)}
\end{align*}
We can upper-bound this by omitting the constraint on~$\pi$, which gives the bound
\begin{align*}
\quad \E_{x,\sigma,\tau} \frac1x \sum_{z \in R^{\sigma,\tau}_x} \sum_{\pi : \pi_{>x}^{-1}(z) < x} q_{\tau^{-1}\pi\sigma,\tau^{-1}(z)}
&\leq \E_{x,\sigma,\tau} \frac1x \sum_{z \in R^{\sigma,\tau}_x} q_{\tau^{-1}(z)}
= \E_{x',\sigma,\tau} \frac1{\sigma(x')} \sum_{z' \in R_{x'}} q_{z'} \\
&= \E_{x'} \parens*{ \E_{\sigma} \frac1{\sigma(x')} } \sum_{z' \in R_{x'}} q_{z'}
\leq \frac{\ln(N)+1}N \E_{x'} \sum_{z' \in R_{x'}} q_{z'} \\
&= \frac{\ln(N)+1}N \sum_{z'=1}^N \Pr_{x'}\mleft( x' \in R^{\mathrm{inv}}_{z'} \mright) q_{z'}
\leq \parens*{ \ln(N)+1 }\frac {r_{\max}} {N^2},
\end{align*}
where we use the notation~$q_\xi := \sum_{\omega \in S_N} q_{\omega,\xi}$ for the marginal distribution of~$\xi$ with respect to~$q_{\omega,\xi}$;
the second step follows by substituting~$x = \sigma(x')$ and~$z = \tau(z')$.
%\MW{Note that the $\tau$ just cancels (in contrast to the preceding proof), which I guess makes sense as a consequence of what I wrote above. Do you agree?}
%\MW{Super strange that I substitute $x = \sigma(x')$ here, but $x = \tau(x')$ below in (iii)??? Am I doing something wrong or is this term so low probability that we don't have to work hard???}

\item Again we begin with
\begin{align*}
    &\quad \E_{x,\sigma,\tau} \frac1x \sum_{z=x+1}^N \sum_{\pi_{x^c} : \pi_{>x}^{-1}(z) = x} \abs{[x] \cap \pi_{>x}^{-1}(R^{\sigma,\tau}_x)} \, \norm[\Big]{ \bra{\pi_{x^c}}_{D_{x^c}} \ket{\phi^{\sigma,\tau,(j)}_{\tau^{-1}(z)}} }^2 \\
% &= \E_{x,\sigma,\tau} \frac1x \sum_{z=x+1}^N \sum_{\pi : \pi_{>x}^{-1}(z) = x} \abs{[x] \cap \pi_{>x}^{-1}(R^{\sigma,\tau}_x)} \, \norm[\Big]{ \bra\pi_D \ket{\phi^{\sigma,\tau,(j)}_{\tau^{-1}(z)}} }^2 \\
&= \E_{x,\sigma,\tau} \frac1x \sum_{z=x+1}^N \sum_{\pi : \pi_{>x}^{-1}(z) = x} \abs{[x] \cap \pi_{>x}^{-1}(R^{\sigma,\tau}_x)} \, q_{\tau^{-1}\pi\sigma,\tau^{-1}(z)} \\
&\leq r_{\max} \E_{x,\sigma,\tau} \frac1x \sum_{z=x+1}^N \sum_{\pi : \pi_{>x}^{-1}(z) = x} q_{\tau^{-1}\pi\sigma,\tau^{-1}(z)} \\
&= r_{\max} \E_{x',\sigma,\tau} \frac1{\tau(x')} \sum_{z' \in \tau^{-1}(\{\tau(x')+1,\dots,N\})} \sum_{\pi'} \boldsymbol{1}_{((\tau \pi' \sigma^{-1})_{>\tau(x')})^{-1}(\tau(z')) = \tau(x')} q_{\pi',z'},
\end{align*}
where the last step follows by substituting~$x = \tau(x')$, $z = \tau(z')$, and $\pi = \tau \pi' \sigma^{-1}$.
As $\sigma$ only occurs in the indicator function, we can rewrite and bound this as
\begin{align}
\nonumber
&\quad r_{\max} \E_{x',\tau} \frac1{\tau(x')} \sum_{z' \in \tau^{-1}(\{\tau(x')+1,\dots,N\})} \sum_{\pi'} \Pr_\sigma\mleft( ((\tau \pi' \sigma^{-1})_{>\tau(x')})^{-1}(\tau(z')) = \tau(x') \mright) q_{\pi',z'} \\
\label{eq:crucial 3 intermediate}
&= r_{\max} \E_{x',\tau} \frac1{\tau(x')} \sum_{z' \in \tau^{-1}(\{\tau(x')+1,\dots,N\})} \sum_{\pi'} \Pr_\sigma\mleft( \sigma_{>\tau(x')}(\tau(x')) = \tau(z') \mright) q_{\pi',z'}
\end{align}
since, for any fixed $\tau$ and $\pi'$, the permutation~$\tau \pi' \sigma^{-1}$ is again uniformly random.
Using part~\ref{it:wrong side 3} of \cref{lem:wrong side}, we see that the inner probability is simply equal to~$\frac1N$.
Hence the above is equal to
\begin{align*}
\frac {r_{\max}} N \E_{x',\tau} \frac1{\tau(x')} \sum_{z' \in \tau^{-1}(\{\tau(x')+1,\dots,N\})} \sum_{\pi'} q_{\pi',z'}
% \\ &\leq \frac {r_{\max}} N \E_{x',\tau} \frac1{\tau(x')} \sum_{z'=1}^N \sum_{\pi'} q_{\pi',z'} =
\leq \frac {r_{\max}} N \E_{x',\tau} \frac1{\tau(x')}
= \frac {r_{\max}} N \E_{x'} \frac1{x'}
\leq \parens*{ \ln(N) + 1 } \frac {r_{\max}} {N^2}.
\end{align*}
\end{enumerate}
\end{proof}

%-----------------------------------------------------------------------------
\subsection{Sparsity Analysis}\label{sec:bounc_active}
%-----------------------------------------------------------------------------
The goal of this section is to upper bound the term
\begin{align}\label{eq:sparsity goal}
  \E_{\substack{x \leftarrow [N], \\ \sigma, \tau \leftarrow S_N}} \frac {\norm*{ \parens[\big]{ I - \proj{+_x}_{D_x} } \ket{\phi^{\sigma,\tau,(j)}} }^2}x,
\end{align}
which remains to be estimated in the right-hand side of \cref{prop:hard-database}.
Intuitively, this quantifies the extent to which a random database register~$D_x$ has been queried by the algorithm, weighted by~$1/x$-.

To analyze \cref{eq:sparsity goal}, recall that~$\ket{\phi^{\sigma,\tau,(j)}}_{AXYD}$ denotes the joint state of the algorithm and database right before the $j$-th query when run with the twirled oracle.
By \cref{lem:output state twisted vs not}, we have
\begin{align*}
  \ket{\phi^{\sigma,\tau,(j)}}_{AXYD} = L^\tau_D R^\sigma_D \ket{\phi^{(j)}}_{AXYD},
\end{align*}
where $\ket{\phi^{(j)}}$ denotes the state right before the $j$-th query when the same algorithm is run with the \emph{untwirled} oracle.
We can thus express \cref{eq:sparsity goal} as follows:
\begin{align}\label{eq:gamma intro}
   \E_{\substack{x \leftarrow [N], \\ \sigma, \tau \leftarrow S_N}}
   \frac {\norm*{ \parens[\big]{ I - \proj{+_x}_{D_x} } \ket{\phi^{\sigma,\tau,(j)}} }^2}x
&= \bra{\phi^{(j)}} \Gamma_D \ket{\phi^{(j)}},
\end{align}
where we have introduced the operator
\begin{align}\label{eq:def gamma}
    \Gamma_D
&:= \E_{x \leftarrow [N]} \frac1x \E_{\sigma, \tau \leftarrow S_N} (L^\tau_D R^\sigma_D)^\dagger \parens[\big]{ I - \proj{+_x}_{D_x} } (L^\tau_D R^\sigma_D)
% &= \frac {H_N} N \cdot I - \E_{x \leftarrow [N]} \frac1x \parens*{ \E_{\sigma, \tau \leftarrow S_N} (L^\tau_D R^\sigma_D)^\dagger \proj{+_x}_{D_x} (L^\tau_D R^\sigma_D) }
\end{align}
where we recall that $H_N$ denotes the harmonic numbers, see \cref{eq:H_N}.

To upper bound the quantity of interest, we now observe that that we can upper bound its growth with each additional query as follows, in terms of the norm of a commutator:
\begin{align}
\nonumber
  \braket{\phi^{(j+1)} | \Gamma_D | \phi^{(j+1)}}
- \braket{\phi^{(j)} | \Gamma_D | \phi^{(j)}}
&= \braket{\phi^{(j)} | Q_{XYD}^\dagger \Gamma_D Q_{XYD} - \Gamma_D | \phi^{(j)}} \\
\nonumber
&= \braket{\phi^{(j)} | Q_{XYD}^\dagger [\Gamma_D, Q_{XYD}] | \phi^{(j)}} \\
&\leq \norm*{ [\Gamma_D, Q_{XYD}] }
\label{eq:gamma induction}
\end{align}
where $Q_{XYD} \in \{ \OSPO_{XYD}, \OinvSPO_{XYD} \}$, depending on whether the $j$-th query is a forward or an inverse query.
The first equality holds because the unitary that the algorithm performs inbetween the two queries does not act on the oracle's database register~$D$.

We now calculate the operator $\Gamma_D$ explicitly and use the result to estimate the norm of the commutator.

\begin{lem}\label{lem:gamma exact}
We have
\begin{align*}
  \Gamma_D
= \frac {H_N - H_N^{(2)}} N I_D
- 2 \frac {H_N^{(2)} - H_N^{(3)}} N W^{(2)}_D
- \frac{ H_N - 3 H_N^{(2)} + 2 H_N^{(3)} } N W^{(3)}_D,
\end{align*}
where we denote $H_N^{(\ell)} := \sum_{\ell=1}^N \frac 1 {x^\ell}$ and $W^{(\ell)} := \E_{\gamma \text{ $\ell$-cycle}} R^\gamma = \E_{\gamma \text{ $\ell$-cycle}} L^\gamma$.%
\footnote{To see this, note that $\E_\gamma R^\gamma \ket\pi = \E_\gamma \ket{\pi\gamma^{-1}} = \E_\gamma \ket{\pi\gamma^{-1}\pi^{-1}\pi} = \E_\gamma \ket{\gamma\pi} = \E_\gamma L_\gamma \ket{\pi}$, since if~$\gamma$ is a uniformly random $\ell$-cycle then so is $\pi\gamma^{-1}\pi^{-1}$, for any permutation $\pi\in S_N$.}
\end{lem}
\begin{proof}
We first compute the action of $\proj{+_x}_{D_x}$ in the permutation basis.
For any $\pi\in S_N$, we have
\begin{align*}
    \proj{+_x}_{D_x} \ket\pi_D
&= \E_{s \leftarrow [x]} \ket{\pi_{>x} \tp x s \pi_{<x}} \\
&= \E_{s \leftarrow [x]} \ket{\pi_{>x} \tp x {t_x} \pi_{<x} \, \pi_{<x}^{-1} \tp x {t_x} \pi_{<x} \, \pi_{<x}^{-1} \tp x s \pi_{<x}} \\
&= \E_{s \leftarrow [x]} \ket{\pi \tp x {\pi_{<x}^{-1}(t_x)} \tp x {\pi_{<x}^{-1}(s)}} \\
&= \E_{s \leftarrow [x]} R^{\tp x s}_D R^{\tp x {\pi_{<x}^{-1}(t_x(\pi))}}_D \ket\pi_D,
\end{align*}
where we denote by $t_x \in [x]$ the number in the decomposition~\eqref{eq:tower-decomposition} of $\pi$, i.e., $\pi = \pi_{>x} \tp x {t_x} \pi_{<x}$; in the last line we write $t_x(\pi)$ to make the dependence on $\pi$ explicit.
Thus,
\begin{align*}
  \proj{+_x}_{D_x}
= \sum_{\pi\in S_N} \E_{s \leftarrow [x]} R^{\tp x s}_D R^{\tp x {\pi_{<x}^{-1}(t_x(\pi))}}_D \ket\pi_D \bra\pi_D.
\end{align*}
We first average this over the left action, which commutes with the right action, and obtain
\begin{align*}
  \E_{\tau \leftarrow S_N} (L^\tau_D)^\dagger \proj{+_x}_{D_x} L^\tau_D
&= \sum_{\pi\in S_N} \E_{s \leftarrow [x]} R^{\tp x s}_D R^{\tp x {\pi_{<x}^{-1}(t_x(\pi))}}_D \underbrace{\E_{\tau \leftarrow S_N} \ket{\tau^{-1}\pi}_D \bra{\tau^{-1}\pi}_D}_{= \frac1{N!} I_D} \\
&= \E_{\pi\in S_N} \E_{s \leftarrow [x]} R^{\tp x s}_D R^{\tp x {\pi_{<x}^{-1}(t_x(\pi))}}_D
= \E_{s, t \leftarrow [x]} R^{\tp x s}_D R^{\tp x t}_D
= \E_{s, t \leftarrow [x]} R^{\tp x s \tp x t}_D.
% &= \frac1x I_D + \frac1{x^2} \sum_{s \neq t \in [x]} R^{\tp x s}_D R^{\tp x t}_D \\
% &= \frac1x I_D + \frac2{x^2} \sum_{s \in [x-1]} R^{\tp x s}_D + \frac1{x^2} \sum_{s \neq t \in [x-1]} R^{\tp x s}_D R^{\tp x t}_D \\
% &= \frac1x I_D + \frac{2(x-1)}{x^2} \E_{s \leftarrow [x-1]} R^{\tp x s}_D + \frac{(x-1)(x-2)}{x^2} \E_{s \neq t \leftarrow [x-1]} R^{\tp x s}_D R^{\tp x t}_D.
\end{align*}
If we now average over the right action, the permutation $\tp x s \tp x t$ is conjugated into a random permutation of the same type (either the identity, a transposition, or a 3-cycle, depending on the cardinality of~$\{x,s,t\}$):
\begin{align*}
  \E_{\sigma, \tau \leftarrow S_N} (L^\tau_D R^\sigma_D)^\dagger \proj{+_x}_{D_x} (L^\tau_D R^\sigma_D)
% &= \E_{s, t \leftarrow [x]} \E_{\sigma \leftarrow S_N} R^{\sigma^{-1}}_D R^{\tp x s}_D R^{\tp x t}_D R^\sigma_D \\
&= \E_{s, t \leftarrow [x]} \E_{\sigma \leftarrow S_N} R^{\tp {\sigma^{-1}(x)} {\sigma^{-1}(s)} \tp {\sigma^{-1}(x)} {\sigma^{-1}(t)}}_D \\
&= \sum_{\ell=1}^3 \Pr_{s,t \leftarrow [x]}\bigl( \abs{\{s,t,x\}} = \ell \bigr) \, W^{(\ell)}_D \\
&= \frac1x I_D + \frac{2(x-1)}{x^2} W^{(2)}_D + \frac{(x-1)(x-2)}{x^2} W^{(3)}_D.
\end{align*}
and finally, using \cref{eq:def gamma} and $\E_{x \leftarrow [N]} \frac1{x^\ell} = H_N^{(\ell)} / N$,
\begin{align*}
  \Gamma_D
&= \E_{x \leftarrow [N]} \frac1x \parens*{ I_D - \E_{\sigma, \tau \leftarrow S_N} (L^\tau_D R^\sigma_D)^\dagger \proj{+_x}_{D_x} (L^\tau_D R^\sigma_D) } \\
&= \E_{x \leftarrow [N]} \parens*{ \parens*{ \frac1x - \frac1{x^2} } I_D - \frac{2(x-1)}{x^3} W^{(2)}_D - \frac{(x-1)(x-2)}{x^3} W^{(3)}_D }.
\qedhere
\end{align*}
% \begin{align*}
%     \E_{x \leftarrow [N]} \frac1{x^2} = \frac {H_N^{(2)}} N, \\
%     \E_{x \leftarrow [N]} \frac{2(x-1)}{x^3} = \E_{x \leftarrow [N]} \frac2{x^2} - \E_{x \leftarrow [N]} \frac2{x^3} =
%     2 \frac {H_N^{(2)} - H_N^{(3)}} N, \\
%     \E_{x \leftarrow [N]} \frac{(x-1)(x-2)}{x^3} = \E_{x \leftarrow [N]} \frac{x^2 - 3x + 2}{x^3} = \frac {H_N - 3 H_N^{(2)} + 2 H_N^{(3)}} N
% \end{align*}
\end{proof}

\begin{lem}\label{lem:gamma growth}
For $Q_{XYD} \in \{ \OSPO_{XYD}, \OinvSPO_{XYD} \}$, we have $\norm*{ [\Gamma_D, Q_{XYD}] } \leq \frac {6(\ln(N)+1)} {N^2}$.
\end{lem}
\begin{proof}
We prove the bound in the case that $Q_{XYD} \in \{ \OSPO_{XYD}$ -- the other case is identical except for using the formula $W^{(\ell)}$ in terms of the left instead of the right action.
We first observe that since $\OSPO_{XYD}$ is controlled on~$X$,
\begin{align*}
    \norm*{ [\Gamma_D, \OSPO_{XYD}] }
= \max_{x\in[N]} \norm*{ [\Gamma_D, O^{\SPO,x}_{YD}] }
\end{align*}
where $O^{\SPO,z}_{YD} := \bra z_X  \OSPO_{XYD} \ket z_X$.
Next, note that if $\gamma \in S_N$ is any permutation such that $\gamma(x) = x$, then
\begin{align*}
    [R^\gamma_D, O^{\SPO,x}_{YD}] = 0.
\end{align*}
since for any $y\in[N]$ and $\pi\in S_N$ we have
$O^{\SPO,x}_{YD} R^\gamma_D \ket{y,\pi}_{YD}
% = O^{\SPO,x}_{YD} \ket{y,\pi\gamma^{-1}}_{YD}
= \ket{y \op \pi(\gamma^{-1}(x)),\pi\gamma^{-1}}_{YD}
= \ket{y \op \pi(x),\pi\gamma^{-1}}_{YD}
= R^\gamma_D O^{\SPO,x}_{YD} \ket{y,\pi}_{YD}$.
Otherwise, if $\gamma(x) \neq x$ then it still holds that $\norm{R^\gamma_D, O^{\SPO,x}_{YD}]} \leq 2$ since the commutator of any two unitaries has operator norm at most two.
Accordingly,
\begin{align*}
  \norm*{ [W^{(2)}_D, O^{\SPO,x}_{YD}] }
\leq 2 \Pr_{\gamma \text{ 2-cycle}}(\gamma(x) \neq x)
= \frac 4 N, \\
  \norm*{ [W^{(3)}_D, O^{\SPO,x}_{YD}] }
\leq 2 \Pr_{\gamma \text{ 3-cycle}}(\gamma(x) \neq x)
= \frac 6 N.
\end{align*}
and hence, using \cref{lem:gamma exact}, the fact that $H_N \geq H_N^{(2)} \geq H_N^{(3)} \geq 0$, and \cref{eq:harmonic bound},
\begin{align*}
  \norm*{ [\Gamma_D, O^{\SPO,x}_{YD}] }
&\leq 2 \frac {H_N^{(2)} - H_N^{(3)}} N \norm*{ [W^{(2)}_D, O^{\SPO,x}_{YD}] }
+ \frac{ H_N - 3 H_N^{(2)} + 2 H_N^{(3)} } N \norm*{ [W^{(3)}_D, O^{\SPO,x}_{YD}] } \\
&\leq 8 \frac {H_N^{(2)} - H_N^{(3)}} {N^2} + 6 \frac{ H_N - 3 H_N^{(2)} + 2 H_N^{(3)} } {N^2} \\
&= \frac {6 H_N - 10 H_N^{(2)} + 4 H_N^{(3)} } {N^2}
\leq \frac {6 H_N} {N^2}
\leq \frac {6(\ln(N) + 1)} {N^2}.
\end{align*}
\end{proof}

\begin{cor}\label{cor:weighted-sparsity}
For all $j$, it holds that
\begin{align*}
  \E_{\substack{x \leftarrow [N], \\ \sigma, \tau \leftarrow S_N}} \frac {\norm*{ \parens[\big]{ I - \proj{+_x}_{D_x} } \ket{\phi^{\sigma,\tau,(j)}} }^2}x
\leq 6 \, \frac {j (\ln(N)+1)} {N^2}
\end{align*}
and hence
\begin{align*}
    4 q r_{\max} \sum_{j=1}^{2q} \E_{\substack{x \leftarrow [N], \\ \sigma, \tau \leftarrow S_N}} \frac {\norm*{ \parens[\big]{ I - \proj{+_x}_{D_x} } \ket{\phi^{\sigma,\tau,(j)}} }^2}x
% &\leq 4 q r_{\max} \frac {6 (\ln(N)+1)} {N^2} \sum_{j=1}^{2q} j \\
% &= 4 q r_{\max} \frac {6 (\ln(N)+1)} {N^2} \frac {2q(2q+1)}2 \\
&\leq 72 \frac {q^3 r_{\max} (\ln(N)+1)} {N^2}
\end{align*}
\end{cor}
\begin{proof}
The first formula follows from \cref{eq:gamma intro,eq:gamma induction,lem:gamma growth} by using induction, since $\bra{\phi^{(0)}} \Gamma_D \ket{\phi^{(0)}} = \bra{\Phi_{\SPO}} \Gamma_D \ket{\Phi_{\SPO}} = 0$.
The second formula follows at once.
\end{proof}

%-----------------------------------------------------------------------------
\subsection{Main Theorem}\label{sec:main_thm}
%-----------------------------------------------------------------------------
Finally, we can use the preceding analysis, together with the fundamental lemma, to establish our main theorem (which formalizes \cref{thm:intro main} announced in the introduction):

\begin{theorem}[Search]\label{thm:main}
Let $\mathcal A$ be a quantum algorithm with quantum query access to a random permutation~$\pi\in S_N$ and its inverse (\cref{defn:standard random perm}), which returns an $x\in[N]$, and let $R \subseteq [N] \times [N]$ be any relation.
If~$\mathcal A$ makes fewer than~$q$ queries, then the probability that it returns an element $x$ such that $(x,\pi(x)) \in R$ is
\begin{align*}
  \Pr_{\pi \leftarrow S_N, \; x \leftarrow \mathcal A^{U^\pi, U^{\pi^{-1}}}} \bigl[ (x, \pi(x)) \in R \bigr]
\leq 914 \, \frac {q^3 r_{\max} \parens[\big]{ \ln(N) + 2 }} N,
\end{align*}
where we recall $r_{\max} = \max \braces*{ \max_x \, \abs{R_x}, \max_y \, \abs{R^\mathrm{inv}_y} }$, with $R_x = \{ y : (x,y) \in R \}$ and~$R^\mathrm{inv}_y = \{ x : (x,y) \in R \}$.
\end{theorem}
\begin{proof}
Without loss of generality we can assume that $\mathcal A$ is a unitary query algorithm on registers~$AXY$, where~$X$ and~$Y$ are the two $N$-dimensional registers that the oracles get applied to such that the classical outcome~$x$ can be obtained by measuring the~$X$ register.
Let $\mathcal B$ denote the unitary query algorithm that first runs $x \leftarrow \mathcal A$, then makes one more query to load~$\pi(x)$ into the~$Y$ register.
Since $\mathcal A$ makes fewer than $q$ queries, the algorithm $\mathcal B$ makes at most~$q$ queries.
For every~$\sigma,\tau\in S_N$, let~$\ket{\phi^{\sigma,\tau}}_{AXYD}$ be the joint state of algorithm and oracle defined by running~$\Init^\SPO_D$ and then~$\mathcal B^{\FTSPOsigmatau_D}$.
Then, \cref{prop:hard-database,cor:weighted-sparsity} combine to
\begin{align*}
    \E_{\substack{x \leftarrow [N], \\ \sigma, \tau \leftarrow S_N}} \norm*{ E^{R^{\sigma,\tau},x}_D \ket{\phi^{\sigma,\tau}} }^2
&\leq 384 \frac {q^2 r_{\max} \parens[\big]{ \ln(N) + 2 }} {N^2} + 72 \frac {q^3 r_{\max} (\ln(N)+1)} {N^2} \\
&\leq 456 \frac {q^3 r_{\max} \parens[\big]{ \ln(N) + 2 }} {N^2}.
\end{align*}
Using \cref{lem:fund lem rhs upper bound,lem:progress vs fundamental}, we can upper bound the quantity $p_\text{(ii)}$ defined in the fundamental lemma (\cref{lem:fundamental}) as follows:
\begin{align*}
  p_\text{(ii)}
\leq N \E_{\substack{x \leftarrow [N], \\ \sigma, \tau \leftarrow S_N}} \norm*{ E^{R^{\sigma,\tau},x}_D \ket{\phi^{\sigma,\tau}} }^2
\leq 456 \frac {q^3 r_{\max} \parens[\big]{ \ln(N) + 2 }} N.
\end{align*}
Finally, the fundamental lemma states that
\begin{align*}
  \sqrt{p_\text{(i)}} \leq \sqrt{p_\text{(ii)}} + \sqrt{\frac{\ln(N) + 1}N}
\end{align*}
and hence 
\begin{align*}
  p_\text{(i)}
\leq 2 \parens*{ p_\text{(ii)} + \frac{\ln(N) + 1}N }
\leq 914 \, \frac {q^3 r_{\max} \parens[\big]{ \ln(N) + 2 }} N
\end{align*}
concluding our proof.
\end{proof}

% !TEX root = compressed-pi.tex
%=============================================================================
\section{Application to One-Round Sponge and Unruh's Conjecture}\label{sec:sponge}
%=============================================================================

We can now apply our results to obtain bounds for the hardness of search problems for algorithms given quantum query access to a random permutation and its inverse.

We first show a bound on the pre-image search problem for the \emph{sponge construction}, instantiated with a random permutation, restricted to one absorption round and one squeezing round.
The sponge function in this special case, for a permutation $\pi\in S_{\{0,1\}^n}$ and with capacity $c<n$, is given by
\begin{equation*}
	f_{\pi} \colon \{0,1\}^{n-c}\to\{0,1\}^{n-c}, \quad
	f_{\pi}(x)=\pi\mleft( x\|0^c \mright)_{[1,n-c]},
\end{equation*}
where $s_{[1,r]}$ denotes the first $r$ bits of a string $s$.
The following result was stated as \cref{cor:intro sponge} in the introduction.

\begin{cor}[One-round sponge]\label{cor:precise sponge}
For any $y\in\{0,1\}^{n-c}$, the probability that a quantum algorithm $\mathcal A$ with quantum query access to a random permutation~$\pi\in S_{\{0,1\}^n}$ and its inverse returns a preimage $x\in\{0,1\}^{n-c}$ under the one-round sponge function~$f_\pi$, by making fewer than $q$ queries, can be upper bounded as
\begin{align*}
	\Pr_{\pi \leftarrow S_{\{0,1\}^n}, \; x \leftarrow \mathcal A^{U^\pi, U^{\pi^{-1}}}} \bigl[ f_\pi(x) = y \bigr]
\leq 914 \, \frac {q^3 (n+2)} {2^{\min(c,n-c)}}.
\end{align*}
\end{cor}
\begin{proof}
Let $\mathcal B$ denote the algorithm that runs $x \leftarrow \mathcal A$ and returns $x' := x \| 0^c$.
Clearly, $f_\pi(x) = y$ if and only if~$(x',\pi(x')) \in R$, where
\begin{align*}
	R = \braces*{ (x',y') \in \{0,1\}^n \times \{0,1\}^n \;:\; x'_{[n-c+1,n]} = 0^c, \ y'_{[1,n-c]} = y }.
\end{align*}
Note that
\begin{align*}
	r_{\max}
= \max \braces*{ \max_{x'} \, \abs{R_{x'}}, \max_{y'} \, \abs{R^\mathrm{inv}_{y'}} }
= 2^{\max(c,n-c)}.
\end{align*}
Applying \cref{thm:main} to $\mathcal B$, which makes the same number of queries as $\mathcal A$, we find that
%\gm{Is it possible that $\ln(N)$ became a $\log_2(N)$?}
\begin{equation*}
	\!\!\Pr_{\substack{\pi \leftarrow S_{\{0,1\}^n}, \\ x \leftarrow \mathcal A^{U^\pi, U^{\pi^{-1}}}}}\!\! \bigl[ f_\pi(x) = y \bigr]
= \!\!\Pr_{\substack{\pi \leftarrow S_{\{0,1\}^n}, \\ x' \leftarrow \mathcal B^{U^\pi, U^{\pi^{-1}}}}}\!\! \bigl[ (x',\pi(x')) \in R \bigr]
\leq
914 \, \frac {q^3 2^{\max(c,n-c)} (n+2)} {2^n}
=
914 \, \frac {q^3 (n+2)} {2^{\min(c,n-c)}}.
\qedhere
\end{equation*}
\end{proof}

Next we consider the \emph{double-sided zero-search conjecture}, which states that no adversary making polynomially many quantum queries to a permutation $\pi\in S_{\{0,1\}^{2n}}$ and its inverse is able to find~$x\in\{0,1\}^n$ such that $\pi(x\|0^n)_{[n+1,2n]} = 0^n$ with non-negligible probability~\cite[Conjecture~1]{unruh2023towards}.
The following corollary confirms Unruh's conjecture and establishes more generally an upper bound on the success probability for the problem with an arbitrary number~$c$ of zeros.
It was stated as \cref{cor:intro zero search} in the introduction.

\begin{cor}\label{cor:precise zero search}
The probability that a quantum algorithm $\mathcal A$ with quantum query access to a random permutation~$\pi\in S_{\{0,1\}^{2n}}$ and its inverse returns $x\in\{0,1\}^n$ such that $\pi(x\|0^n)_{[n+1,2n]} = 0^n$, by making fewer than $q$ queries, can be upper bounded as
\begin{align*}
\Pr_{\pi \leftarrow S_{\{0,1\}^{2n}}, \; x \leftarrow \mathcal A^{U^\pi, U^{\pi^{-1}}}} \bigl[ \exists y \in \{0,1\}^n : \pi(x\|0^n) = y \| 0^n \bigr]
\leq 1828 \, \frac {q^3 \parens[\big]{ n+1 }} {2^n}.
\end{align*}
More generally, it holds for any $c\in[2n]$ and any algorithm $\mathcal A$ that returns bitstrings $x\in\{0,1\}^{2n-c}$ that
\begin{align*}
\Pr_{\pi \leftarrow S_{\{0,1\}^{2n}}, \; x \leftarrow \mathcal A^{U^\pi, U^{\pi^{-1}}}} \bigl[ \exists y \in \{0,1\}^{2n-c} : \pi(x\|0^c) = y \| 0^c \bigr]
\leq 1828 \, \frac {q^3 \parens[\big]{ n+1 }} {2^c}.
\end{align*}
\end{cor}
\begin{proof}
It suffices to establish the second claim since it implies the first for~$c=n$.
Similarly to the proof of \cref{cor:precise sponge} we can reduce to a relation, namely
\begin{align*}
	R = \braces*{ (x',y') \in \{0,1\}^{2n} \times \{0,1\}^{2n} \;:\; x'_{[2n-c+1,2n]} = y'_{[2n-c+1,2n]} = 0^c }.
\end{align*}
Note that $r_{\max} = 2^{2n-c}$.
Thus, \cref{thm:main} yields, with $N=2^{2n}$,
\begin{equation*}
\Pr_{\pi \leftarrow S_{\{0,1\}^{2n}}, \; x \leftarrow \mathcal A^{U^\pi, U^{\pi^{-1}}}} \bigl[ \exists y \in \{0,1\}^{2n-c} : \pi(x\|0^c) = y \| 0^c \bigr]
\leq 914 \, \frac {q^3 2^{2n-c} \parens[\big]{ 2n + 2 }} {2^{2n}}
= 1828 \, \frac {q^3 \parens[\big]{ n+1 }} {2^c}.
\qedhere
\end{equation*}
\end{proof}

% OLD. I THINK IT 2^c SHOULD BE 2^{n-c} IN THE NUMERATOR
%
% The following corollary establishes a bound on the success probability on such a problem.
% \begin{cor}
% 	Let $\mathcal A^{U^\pi, U^{\pi^{-1}}}$ be an algorithm making  $q$ quantum queries to a random permutation $\pi\in S_{2^n}$ and its inverse. Then
% 	\begin{equation*}
% 		\Pr_{x\leftarrow \mathcal A^{U^\pi, U^{\pi^{-1}}}}[\exists y\in\{0,1\}^r:\pi(x\|0^c)=y\|0^c]
% \le 384 \frac {q^2 2^c \parens[\big]{ n + 2 }} {2^n} +
% 192 \frac{q^3 2^c n}{2^n}+\sqrt{\frac{n + 1}{2^n}}.
%  \end{equation*}
% \end{cor}
% \begin{proof}
% 	Applying \cref{thm:main} with
% 	\begin{equation*}
% 		R=\left\{(x',y')\in\{0,1\}^{2n}\Big|x'_{[r+1,n]}=0^c, y'_{[r+1,n]} =0^c\right\}
% 	\end{equation*}
% 	immediately yields the inequality.
% \end{proof}

\section*{Acknowledgments}
% \addcontentsline{toc}{section}{Acknowledgments}
The authors thank Saliha Tokat for feedback on an earlier version of this manuscript.

GM and MW acknowledge support by the Deutsche Forschungsgemeinschaft (DFG, German Research Foundation) under Germany's Excellence Strategy - EXC\ 2092\ CASA - 390781972.
GM also acknowledges support by the European Research Council through an ERC Starting Grant (grant agreement No.~101077455, ObfusQation).
MW also acknowledges support by the European Research Council through an ERC Starting Grant (grant agreement No.~101040907, SYMOPTIC), by the NWO through grant OCENW.KLEIN.267, and by the BMBF through project Quantum Methods and Benchmarks for Resource Allocation (QuBRA).
CM acknowledges support by the Independent Research Fund Denmark via a DFF Sapere Aude grant (IM-3PQC, grant ID 10.46540/2064-00034B).
MW acknowledges the Simons Institute for the Theory of Computing at UC Berkeley for its hospitality and support.
MW and CM thank the Leibnitz Center for Informatics for its hospitality.

\bibliographystyle{alphaurl}
% \addcontentsline{toc}{section}{References}
\bibliography{bib}

% \appendix
%\input{old/multi-input-predicates}
% \input{9-appendix}

\end{document}